%% file: main.tex
\documentclass[11pt]{article}
\usepackage{graphicx}
\usepackage{latexsym}
\usepackage{amssymb}
\usepackage{amsmath}
\usepackage{amsthm}
\usepackage{thmtools}
\usepackage{forloop}
\usepackage[paper=letterpaper,margin=1in]{geometry}
\usepackage{multirow}
\usepackage[ruled,vlined,linesnumbered]{algorithm2e}
\usepackage{paralist}
\usepackage{mleftright} \mleftright
\usepackage{xcolor}
\usepackage{scalerel}
\usepackage{nicefrac}
\usepackage{threeparttable}
\usepackage{hyperref}
\usepackage[numbers]{natbib}

\newtheorem{theorem}{Theorem}[section]
\newtheorem{lemma}[theorem]{Lemma}

\newtheorem{corollary}[theorem]{Corollary}
\newtheorem{definition}{Definition}[section]
\newtheorem*{remark}{Remark}
\newtheorem{proposition}[theorem]{Proposition}
\newtheorem{observation}[theorem]{Observation}

\newtheorem{assumption}[theorem]{Assumption}

\newcommand{\defcal}[1]{\expandafter\newcommand\csname c#1\endcsname{{\mathcal{#1}}}}
\newcommand{\defbb}[1]{\expandafter\newcommand\csname b#1\endcsname{{\mathbb{#1}}}}
\newcommand{\defvec}[1]{\expandafter\newcommand\csname v#1\endcsname{{\mathbf{#1}}}}
\newcounter{calBbCounter}
\forLoop{1}{26}{calBbCounter}{
    \edef\letter{\alph{calBbCounter}}
		\edef\Letter{\Alph{calBbCounter}}
    \expandafter\defcal\Letter
		\expandafter\defbb\Letter
		\expandafter\defvec\letter
}

\newcommand{\eps}{\varepsilon}
\newcommand{\nnR}{{\bR_{\geq 0}}}
\newcommand{\psum}{\oplus}
\newcommand{\vzero}{{\mathbf{0}}}
\newcommand{\hprod}{\odot}
\newcommand{\inner}[2]{{\left\langle#1,#2\right\rangle}}
\newcommand{\vone}{{\mathbf{1}}}
\newcommand{\FF}[1]{F\left( #1 \right)}
\newcommand{\ff}[1]{f\left( #1 \right)}

\newcommand{\ex}[1]{ \mathbb{E}\left[ #1 \right] }

\newcommand{\opt}[0]{\textnormal{\textrm{OPT}}}
\newcommand{\optvalue}[0]{\textnormal{\textrm{opt}}}
\newcommand{\alg}[0]{\textnormal{\textrm{ALG}}}
\newcommand{\characteristic}{{\mathbf{1}}}
\newcommand{\RSet}{{\texttt{R}}}
\newcommand{\cardBudget}{B}
\newcommand{\marge}[2]{f \left( #1 \mid #2 \right) }
\newcommand{\margen}[3]{#1 ( #2 \mid #3 ) }

\newcommand{\densitygreedy}{\textsc{Density-Based Greedy}}
\newcommand{\diff}[1]{\textnormal{d} #1 }
\newcommand{\email}[1]{{\href{mailto:#1}{#1}}}

\DeclareMathOperator*{\PSum}{\scalerel*{\oplus}{\sum}}
\DeclareMathOperator*{\argmax}{arg\,max}
\DeclareMathOperator*{\rank}{rank}

\makeatletter
\providecommand*{\cupdot}{%
  \mathbin{%
    \mathpalette\@cupdot{}%
  }%
}
\newcommand*{\@cupdot}[2]{%
  \ooalign{%
    $\m@th#1\cup$\cr
    \hidewidth$\m@th#1\cdot$\hidewidth
  }%
}
\makeatother

\newtoggle{SODAVersion}
\togglefalse{SODAVersion}

\title{Bicriteria Submodular Maximization}

\iftoggle{SODAVersion}
{\author{}}
{
\author{Moran Feldman\thanks{Department of Computer Science, University of Haifa. E-mail: \email{moranfe@cs.haifa.ac.il}}
			  \and
				Alan Kuhnle\thanks{Department of Computer Science \& Engineering, Texas A\&M University. E-mail: \email{kuhnle@tamu.edu}}}
}

\begin{document}

\maketitle

\input{Abstract}
\thispagestyle{empty}
\pagenumbering{Alph}
\newpage
\pagenumbering{arabic}
\setcounter{page}{1}

\input{Introduction}
\input{OurResults}
\input{Preliminaries}
\input{Warmup}
\input{MonotoneFunctions}

\input{GeneralFunctions}
\input{SymmetricFunctions}

\appendix

\input{PipageRounding}
\input{ImprovedAnalysisUnified}

\input{KnapsackEquality}
\bibliographystyle{plainnat}
\bibliography{bicriteria,bicriteria-alan}

\end{document}

%% file: Abstract.tex
\begin{abstract}
  Submodular functions and their optimization have found applications in diverse settings ranging from machine learning and
  data mining to game theory and economics. 
  In this work, we consider the constrained maximization of a submodular function,
  for which we conduct a principled study of \textit{bicriteria approximation algorithms}---algorithms
  which can violate the constraint, but only up to a bounded factor.
  Bicrteria optimization allows constrained submodular maximization to capture additional important settings, such as the well-studied
  submodular cover problem and optimization under soft constraints.
  We provide results that span both multiple types of constraints (cardinality, knapsack, matroid and convex set)
  and multiple classes of submodular functions (monotone, symmetric and general).
  For many of the cases considered, we provide optimal results. In other cases, our results improve over the state-of-the-art, sometimes
  even over the state-of-the-art for the special case of single-criterion (standard) optimization.
  Results of the last kind demonstrate that relaxing the feasibility constraint may give a perspective about the problem
  that is useful even if one only desires feasible solutions.
 
\medskip

\noindent \textbf{keywords:} submodular function, bicriteria optimization, cardinality constraint, knapsack constraint, matroid constraint, monotone function, symmetric function

\end{abstract}


%% file: Introduction.tex
\section{Introduction} \label{sec:introduction}

Submodular set functions and their
optimization have found applications in diverse settings,
such as machine learning and data mining~\cite{das2011submodular,elenberg2017streaming,krause2010submodular,lin2011class},
economics~\cite{dobzinski2013communication,feige2010submodular}, game theory~\cite{eden2024combinatorial,feldman2015combinatorial}, networks and graphs~\cite{mossel2010submodularity,norouzifard2018beyond}.
The usefulness of submodular functions stems from the observation that submodularity captures the intuitive notion of diminishing returns.
Formally, let $\cN$ be a finite ground set, and let
$f\colon 2^{\cN} \to \bR$ be a set function.
The function $f$ is \textit{submodular}\footnote{An equivalent alternative definition of submodularity is that $f$ is submodular if the inequality $f(S) + f(T) \geq f(S \cup T) + f(S \cap T)$ holds for every $S, T \subseteq \cN$.}
if for all $S \subseteq T \subseteq \cN$
and $u \in \cN \setminus T$,
\[f(S \cup \{ u \}) - f(S) \ge f(T \cup \{ u \} ) - f(T). \]
Additionally, the function $f$ is \textit{monotone} if
for all $S \subseteq T \subseteq \cN$,
$f(S) \le f(T)$.

In this work, we study the \emph{bicriteria optimization} of submodular set functions
subject to various, standard
constraint classes: cardinality constraint, knapsack constraint, matroid constraint, and convex set constraint.
In this context, a bicriteria algorithm may return an infeasible set, where the infeasibility is bounded.
To make this more formal, let $\cI$ be the collection of the sets that are feasible according to the relevant constraint. Then, we consider the problem
$\max_{S \in \mathcal I} f(S)$.
We denote by $\opt$ an optimal solution (i.e., a set in $\mathcal I$
satisfying $f( \opt ) = \max_{S \in \mathcal I} f(S)$),
and let $\cP \subseteq [0,1]^{\cN}$ denote the convex hull of the characteristic
vectors of the feasible sets in $\cI$.\footnote{The characteristic vector of a set $S \subseteq \cN$, denoted by $\characteristic_S$, is a vector in $[0, 1]^\cN$ that takes the value $1$ in the coordinates corresponding to the elements of $S$, and the value $0$ in the other coordinates.}
An algorithm is a bicriteria approximation algorithm
with \textit{bicriteria ratio} $(\alpha, \beta) \in [0,1] \times [0, \infty)$
if it finds a set $S$ such that $f(S) \ge \alpha \cdot f( \opt )$,
and there exist vectors $\vy_1, \vy_2 \in \cP$ such that
$\vy_1 \le \vone_S \le \beta \vy_2$.

\begin{remark}
  Values of $\beta$ that are smaller than $1$ do not always make sense. However, they are useful, for example, when $\cP$ is down-closed (i.e., when $\vy \in \cP$ implies $\vx \in \cP$ for every vector $\vzero \leq \vx \leq \vy$). Notice that when $\cP$ is down-closed, the condition $\vy_1 \le \vone_S \le \beta \vy_2$ reduces to $\vone_S/\beta \in \cP$, and thus, one can view $\beta < 1$ as requiring the algorithm to output a solution that is ``more'' feasible than what is required by the constraint.
\end{remark}

To motivate the study of bicriteria optimization of submodular functions (under the above definition), the reader is invited to consider
the following settings. 
\begin{itemize}
  \item When $f$ is monotone and submodular,
    the goal in the monotone submodular cover problem
    is to find a set $S$ of minimum cost such that
    $f(S) = \max_{T \subseteq \cN}f( T )$. 
		This problem generalizes the
    set cover problem, and has been studied extensively
    \citep{Crawforda, Hochbaum1982, Iyer2013, Wan2010a, Wolsey1982}.
		Interestingly, this problem is one of the few examples of submodular optimization problems for which bicriteria optimization has been explicitly studied. In particular, \cite{goyal2013minimizing} has shown that the greedy algorithm can produce a set $S$ of value at least $(1 - \eps) \cdot \max_{T \subseteq \cN}f( T )$ using a cost larger than the cost of the optimal solution by at most a factor of $1 + \ln \eps^{-1}$. 

		\citet{Iyer2013} noted that submodular cover problem is related to maximization of submodular functions subject to knapsack constraints, and proved that bicriteria approximation results for each one of these problems translates into bicriteria approximation results for the other. Using their machinery, the above stated result of \cite{goyal2013minimizing} translates into a $(1 - \eps, 1 + \ln \eps^{-1})$-bicriteria approximation for maximizing a monotone submodular function subject to a cardinality or knapsack constraint. Later, \citet{crawford2023scalable}
    described an algorithm
    guaranteeing $(1/2 - \eps, O(\eps^{-2}))$-bicriteria approximation for maximizing general (not necessarily monotone)
    submodular functions subject to a cardinality or knapsack constraint.
    
  \item Many naturally occurring applications of submodular maximization involve soft constraints. One standard approach models such applications as maximization of $f(S) - g(S)$, where $f$ is the objective function and $g$ is a linear function representing the soft constraint \citep{Bodek2022, feldman2021guess, Harshawb, Jin2021, Kazemib, Lu2024}. However, this modeling approach (provably) does not enjoy constant multiplicative approximation guarantees. Furthermore, getting a bound on $f - g$ is sometimes not useful enough as the same value for $f - g$ can be obtained, for example, both by a set of high value that violates the soft constraint by much, and a set of low value that is feasible. Bicriteria optimization provides an alternative way to model soft constraints that is more precise. Specifically, bicriteria results  bound the quality of the approximation that can be guaranteed for the objective function $f$ for any amount of allowed violation for the soft constraint represented by $g$.
	
  \item Constrained maximization of a submodular function has
    been studied intensively \citep{buchbinder2019constrained,buchbinder2024constrained,chekuri2014submodular,ene2016constrained,feldman2011unified}. Despite this scrutiny,
    even in the simplest case where the constraint is only on cardinality,
    the optimal approximation factor is still unknown, although it is known to lie
    in the range $[0.401, 0.478]$ (see~\citep{buchbinder2024constrained,qi2022maximizing}). On the other hand, if $f$ is
    unconstrained, the optimal approximation of $1/2$ has been obtained \citep{buchbinder2015tight,feige2011maximizing}.
    As was shown by \citet{crawford2023scalable} for knapsack constraints, and we show in the following for more general constraints, relaxing the feasibility allows for an approximation factor of roughly $1/2$ in the bicriteria setting, which is necessarily optimal by the hardness for the unconstrained case.
    These results may shed some
    light, at least indirectly, on the basic cause for the hardness of the classical (non-bicriteria) setting. 
\end{itemize}
Despite the relevance of bicriteria submodular optimization,
prior work exists mostly in disparate settings, such
as the (above-discussed) submodular cover problem.
In this work, we initiate a more principled study of such algorithms. 


%% file: OurResults.tex
\subsection{Our Results} \label{ssc:our_results}

In this section, and throughout the paper, when discussing an $(\alpha, \beta)$-bicriteria approximation algorithm, we refer to $\alpha$ and $\beta$ as the approximation ratio and the infeasibility ratio of the algorithm, respectively. The special case of $\beta = 1$ corresponds to standard (feasibly constrained) optimization, and we refer to this case as \emph{single-criterion} optimization. Additionally, we note that our algorithms for maximization subject to a convex set constraint maximize the multilinear extension of the objective set function (see Section~\ref{sec:preliminaries} for the definition of the multilinear extension). Accordingly, our algorithms for convex set constraints output fractional solutions. Nevertheless, the approximation guarantees of these algorithms are often with respect to the optimal integer solution, and for conciseness, we refer to approximation guarantees of this kind as \emph{integer-approximation}.

Below, we discuss our individual results. For convenience, we group these results by the class of submodular objective functions that they apply to: monotone, general or symmetric, and discuss separately each such group of results. Our results are also summarized by Table~\ref{tbl:results}. 

\newcommand{\tableRef}[1]{\hspace{-2mm}{(\ref{#1})}}
\newcommand{\doubleTableRef}[2]{\hspace{-5mm}{(\ref{#1}+\ref{#2})}}
\begin{table}[ht]
\begin{center}
\small
\begin{threeparttable}
\setlength{\tabcolsep}{1mm}
\begin{tabular}{c|lr|lr|lr}
	& \multicolumn{2}{|c|}{\textbf{Monotone Functions}} & \multicolumn{2}{|c|}{\textbf{General Functions}} & \multicolumn{2}{|c}{\textbf{Symmetric Functions}} \\
	\hline
	\multirow{2}{*}{\textbf{Convex Set}} & \multirow{2}{*}{$(1 - \eps, \ln \eps^{-1})$-int\tnote{*}} & \multirow{2}{*}{\tableRef{cor:monotone_general}} & \rule{0pt}{2.2ex}$(\alpha, \beta)$-hard & \tableRef{thm:hardness_non_down-closed} &&\\
	&&& \multicolumn{2}{|l|}{$\forall \alpha \in (0, 1], \beta \geq 1$} &&\\
	\hline
	\multicolumn{1}{p{24mm}|}{\multirow{4}{24mm}{\centering \textbf{Down-Closed Convex Set}}} & \rule{0pt}{2.2ex}$(1 - \eps, \ln \eps^{-1})$-int. & \tableRef{cor:monotone_downclosed} & \multirow{4}{*}{$(\tfrac{1}{2} - \eps, O(\eps^{-1}))$} & \multirow{4}{*}{\tableRef{thm:general_down-closed_convex_body}} & \multirow{4}{*}{$(\tfrac{1}{2} - \eps, \tfrac{1}{2}\ln(\tfrac{1}{2\eps}))$-int} & \multirow{4}{*}{\tableRef{thm:symmetric_polytope}}\\
	& $(1 - \eps + \delta, \ln \eps^{-1})$-hard & \tableRef{obs:monotone_hardnesss} &&&&\\
	& \multicolumn{2}{|l|}{\textit{The hardness is also implied}} &&&&\\
	& \multicolumn{2}{|l|}{\textit{by a construction of \cite{vondrak2013symmetry}.}} &&&&\\
	\hline
	\multirow{5}{*}{\textbf{Cardinality}} & \rule{0pt}{2.2ex}$(1 - \eps, \lceil\ln \eps^{-1}\rceil)$\tnote{$\dagger$} & \tableRef{thm:monotone_knapsack_greedy} & $(\tfrac{1}{2} - \eps, O(\eps^{-1}))$\tnote{$\dagger$} & \tableRef{thm:general_knapsack_matroid} &&\\
	& $(1 - \eps, \rho(c, \eps) + \delta)$\tnote{$\ddagger$} & \tableRef{thm:knapsack_fractional} &&& $(\tfrac{1}{2} - \eps - \delta, \lceil \tfrac{1}{2} \ln (\tfrac{1}{2\eps}) \rceil)$\tnote{$\dagger$} & \tableRef{thm:symmetric_knapsack_greedy}\\
	& $(1 - \eps + \delta, \rho(c, \eps))$-hard & \tableRef{thm:monotone_cardinality_inapproximability} & \multicolumn{2}{|l|}{$(\frac{2(1 - e^{-\beta/2})}{e^{\beta/2}} + \delta, \beta)$-hard} &&\\
	& \multicolumn{2}{|l|}{\textit{The combinatorial result is}}& $\forall \beta \in [0, \ln 2]$ & \tableRef{thm:hardness_general} & $(\tfrac{1}{2} - \eps, \frac{1 - (2\eps)^c}{2c} + \delta)$& \tableRef{thm:symmetric_knapsack_fractional}\\
	& \multicolumn{2}{|l|}{\textit{also implied by \cite{goyal2013minimizing}.}}&&&&\\
	\cline{1-3}\cline{6-7}
	\multirow{5}{*}{\textbf{Knapsack}} & \rule{0pt}{2.2ex}$(1 - \eps, 1 + \ln \eps^{-1})$\tnote{$\dagger$} & \tableRef{thm:monotone_knapsack_greedy} & \multicolumn{2}{|l|}{\textit{$(\tfrac{1}{2} + \delta, \beta)$-hard $\forall \beta \geq 1$}} && \\
	& $(1 - \eps - \delta, \ln \eps^{-1})\tnote{*}$ & \doubleTableRef{cor:monotone_downclosed}{lem:rounding_knapsack} & \multicolumn{2}{|l|}{\textit{follows from \cite{feige2011maximizing}.}} & $(\tfrac{1}{2} - \eps - \delta, 1 + \tfrac{1}{2} \ln (\tfrac{1}{2\eps}))\tnote{$\dagger$} $ & \tableRef{thm:symmetric_knapsack_greedy}\\
	& $(1 - \eps, \rho(c, \eps) + 1)$\tnote{$\ddagger$}  & \tableRef{thm:knapsack_fractional} &&&&\\
	& \multicolumn{2}{|l|}{\textit{The combinatorial result is}}& \multicolumn{2}{|l|}{\textit{For cardinality and}} & $(\tfrac{1}{2} - \eps, \frac{1 - (2\eps)^c}{2c} + 1)$ & \tableRef{thm:symmetric_knapsack_fractional}\\
	& \multicolumn{2}{|l|}{\textit{also implied by \cite{goyal2013minimizing}.}}& \multicolumn{2}{|l|}{\textit{knapsack constraints,}} &&\\
	\cline{1-3}\cline{6-7}
	\multirow{2}{*}{\textbf{Matroid}} & \rule{0pt}{2.2ex}$(1 - \eps, \lceil \log_2 \eps^{-1} \rceil)$\tnote{$\dagger$} & \tableRef{thm:monotone_matroid_greedy} & \multicolumn{2}{|l|}{\textit{\cite{crawford2023scalable}} showed} & \multirow{2}{*}{$(\tfrac{1}{2} - \eps, \lceil \tfrac{1}{2}\ln(\tfrac{1}{2\eps})\rceil)$} & \multirow{2}{*}{\doubleTableRef{thm:symmetric_polytope}{lem:rounding_matroid}}\\
	& $(1 - \eps, \lceil \ln \eps^{-1} \rceil)$ & \doubleTableRef{cor:monotone_downclosed}{lem:rounding_matroid} & \multicolumn{2}{|l|}{$(\nicefrac{1}{2} - \eps, O(\eps^{-2}))$.} &&\\
\end{tabular}
\begin{tablenotes}
	\item[*] Applies only for $\forall \eps \in (0, 1/e]$.
	\item[$\dagger$] Obtained using a simple combinatorial algorithm.
	\item[$\ddagger$] See Section~\ref{ssc:monotone_per_c_analysis} for the definition of $\rho(c, \eps)$.
\end{tablenotes}
\end{threeparttable}
\end{center}
\caption{Summary of our results, and previous results (in italic). The guarantees stated in the table omit, for simplicity, $o(1)$ error terms. Values of $\delta$ in these guarantees represent arbitrarily small constants, and values of $\eps$ are constants that can take any positive value from $0$ until the natural upper limit for the guarantee (unless another upper limit is explicitly stated in the table). Inapproximability results are marked with the word ``hard'', and bicriteria integer-approximation results are marked with the word ``int''. The other results given in the table are bicriteria approximation results. The parenthesis near each result contain the numbers of one or two claims (theorems, corollaries or observations) that either explicitly state the result or imply it.}\label{tbl:results}
\end{table}

\paragraph{Results for Monotone Submodular Functions.} As mentioned above, bicriteria maximization of monotone submodular functions subject to cardinality and knapsack constraints has been implicitly studied in the literature through the lens of the submodular cover problem, and it is known that a simple greedy algorithm can obtain for this problem an approximation guarantee of roughly $(1 - \eps, 1 + \ln \eps^{-1})$ for any $\eps \in (0, 1]$, which is the best possible up to the constant $1$ in the infeasibility ratio. For completeness, we reprove in this paper the guarantee of this greedy algorithm (Theorem~\ref{thm:monotone_knapsack_greedy}), and the hardness result (Observation~\ref{obs:monotone_hardnesss}). Unsurprisingly, we also show that the Measured Continuous Greedy of \cite{feldman2011unified} can be used to get essentially the same result also for maximizing the multilinear extension of a monotone submodular function subject to a general solvable\footnote{A covex set is solvable if one can efficiently optimize linear functions subject to it.} down-closed convex set (Corollary \ref{cor:monotone_downclosed}), and even subject to non-down-closed solvable convex sets when $\eps \leq 1/e$ (Corollary~\ref{cor:monotone_downclosed}). Using an appropriate rounding procedure, this immediately implies $(1 - \eps, \lceil\ln \eps^{-1}\rceil)$-bicriteria approximation for the problem maximizing a monotone submodular function subject to a matroid constraint (for any $\eps \in (0, 1]$). Using similar ideas, we also get a simple greedy algorithm that guarantees $(1 - \eps, \lceil \log_2 \eps^{-1} \rceil)$-bicriteria approximation for the same problem (Theorem~\ref{thm:monotone_matroid_greedy}).

Going beyond the above relatively straightforward results, we study the bicriteria approximability of monotone submodular functions subject to cardinality constraints as a function of the density\footnote{The density of a constraint is intuitively the infeasibility ratio that would have resulted by returning the entire ground set $\cN$ as the output solution.} $c$ of these constraints. This problem was recently studied by \citet{filmus2025separating} in the special case of single-criterion maximization. \citet{filmus2025separating} left a gap between their algorithmic and inapproximability results, which we were able to close. 
Furthermore, our results hammer down the exact approximability of the problem also for bicriteria approximation. Specifically, we show (Theorem~\ref{thm:knapsack_fractional}) that for density $c \leq 1/2$, the Measured Continuous Greedy of \cite{feldman2011unified} obtains the optimal approximation guarantee, while for $c \geq 1/2$ the optimal guarantee can be obtained by adapting the Continuous Double Greedy algorithm suggested by \cite{buchbinder2014submodular} for the special case of single-criterion maximization. Theorem~\ref{thm:monotone_cardinality_inapproximability} proves that the results of these two algorithms are optimal. For $c \geq 1/2$ this is done by extending a result of \cite{filmus2025separating} to the bicriteria setting. However, getting a tight inapproximability result for $c \leq 1/2$ turned out to be much more involved. As was observed by \cite{filmus2025separating}, the symmetry gap technique of \cite{vondrak2013symmetry} does not seem to be able to do that (except at densities $c$ for which $1/c$ happens to be an integer). Thus, to get an inapproximability result that applies for every values of $c$, we had to come up with a novel class of hard instances. Unlike the instances resulting from the symmetry gap technique, instances in our class do not posses more than a single optimal solution. It should also be noted that the above stated algorithmic results for cardinality constraints extend also to the more general knapsack constraints, up to an additional loss term of $1$ in the infeasibility ratio resulting from the rounding required for such constraints. Our results for monotone submodular functions are formally stated and proved in Section~\ref{sec:monotone}.


\paragraph{Results for General Submodular Functions.} \citet{crawford2023scalable} described an algorithm guaranteeing $(1/2 - \eps, O(\eps^{-2}))$-bicriteria approximation for maximizing general submodular functions subject to cardinality and knapsack constraints. We improve over this state of the art by obtaining  $(1/2 - \eps, O(\eps^{-1}))$-bicriteria approximation for any down-closed solvable 
convex set (Theorem~\ref{thm:general_down-closed_convex_body}). The algorithm that we use to obtain this result uses $O(\eps^{-1})$ executions of a variant of Measured Continuous Greedy to get $O(\eps^{-1})$ almost disjoint solutions,\footnote{To be more specific, the solutions are fractional, and they are ``almost disjoint'' in the sense that for every element, the sum of its coordinates in all the solutions adds up to at most $2$.} which are then post-processed (with the help of an algorithm for unconstrained submodular maximization) to get the final solution. By using discrete greedy algorithms instead of Measured Continuous Greedy, we also get a simple combinatorial algorithm that has the same guarantee for cardinality, knapsack and matroid constraints (Theorem~\ref{thm:general_knapsack_matroid}).
To illustrate the ideas behind these algorithms,
we provide in Section~\ref{sec:warmup}, as a warmup, a simplified version of Theorem~\ref{thm:general_knapsack_matroid}
that applies only to cardinality constraints. 

We complement the above algorithmic results by proving two inapproximability results extending two known inapproximability results for single-criterion optimization (due to \cite{gharan2011submodular,qi2022maximizing,vondrak2013symmetry}). The first inapproximability result (Theorem.~\ref{thm:hardness_non_down-closed}) shows that the restriction of our first algorithm to down-closed convex sets is unavoidable because in general one cannot get any constant bicriteria approximation guarantee for non-down-closed convex sets. The second inapproximability result (Theorem.~\ref{thm:hardness_general}) upper bounds the approximation ratio that can be obtained for any infeasibility ratio $\beta$ even for simple cardinality constraints. Interestingly, we also show that Measured Continuous Greedy is able to secure this approximation guarantee on instances having two disjoint optimal solutions (or two disjoint solutions that mimic optimal solutions is some sense). Since existing techniques for proving inapproximability results for submodular functions usually generate instances having such two disjoint optimal solutions, this observation implies that further pushing the inapproximability beyond the bound that we have proved is likely to require a significant technical innovation. Our results for general submodular functions are formally stated and proved in Section~\ref{sec:general}. 

\paragraph{Results for Symmetric Submodular Functions.} The starting point for our results for bicriteria maximization of symmetric submodular functions is a variant of the Measured Continuous Greedy algorithm due to \cite{feldman2017maximizing}. By running this algorithm for a variable amount of time, it is possible to get roughly $(\nicefrac{1}{2} - \eps, \nicefrac{1}{2} \ln (\eps^{-1}/2))$-bicriteria integer-approximation for maximizing the multilinear extensions of such functions subject to arbitrary down-closed solvable convex sets (Theorem.~\ref{thm:symmetric_polytope}). Using appropriate rounding procedures, this implies similar bicriteria approximation results for maximizing symmetric submodular functions subject to particular kinds of constraints (such as matroids). Furthermore, the idea underlying the above variant of the Measured Continuous Greedy algorithm can be used to get simple greedy algorithms for maximizing symmetric submodular functions subject to cardinality and knapsack constraints (Theorem~\ref{thm:symmetric_knapsack_greedy}).

As for monotone functions, we also have improved results for cardinality and knapsack constraints as function of the density $c$ of these constraints, and these results are the more technically challenging ones. It is possible to derive some such results using the above mentioned variant of the Measured Continuous Greedy algorithm due to \cite{feldman2017maximizing}. However, we are able to get an even better dependence on $c$ (roughly, $(\nicefrac{1}{2} - \eps, \frac{1 - (2\eps)^c}{2c})$-bicriteria approximation) by developing a novel variant of Measured Continuous Greedy. Our novel variant is noticeable for increasing the solution in a reduced rate compared to standard Measured Continuous Greedy. This reduced rate decreases the infeasibility ratio of our variant, and also allows it to avoid the non-smooth decrease step used by \cite{feldman2017maximizing} to get results for symmetric functions. It should also be noted that even in the special case of a single-criterion maximization, our novel variant of Measured Continuous Greedy improves over the state-of-the-art for values of $c$ smaller than roughly $0.282$. Our results for symmetric submodular functions are formally stated and proved in Section~\ref{sec:symmetric}.

We conclude this section with two remarks.
	First, as is standard in the literature, we assume that algorithms have access to the submodular objective function $f$ only through a value oracle that given a subset $S$ of the ground set returns $f(S)$. This assumption is especially essential for the proofs of our inapproximability results.
	Second, many of our results involve constant values (for example: $\eps$, $\delta > 0$).
The time complexities of our algorithms are mostly polynomial in these constant values (the sole exception are algorithms involving rounding based on Lemma~\ref{lem:rounding_knapsack}), and thus, our algorithms can
usually be used also when these values are non-constant.
However, the analyses of the algorithms sometimes depend on the assumption that, as the size of the ground set of the input instance increases, the constant values do not approach ``too fast'' the open ends of their respective allowed ranges.


%% file: Preliminaries.tex
\section{Preliminaries} \label{sec:preliminaries}

This section presents the notation used by the paper, and the existing results that we use. The section is mostly arranged based on the kinds of constraints that we consider.

\subsection{Cardinality and Knapsack Constraints}

Given a ground set $\cN$, a knapsack constraint is formally defined by a non-negative vector $\vp \in \bR^\cN$ and a non-negative budget $B$. A set $S$ is then feasible according to the knapsack constraint if $\sum_{u \in S} p_u = \inner{\characteristic_S}{\vp} \leq B$. A cardinality constraint is a particularly important special case of a knapsack constraint in which $\vp = \characteristic_\cN$ and $B$ is integral.

We are often interested in the problem of maximizing a non-negative submodular function subject to either a cardinality or a knapsack constraint. Notice that the restriction to non-negative submodular functions is unavoidable as we aim for multiplicative approximation guarantees. The approximation ratios that we obtain for this problem sometimes depend on the density of the constraint, which is defined as $c \triangleq B / \|\vp\|_1$. One can assume that the density is a value from the range $(0, 1)$ because for $c = 0$ the problem is trivial ($\varnothing$ is the only fesible solution), and for $c \geq 1$ the problem reduces to the well-understood case of unconstrained submodular maximization.

Except in the context of Theorem~\ref{thm:symmetric_knapsack_equality} (which is discussed separately in Section~\ref{ssc:improved_results_symmetric}), whenever we consider knapsack constraints, we implicitly assume that the cost of every single element does not exceed the budget. This assumption is without loss of generality because if it is violated, one can drop from the ground set the violating elements before executing our algorithms. Dropping such elements does not deteriorate the guarantees of our algorithms for two reasons. First, since the dropped elements cannot be part of any feasible solution, dropping them does not affect the value of the optimal solution. Second, dropping elements can only increase the density of the constraint, and the guarantees of all of our algorithms (except the one of Theorem~\ref{thm:symmetric_knapsack_equality}) do not deteriorate when the density increases.

\subsection{Matroid Constraints}

Given a ground set $\cN$, a matroid $\cM = (\cN, \cI)$ is a pair such that: (i) $\cI$ is a non-empty subset of $2^\cN$, (ii) if $T \in \cI$, then every set $S \subseteq T$ also belongs to $\cI$, and finally, (iii) if $S, T \in \cI$ and $|S| < |T|$, then there exists an element $u \in T \setminus S$ such that $S \cup \{u\} \in \cI$. The sets in $\cI$ are called \emph{independent sets}, and they are the feasible sets when the matroid $\cM$ is used as a constraint. Matroids are useful objects as they capture diverse settings ranging from independent sets in a linear vector space to acyclic sets of edges in a graph. Thus, we are interested in the problem of maximizing a non-negative submodular function subject to a matroid constraint.

A base of a matroid is a set that is maximally independent. One can verify that part (iii) of the definition of a matroid implies that all the bases of a particular matroid have the same size, which is termed the rank of the matroid. The following lemma states another useful property of matroid bases.
\begin{lemma}[Proved by~\cite{brualdi1969comments}, and can also be found as Corollary~39.12a in~\cite{schrijver2003combinatorial}] \label{le:perfect_matching_two_bases}
Let $A$ and $B$ be two bases of a matroid $\cM = (\cN, \cI)$. Then, there exists a bijection $h\colon A \setminus B \rightarrow B \setminus A$ such that for every $u \in A \setminus B$, $(B - h(u)) + u \in \cI$.
\end{lemma}
One can extend the domain of the function $h$ from the last lemma to the entire set $A$ by defining $h(u) = u$ for every $u \in A \cap B$. This yields the following corollary.
\begin{corollary} \label{cor:perfect_matching_two_bases}
Let $A$ and $B$ be two bases of a matroid $\cM = (\cN, \cI)$. Then, there exists a bijection $h\colon A \rightarrow B$ such that for every $u \in A$, $(B - h(u)) + u \in \cI$ and $h(u) = u$ for every $u \in A \cap B$.
\end{corollary}

\subsection{Convex Sets}

Many modern submodular maximization algorithms are based on (approximately) optimizing a fractional relaxation of the problem, and then rounding the fractional solution obtained. The objective function used by these fractional relaxations is usually the multilinear extension of the original set function. To be more specific, given a set function $f\colon 2^\cN \to \bR$, its multilinear extension is a function $F\colon [0, 1]^\cN \to \bR$ defined as follows. For every vector $\vx \in [0, 1]^\cN$, $\RSet(\vx)$ is a random set that contains every element $u \in \cN$ with probability $x_u$, independently. Then,
\[
	F(\vx)
	\triangleq
	\bE[f(\RSet(\vx))]
	=
	\sum_{S \subseteq \cN} f(S) \cdot \prod_{u \in S} x_u \cdot \prod_{u \in \cN \setminus S} (1 - x_u)
	\enspace.
\]
One can observe that, as its name suggests, the multilinear extension is a multilinear function of the coordinates of $\vx$. This implies the following observation (that appears, for example, as Observation 2.2 of \citet{feldman2011unified}).
\begin{observation}\label{obs:multilinear-partial}
  Let $F( \vx )$ be the multinear extension of a submodular function
  $f\colon 2^{\cN}\to \bR$. Then, for every $u \in \cN$,
  \[ \frac{\partial F( \vx )}{\partial x_u} = \frac{F(\vx \vee \vone_{\{u\}}) - F(\vx)}{1 - \vx_u } = \frac{F( \vx ) - F( \vx \wedge \vone_{\cN \setminus \{u\}) }}{\vx_u } \enspace. \]
\end{observation}

Solving the above-mentioned relaxations boils down to maximizing the multilinear extension of a non-negative submodular function subject to a constraint defined by a convex set $\cP \subseteq [0, 1]^\cN$. There are two technical issues stemming from the fact that this is a continuous (rather than combinatorial) optimization problem.
\begin{itemize}
	\item We need to adapt the definition of $(\alpha, \beta)$-bicriteria approximation algorithms (given in Section~\ref{sec:introduction}) to continuous problems. The most immediate adaptation is that such an algorithm can output a vector $\vx$ if there exist vectors $\vy_1, \vy_2 \in \cP$ such that $\vy_1 \leq \vx \leq \beta \cdot \vy_2$. However, using this definition is problematic since the multilinear extension is defined only for vectors in $[0, 1]^\cN$. Thus, we must further restrict bicriteria approximation algorithms to outputting only vectors from the set $\{\vx \in [0, 1]^\cN \mid \exists_{\vy_1, \vy_2 \in \cP}\; \vy_1 \leq \vx \leq \beta \cdot \vy_2\}$.
	\item  As mentioned in Section~\ref{ssc:our_results}, our algorithms for maximizing multilinear extensions subject to convex set constraints produce fractional solutions, but have bicriteria approximation guarantees with respect to the best feasible integral solution. To stress this, we refer to these algorithms as \emph{bicriteria integer-approximation algorithms}. 
Fortunately, approximation guarantees with respect to the best integral solution are good enough when our goal is to round the fractional solution obtained, and get a bicriteria approximation algorithm for the original combinatorial problem.
\end{itemize}

The Pipage Rounding algorithm of \citet{calinescu2011maximizing} can be used to round fractional solutions for the problem of maximizing a submodular function subject to a matroid constraint at the cost of rounding up the infeasibility ratio. The next lemma formally describes the guarantee that can be obtained in this way (see Appendix~\ref{app:pipage_rounding_matroid} for the proof).
\begin{restatable}{lemma}{lemRoundingMatroid} \label{lem:rounding_matroid}
Consider a matroid $\cM = (\cN, \cI)$, the matroid polytope $\cP$ of $\cM$ (i.e., $\cP$ is the convex hull of the characteristic vectors of all the sets in $\cI$), a non-negative submodular function $f$, and the multilinear extension $F$ of $f$. There exists a polynomial time algorithm that gets as input a vector $\vx \in [0, 1]^\cN$ and a scalar $\beta > 0$ such that $\vx / \beta \in \cP$, and outputs a random set $S$ obeying two properties: $\bE[f(S)] \geq F(\vx)$, and $\characteristic_S / \lceil \beta \rceil \in \cP$.
\end{restatable}

For the knapsack constraints, it is still possible to use a variant of Pipage Rounding. However, this time the increase in the infeasibility ratio is $1$ (except in the special case of a cardinality constraint, in which this increase is at most $1/B$). See the next lemma for more details (the proof of Lemma~\ref{lem:rounding_cardinality_knapsack} can be found in Appendix~\ref{app:pipage_rounding_knapsack}).

\begin{lemma} \label{lem:rounding_cardinality_knapsack}
Consider the problem of maximizing a non-negative submodular function $f\colon 2^\cN \to \nnR$ subject to a knapsack constraint consisting of costs $\vp$ and budget $B$, and let $F$ be the multilinear extension of $f$. There exists a polynomial time algorithm that given a vector $\vx \in [0, 1]^\cN$ outputs a set $S \subseteq \cN$ such that $\sum_{u \in S} p_u \leq \inner{\vp}{\vx} + B$ and $\bE[f(S)] \geq F(\vx)$. In other words, the algorithm rounds $\vx$ without losing value (in expectation) at the cost of increasing the cost by at most $B$. In the special case of a cardinality constraint, the the guarantee of the above algorithm improves to $|S| = \sum_{u \in S} p_u \leq \lceil \inner{\vp}{\vx}\rceil$.
\end{lemma}

\citet{kulik2013} described a different method for rounding fractional solutions under a knapsack constraint that does not increase the infeasibility ratio at all. Their method consists of three steps: (i) a pre-processing step that modifies the input instance, (ii) obtaining a fractional solution for the modified instance using an arbitrary algorithm, and (iii) a rounding step. The existence of the pre-processing step has the unfortuante consequence that the method of \cite{kulik2013} does not preserve approximation ratios that depend on properties of the input instance such as the density of the knapsack constraint or the symmetry of the objective function. The following lemma formally states the guarantee that can be obtained using the method of \cite{kulik2013} (technically, \citet{kulik2013} has worked on single-criteria optimization only, but their method extends to bicriteria optimization).

\begin{lemma} \label{lem:rounding_knapsack}
Let $\alg$ be an $(\alpha, \beta)$-bicriteria approximation algorithm for maximizing the multilinear extension of a non-negative submodular function subject to a knapsack constraint. If $\alpha$ and $\beta \geq 1$ are independent of the properties of the input instance (except the size of its ground set and whether the objective function is monotone), then for every constant $\delta > 0$, there exists an $(\alpha - \delta, \beta)$-bicriteria approximation algorithm for maximizing a non-negative submodular function subject to a knapsack constraint.
\end{lemma}

\subsection{Additional Notation}

Given a set $S$ and an element $u$, we often use $S + u$ and $S - u$ as shorthands for $S \cup \{u\}$ and $S \setminus \{u\}$, respectively. When given also a set function $f$ and another set $T$, we employ the additional shorthands $f(u \mid S) \triangleq f(S + u) - f(S)$ and $f(T \mid S) \triangleq f(S \cup T) - f(S)$ to denote the marginal contributions to $S$ of $u$ and $T$, respectively. Consider now a ground set $\cN$, and two vectors $\vx, \vy \in [0, 1]^\cN$. We use the following binary operators on such vectors.
\begin{itemize}
	\item $\vx \vee \vy$ denotes the coordinate-wise maximum of $\vx$ and $\vy$. In other words, $\vx \vee \vy$ is a vector in $[0, 1]^\cN$ such that $(\vx \vee \vy)_u = \max\{x_u, y_u\}$.
	\item $\vx \wedge \vy$ denotes the coordinate-wise minimum of $\vx$ and $\vy$.
	\item $\vx \hprod \vy$ denotes the coordinate-wise product (also known as Hadamard product) of $\vx$ and $\vy$.
	\item $\vx \psum \vy$ denotes the coordinate-wise probabilistic sum of $\vx$ and $\vy$. Formally, $\vx \psum \vy = \vone_\cN - (\vone_\cN - \vx) \hprod (\vone_\cN  - \vy) = \vx + \vy - \vx \hprod \vy$.
\end{itemize}
One can verify that all the above binary operators are associative and commutative. We also note that whenever we use a comparison between vectors (such as $\vx \leq \vy$), this comparison holds coordinate-wise.


%% file: Warmup.tex
\section{Warmup} \label{sec:warmup}

This section aims to demonstrate the intuition underlying our bicriteria algorithms for general submodular functions. Towards this goal, we present and prove in this section the following simplified version of Theorem~\ref{thm:general_knapsack_matroid} that applies only for cardinality constraints (see Section~\ref{ssc:combinatorial_general} for the formal statement and proof of the original Theorem~\ref{thm:general_knapsack_matroid}). 

\begin{theorem} \label{thm:general_cardinality}
For every constant $\eps \in (0, 1/2)$, there exists a polynomial time $(1/2 - \eps, O(\eps^{-1}))$-bicriteria approximation algorithm for the problem of maximizing a non-negative submodular function $f\colon 2^\cN \to \nnR$ subject to a cardinality constraint.
\end{theorem}

The algorithm that we use to prove Theorem~\ref{thm:general_cardinality} appears as Algorithm~\ref{alg:comb-bicriteria}. This algorithm starts by greedily constructing $\ell$ sets $A_1, A_2, \dotsc, A_\ell$. The set $A$ is the union of these sets. For every $i \in [\ell]$, the algorithm also constructs a set $D_i$ which approximates the best way to extend $A_i$ using elements of $A$. The output of the algorithm is then the set $A_i \cup D_i$ for the best choice of $i$. Algorithm~\ref{alg:comb-bicriteria} implicitly assumes that the ground set includes $2\ell B$ dummy elements that do not affect the value of $f$.\footnote{In other words, if $Z$ is the set of dummy elements, then the equality $f(S) = f(S \setminus Z)$ holds for every $S \subseteq \cN$.} This assumption is without loss of generality because one can artificially add such elements before executing Algorithm~\ref{alg:comb-bicriteria}, and then remove them from the output of the algorithm, which does not affect the value of this output. Notice that the existence of the dummy elements guarantees that Algorithm~\ref{alg:comb-bicriteria} always has an element $u_{i, j} \in \cN \setminus A$ to choose on Line~\ref{line:greedy_selection}, and furthermore, the element selected always obeys $f(u_{i, j} \mid A_i) \geq 0$ with respect to the set $A_i$ at the time of the selection.

\begin{algorithm}[h]
\caption{\textsc{Combinatorial Algorithm for Cardinality Constraints} $(f, B, \eps)$} \label{alg:comb-bicriteria}
Let $A \gets \varnothing$ and $\ell \gets \lceil \frac{1}{2\eps} \rceil$.\\
\For{$i = 1$ to $\ell$}
{
	Let $A_i \gets \varnothing$.\\
	\For{$j = 1$ \KwTo $2B$}
	{
		Let $u_{i, j}$ be an element in $\cN \setminus A$ maximizing $f(u_{i, j} \mid A_i)$.\label{line:greedy_selection}\\
    Update $A_i \gets A_i + u_{i, j}$ and $A \gets A + u_{i, j}$.
  }
}
\For{$i = 1$ to $\ell$}
{
  Let $D_i$ be a subset of $A$ obeying $f(A_i \cup D_i) \geq \nicefrac{1}{4} \cdot f(A_i) + \nicefrac{1}{2} \cdot \max_{D \subseteq A} f(A_i \cup D)$.\label{D_i_selection}
}
Let $i^* \gets \argmax_{i \in [\ell]} f( D_i \cup A_i )$.\\
\Return $D_{i^*} \cup A_{i^*}$.
\end{algorithm}

We begin the analysis of Algorithm~\ref{alg:comb-bicriteria} with the following two observations.
\begin{observation} \label{obs:polynomial}
Algorithm~\ref{alg:comb-bicriteria} has a polynomial time implementation.
\end{observation}
\begin{proof}
To prove the observation, it suffices to show that Line~\ref{D_i_selection} of Algorithm~\ref{alg:comb-bicriteria} has a polynomial time implementation. Let us define $g_i \colon 2^A \to \nnR$ to be the function $g_i(S) \triangleq f( A_i \cup S )$. Then, Line~\ref{D_i_selection} can be rewritten as finding a set $D_i \subseteq A$ such that $g_i(D_i) \geq \nicefrac{1}{4} \cdot g_i(\varnothing) + \nicefrac{1}{2} \cdot \max_{D \subseteq A} g_i(D)$. Since the function $g_i$ is also a non-negative submodular function, one can deterministically find such a set $D_i$ using the algorithm of \citet{buchbinder2018deterministic} for unconstrained submodular maximization.\footnote{In the unconstrained submodular maximization problem, given a non-negative submodular function $g$, the objective is to find an arbitrary subset of the ground set of $g$ that maximizes $g$.} We note that this algorithm is a derandomized version of the well-known Double Greedy algorithm of~\cite{buchbinder2015tight}. 
\end{proof}

\begin{observation} \label{obs:infeasibility}
The infeasibility ratio of Algorithm~\ref{alg:comb-bicriteria} is at most $O(\eps^{-1})$.
\end{observation}
\begin{proof}
Notice that $A$ is the disjoint union of the $\ell$ sets $A_1, A_2, \dotsc, A_\ell$, and the size of each one of these sets is $2B$. Thus,
\[
	|A|
	=
	\ell \cdot 2B
	=
	\bigg\lceil \frac{1}{2\eps} \bigg\rceil \cdot 2B
	\leq
	\frac{2B}{\eps}
	=
	O(\eps^{-1}) \cdot B
	\enspace.
\]
The observation now follows since the output set of Algorithm~\ref{alg:comb-bicriteria} is a subset of $A$.
\end{proof}

The rest of this section is devoted to bounding the approximation ratio of Algorithm~\ref{alg:comb-bicriteria}. The next lemma shows that the first half of the 
elements selected for every set $A_i$ have at least half of the value of the entire set. This
is a consequence of the submodularity of $f$ and the greedy rule used to fill $A_i$.
To formally state this lemma, it is convenient to define $A_{i, j}$ as the set of the first $j$ elements added to $A$. Formally, $A_{i, j} \triangleq \{a_{i, k} \mid k \in [j]\}$ for every $i \in [\ell]$ and $j \in \{0, 1, \dotsc, 2B\}$.
\begin{lemma} \label{lem:greedy}
For all $i \in [\ell]$, $\ff{ A_{i,B}}\ge \ff{ A_{i,2B} } / 2$.
\end{lemma}
\begin{proof}
Notice that
  \begin{align*}
    \ff{ A_{i,2B} } &= \ff{ A_{i,B} } + \sum_{k=1}^{B} \marge{ a_{i,k + B} }{ A_{i,k + B -1} }
                                       \le \ff{ A_{i,B} } + \sum_{k=1}^{B} \marge{ a_{i,k + B} }{ A_{i,k-1} } \\
																			 &\le \ff{ A_{i,B} } + \sum_{k=1}^{B} \marge{ a_{i,k} }{ A_{i,k-1} }
                                       = 2 \cdot \ff{ A_{i,B} } - \ff{ \varnothing }
																			\le 2 \cdot \ff{ A_{i,B} }
																			\enspace,
  \end{align*}
  where the first inequality follows from the submodularity of $f$,
  the second inequality follows from the greedy rule used by Algorithm~\ref{alg:comb-bicriteria} to select $a_{i,k}$, and
	the last inequality follows from the non-negativity of $f$.
  The lemma now follows by rearranging this inequality.
\end{proof}

Let us now define $\hat O \triangleq \opt \setminus A$ and $\dot O \triangleq \opt \cap A$. Since any element of $\hat O$ could have 
been selected at any iteration of Algorithm~\ref{alg:comb-bicriteria}, but was never selected, we can upper bound the contribution of these elements to $A_i$ by the contribution of the elements selected in the second half of the construction of $A_i$. This leads to the following lemma.
\begin{lemma} \label{lem:ohat}
  For all $i \in [\ell]$,
  $f(\hat O \mid A_i ) \le f(A_i)/2$.
\end{lemma}
\begin{proof}
Let us denote the elements of $\hat O$ by $o_1, o_2, \dotsc, o_{|\hat O|}$. Then,
  \begin{align*}
    f(\hat O \mid A_i ) &\le \mspace{-9mu} \sum_{j = 2B - |\hat O| + 1}^{2B} \mspace{-18mu} f(o_{j + |\hat O| - 2B} \mid A_{i})
																				\le \mspace{-9mu} \sum_{j = 2B - |\hat O| + 1}^{2B} \mspace{-18mu} f(o_{j + |\hat O| - 2B} \mid A_{i, j - 1})
                                        \le \mspace{-9mu} \sum_{j = 2B - |\hat O| + 1}^{2B} \mspace{-18mu} f(a_{i, j} \mid A_{i, j - 1}) \\ \\
                                        &\le \sum_{j = B + 1}^{2B} f(a_{i, j} \mid A_{i, j - 1}) 
                                        = \ff{A_i } - \ff{ A_{i,B} }
                                        \le \ff{ A_{i} } - \ff{A_i} / 2 = \ff{A_i}/2
																				\enspace,
  \end{align*}
  where the first two inequalities hold by the submodularity of $f$,
  the third inequality follows from the greedy selection rule used by Algorithm~\ref{alg:comb-bicriteria} to select $a_{i,j}$,
	the fourth inequality holds since the introduction of the dummy elements guarantees that $f(a_{i, j} \mid A_{i, j - 1}) \ge 0$ for every $i \in [\ell]$ and $j \in [2B]$,
	and the last inequality follows from Lemma~\ref{lem:greedy}.
\end{proof}

We are now ready to prove the approximation ratio of Algorithm~\ref{alg:comb-bicriteria}.
\begin{lemma} \label{lem:approximation-ratio}
The approximation ratio of Algorithm~\ref{alg:comb-bicriteria} is at most $1/2 - \eps$.
\end{lemma}
\begin{proof}
Using repeated applications of the submodularity of $f$, one can get
\[
	\sum_{i \in [\ell]} \ff{ A_i \cup \opt } \ge \sum_{i = 2}^\ell f\Big(\Big(A_i \cap \bigcup\nolimits_{j = 1}^{i - 1} A_j\Big) \cup \opt\Big) + f( A \cup \opt )
	\ge
	( \ell - 1 ) \cdot f(\opt)
	\enspace,
\]
where the last inequality follows from the non-negativity of $f$ and the fact the $A_i$'s are pairwise disjoint. Thus, by an averaging argument, there must exist an index $\hat{\imath} \in [\ell]$ for which $f(A_{\hat{\imath}} \cup \opt) \ge ( 1 - \frac{1}{\ell } ) \cdot f(\opt)$. In the following, we show that $\ff{ A_{\hat{\imath}} \cup D_{\hat{\imath}} } \ge ( \frac{1}{2} - \eps) \cdot f( \opt )$, which implies the lemma since Algorithm~\ref{alg:comb-bicriteria} outputs a set of value $\max_{i \in [\ell]} f(A_i \cup D_i)$.

Recall that, by the way $D_{\hat{\imath}}$ is chosen by Algorithm~\ref{alg:comb-bicriteria},
  \begin{equation} \label{eq:dg_guarantee}
    \ff{ D_{\hat{\imath}} \cup A_{\hat{\imath}} } \ge \frac{1}{4} \cdot f(A_{\hat{\imath}}) +  \frac{1}{2} \cdot \max_{D \subseteq A} f(A_{\hat{\imath}} \cup D) \ge \frac{1}{4} \cdot f(A_{\hat{\imath}}) + \frac{1}{2} \cdot f( A_{\hat{\imath}} \cup \dot O) \enspace,
  \end{equation}
  where the second inequality holds since $\dot O$ is a subset of $A$ by definition. Thus,
  \begin{align*}
    (1 - 2 \eps) \cdot f(\opt) &\le \Big( 1 - \frac{1}{\ell } \Big) \cdot f(\opt)
    \le f(A_{\hat{\imath}} \cup \opt) 
                        \le f( A_{\hat{\imath}} \cup \hat O  )  - \ff{ A_{\hat{\imath}} } + f( A_{\hat{\imath}} \cup \dot O ) \\
												&= f(\hat{O} \mid A_{\hat{\imath}}) + f( A_{\hat{\imath}} \cup \dot O )
                        \le \frac{f(A_{\hat{\imath}})}{2} + \bigg[2f( D_{\hat{\imath}} \cup A_{\hat{\imath}} ) - \frac{f(A_{\hat{\imath}})}{2}\bigg] \\
                        &= 2f( D_{\hat{\imath}} \cup A_{\hat{\imath}} ),
  \end{align*}
  where the first inequality follows from the definition of $\ell$,
	the third inequality follows from the submodularity of $f$,
  and the last inequality follows from Lemma~\ref{lem:ohat} and Inequality~\eqref{eq:dg_guarantee}.
\end{proof}

Theorem~\ref{thm:general_cardinality} now follows from Observations~\ref{obs:polynomial} and~\ref{obs:infeasibility} and Lemma~\ref{lem:approximation-ratio}.


%% file: MonotoneFunctions.tex
\section{Monotone Submodular Functions} \label{sec:monotone}

This section includes our results for bicriteria maximization of monotone submodular functions. In Section~\ref{ssc:monotone_general}, we present optimal results for general and down-closed convex set constraints. For cardinality and knapsack constraints, we have improved results depending on the density of these constraints. These results are presented in Section~\ref{ssc:monotone_per_c_analysis}, and are optimal for cardinality constraints. Finally, in Section~\ref{ssc:simple_greedy_monotone}, we present simple greedy algorithms for bicriteria maximization of monotone submodular functions subject to cardinality, knapsack and matroid constraints.

\subsection{Algorithms for Down-Closed and General Convex Sets} \label{ssc:monotone_general}

In this section, study the maximization of the multilinear extension of a monotone submodular function subject to an arbitrary solvable convex set $\cP$. We begin by proving an inapproximability result for this problem.

\begin{observation} \label{obs:monotone_hardnesss}
For every two constants $\eps \in (0, 1)$ and $\delta \in (0, \eps]$, no polynomial time algorithm can guarantee $(1 - \eps + \delta, \ln \eps^{-1})$-bicrtieria integer-approximation for the problem of maximizing the multilinear extension of a non-negative monotone submodular function $f\colon 2^\cN \to \nnR$ subject to a solvable convex set $\cP \subseteq [0, 1]^\cN$. Furthermore, this inapproximability result holds even when $\cP$ is restricted to be the continuous relaxation of a cardinality constraint (i.e., $\cP = \{\vx \in [0, 1]^\cN \mid \|\vx\|_1 \leq \cardBudget\}$ for some positive integer $\cardBudget$).
\end{observation}
\begin{proof}
Assume towards a contraction that the observation is false. Thus, there exists a polynomial time algorithm $\alg$ that given a non-negative monotone submodular function $f\colon 2^\cN \to \nnR$ and an integer $\cardBudget$ outputs a vector $\vx \in [0, 1]^\cN$ such that $\|\vx\| \leq \cardBudget \cdot \ln \eps^{-1}$ and $\bE[F(\vx)] \geq (1 - \eps + \delta) \cdot f(\opt)$, where $\opt$ is the set of size $\cardBudget$ maximizing $f$. By Lemma~\ref{lem:rounding_cardinality_knapsack}, there is a rounding procedure that given such a vector $\vx$ outputs a set $S \subseteq \cN$ such that $|S| \leq \cardBudget \cdot (\ln \eps^{-1} + 1/\cardBudget)$ and $\bE[f(S)] \geq (1 - \eps + \delta) \cdot f(\opt)$. Combining $\alg$ with this rounding procedure yield a polynomial time algorithm for the problem of maximizing a non-negative monotone submodular function subject to a cardinality constraint that guarantees $(1 - \eps + \delta, \ln \eps^{-1})$-bicriteria approximation regardless of the density of the cardinality constraint.

Let us explain why the existence of such an algorithm contradicts Theorem~\ref{thm:monotone_cardinality_inapproximability} from Section~\ref{ssc:monotone_per_c_analysis} below. This theorem shows that every sub-exponential time algorithm that obtains a constant approximation ratio better than $1 - \eps + \delta/2$ for maximizing a non-negative monotone submodular function subject to a cardinality constraint of constant density $c \in (0, 1/2)$ must have an infeasibility ratio of at least $\frac{1 - (\eps - \delta/2)^c}{c}$. Furthermore, since the problem is polynomially solvable when $\cardBudget$ is constant, this result must apply even when $\cardBudget$ is assumed to be arbitrarily large. However, for small enough values of $c$ and large enough values of $\cardBudget$ the infeasibility ratio $\frac{1 - (\eps - \delta/2)^c}{c}$ is larger than $\ln \eps^{-1} + 1/\cardBudget$ because, by L'H\^{o}pital's rule,
\[
	\lim_{c \to 0} \frac{1 - (\eps - \delta/2)^c}{c}
	=
	\lim_{c \to 0} \frac{-(\eps - \delta/2)^c\ln(\eps - \delta/2)}{1}
	=
	-\ln(\eps - \delta/2)
	>
	-\ln \eps
	=
	\lim_{B \to \infty} (\ln \eps^{-1} + 1/B)
	\enspace.
	\qedhere
\]
\end{proof}

Given the last observation, our objective in this section is to show that there exists an algorithm guaranteeing roughly $(1 - \eps, \ln \eps^{-1})$-bicriteria approximation for the problem we consider. In the special case in which $\cP$ is a down-closed convex set, such an algorithm follows from the following known result about an algorithm named Measured Continuous Greedy.
\begin{theorem}[Part of Theorem~I.2 of \cite{feldman2011unified}] \label{thm:measured_continuous_greedy_monotone}
For any non-negative monotone submodular function $f\colon 2^\cN \to \nnR$, down-closed solvable convex set $\cP \in [0, 1]^N$ and a constant stopping time $T \geq 0$, the Measured Continuous Greedy algorithm is a polynomial time algorithm that finds a point $\vx \in [0, 1]^\cN$ such that $\vx / T \in \cP$ and $\bE[F(\vx)] \geq (1 - e^{-T} - o(1)) \cdot f(\opt)$, where $\opt$ is a set maximizing $f(\opt)$ among all the sets $\{S \subseteq \cN \mid \characteristic_S \in P\}$.\footnote{Technically, as stated in~\citep{feldman2011unified}, the result is restricted to normalized monotone submodular functions and down-closed solvable polytopes rather than general non-negative monotone submodular functions and down-closed solvable convex sets. However, the proof of~\cite{feldman2011unified} does not use these restrictions in any way.}
\end{theorem}
\begin{corollary} \label{cor:monotone_downclosed}
For every constant $\eps \in (0, 1]$, there exists a polynomial time $(1 - \eps - o(1), \ln \eps^{-1})$-bicriteria integer-approx\-imation algorithm for the problem of maximizing the multilinear extension of a non-negative monotone submodular function $f\colon 2^\cN \to \nnR$ subject to a down-closed solvable polytope $\cP \subseteq [0, 1]^\cN$.
\end{corollary}
\begin{proof}
The corollary follows by plugging $T = \ln \eps^{-1}$ into Theorem~\ref{thm:measured_continuous_greedy_monotone}.
\end{proof}

Next, we would like to remove the requirement that $\cP$ is down-closed. The part of the proof of Theorem~\ref{thm:measured_continuous_greedy_monotone} that lower bounds $\bE[F(\vx)]$ uses the down-closedness of $\cP$ only to show that $f(\opt)$ upper bounds the value of $f(\{u\})$ for any element $u \in \cN$ that appears in some integral solution. When $f$ is monotone, this property holds even when $\cP$ is not down-closed, and thus, the lower bound on $\bE[F(\vx)]$ given by Theorem~\ref{thm:measured_continuous_greedy_monotone} applies also to such $\cP$'s. Preserving the property that $\vx/T \in \cP$ when $\cP$ is not down-closed is more involved. More specifically, the output of Measured Continuous Greedy is a vector
\[
	\vx
	=
	\PSum_{i = 1}^{T/\delta} (\delta \cdot \vx^{(i)})
\]
for some value $\delta \in (0, 1)$ such that $T/\delta$ is integral, and vectors $\vx^{(1)}, \vx^{(2)}, \dotsc, \vx^{(T/\delta)} \in \cP$. When $\cP$ is down-closed, such a vector is guaranteed to be included in $T \cdot \cP$. However, this is not the case in general. Thus, we suggest modifying Measured Continuous Greedy so that it outputs the vector $\vx \vee \vx'$, where
\[
	\vx'
	=
	\delta \cdot \bigg(\sum_{i = 1}^{\lceil \delta^{-1} - 1\rceil} \vx^{(i)} + (\delta^{-1} - \lceil \delta^{-1} - 1\rceil) \cdot \vx^{(\lceil \delta^{-1} \rceil)}\bigg)
	\enspace.
\]
Notice that such a modification can only be done when $T \geq 1$ because otherwise some of the vectors used in the last expressions might not be defined.

\begin{lemma} \label{lem:modified_measured_continuous_greedy_monotone}
For any normalized monotone submodular function $f\colon 2^\cN \to \nnR$, solvable convex set $\cP \in [0, 1]^N$ and constant stopping time $T \geq 1$, our modified version of
Measured Continuous Greedy algorithm is a polynomial time algorithm that finds a point $\vx \vee \vx' \in [0, 1]^\cN$ such that $\bE[F(\vx \vee \vx')] \geq (1 - e^{-T} - o(1)) \cdot f(\opt)$, where $\opt$ is a set maximizing $f(\opt)$ among all the sets $\{S \subseteq \cN \mid \characteristic_S \in P\}$, and there are vectors $\vy^{(1)} \in \cP$ and $\vy^{(2)} \in T \cdot \cP$ such that $\vy^{(1)} \leq \vx \vee \vx' \leq \vy^{(2)}$. 
\end{lemma}
\begin{proof}
Note that $\vx'$ is a convex combination of vectors in $\cP$, and therefore, also belongs to $\cP \subseteq [0, 1]^\cN$. Therefore, $\vx \vee \vx' \geq \vx' \in \cP$. Furthermore, both $\vx$ and $\vx'$ can be individually upper bounded by $\sum_{i = 1}^{T/\delta} \delta \cdot \vx^{(i)}$, and therefore,
\[
	\vx \vee \vx'
	\leq
	\sum_{i = 1}^{T/\delta} \delta \cdot \vx^{(i)}
	\in
	T \cdot \cP
	\enspace.
\]
It remains to note that since $f$ is monotone, $\bE[F(\vx \vee \vx')] \geq \bE[F(\vx)] \geq (1 - e^{-T} - o(1)) \cdot \max_{\vy \in P} F(\vy)$, where the second inequality uses the fact that the guarantee $\bE[F(\vx)] \geq (1 - e^{-T} - o(1)) \cdot \max_{\vy \in P} F(\vy)$ of Measured Continuous Greedy does not depend on the down-closedness of $\cP$, as explained above.
\end{proof}

We now get a version of Corollary~\ref{cor:monotone_downclosed} that works also when $\cP$ is a general convex set (but is restricted to $\eps \leq 1/e$).

\begin{corollary} \label{cor:monotone_general}
For every constant $\eps \in (0, 1/e]$, there exists a polynomial-time $(1 - \eps - o(1), \ln \eps^{-1})$-bicriteria integer-approximation algorithm for the problem of maximizing the multilinear extension of a non-negative monotone submodular function $f\colon 2^\cN \to \nnR$ subject to a solvable convex set $\cP \subseteq [0, 1]^\cN$.
\end{corollary}
\begin{proof}
The corollary follows by plugging $T = \ln \eps^{-1}$ into Lemma~\ref{lem:modified_measured_continuous_greedy_monotone}.
\end{proof}

\subsection{Improved Results for Knapsack and Cardinality Constraints} \label{ssc:monotone_per_c_analysis}

In this section, we show that one can improve over the guarantee of Section~\ref{ssc:monotone_general} when the convex set $\cP$ defining the constraint represents either a cardinality constraint or a knapsack constraint. Specifically, we prove Theorem~\ref{thm:knapsack_fractional}. Recall that a cardinality constraint is a special case of a knapsack constraint, and that the density of a knapsack constraint is $c \triangleq B / \|\vp\|_1$, where $B$ is the budget of the constraint, and $\vp$ is the vector of element prices. We also define
\[
	\rho(c, \eps)
	\triangleq
	\begin{cases}
		\frac{1 - \eps^c}{c} & \text{if $c \leq 1/2$} \enspace,\\
		\frac{1 - 2(1 - c)\sqrt{\eps} - \eps(2c - 1)}{c} & \text{if $c \geq 1/2$} \enspace.
	\end{cases}
\]

\begin{theorem} \label{thm:knapsack_fractional}
For every constant $\eps \in (0, 1]$, there exists a polynomial time $(1 - \eps - o(1), \rho(c, \eps))$-bicriteria integer-approximation algorithm for the problem of maximizing the multilinear extension of a non-negative monotone submodular function $f\colon 2^\cN \to \nnR$ subject to the polytope $\{\vx \in [0, 1]^\cN \mid \inner{\vp}{\vx} \leq B\}$ of a knapsack constraint with density $c$.
\end{theorem}

By Lemma~\ref{lem:rounding_cardinality_knapsack}, Theorem~\ref{thm:knapsack_fractional} immediately yields $(1 - \eps - o(1), \rho(c, \eps) + 1/\cardBudget)$-bicriteria approximation for maximizing a non-negative monotone submodular function subject to a cardinality constraint.\footnote{To be more accurate, the guarantee obtained is $(1 - \eps - o(1), \cardBudget^{-1}\lceil \cardBudget\rho(c, \eps)\rceil)$-bicriteria approximation, which is sometimes slightly stronger than $(1 - \eps - o(1), \rho(c, \eps) + 1/\cardBudget)$-bicriteria approximation.} Here one should think of $1/\cardBudget$ as an arbitrarily small constant because the problem is polynomially solvable for any constant $\cardBudget$. The following theorem proves that Theorem~\ref{thm:knapsack_fractional} is optimal up to arbitrarily small additive constants even in the special case of a cardinality constraint.

\begin{theorem} \label{thm:monotone_cardinality_inapproximability}
For every rational constant $c \in (0, 1)$ and two constants $\eps \in (0, 1)$ and $\delta \in (0, \eps]$, no algorithm using a sub-exponential number of value oracle queries can guarantee $(1 - \eps + \delta, \rho(c, \eps))$-bicriteria approximation for the problem of maximizing a non-negative monotone submodular function subject to a cardinality constraint.
\end{theorem}

Previously, \cite{filmus2025separating} studied the optimal (single-criteria) approximation ratio that can be obtained for maximizing a non-negative monotone submodular function subject to a cardnality constraint as a function of the density $c$ of the constraint. They managed to find this ratio for values of $c$ that are equal to $1$ over a positive integer, but had a small gap between their algorithmic and inapproximability results for other values of $c$, and left the determiniation of the the optimal ratio for such values of $c$ as an open problem. Theorems~\ref{thm:knapsack_fractional} and~\ref{thm:monotone_cardinality_inapproximability} solve this problem for all values of $c$ (only rational values of $c$ makes sense for a cardinality constraint) even in the more general bicriteria setting.

For the more general problem of maximizing a non-negative monotone submodular function subject to a knapsack constraint, Theorem~\ref{thm:knapsack_fractional} still gives roughly optimal fractional solution, but we can only get from it $(1 - \eps - o(1), \rho(c, \eps) + 1)$-bicriteria approximation via Lemma~\ref{lem:rounding_cardinality_knapsack}. Notice that this bicriteria approximation guarantee is incomparable with the $(1 - \eps - \delta, \ln \eps^{-1})$-bicriteria approximation for every $\eps \in (0, 1/e]$ that can be obtained for this problem by combining Corollary~\ref{cor:monotone_downclosed} with Lemma~\ref{lem:rounding_knapsack}.

The proof of Theorem~\ref{thm:knapsack_fractional} can be found in Section~\ref{sssc:knapsack_fractional}, and the proof of Theorem~\ref{thm:monotone_cardinality_inapproximability} appears in Section~\ref{sssc:cardinality_inapproximability_small_c} (for $c \leq 1/2$) and Section~\ref{sssc:cardinality_inapproximability_large_c} (for $c > 1/2$).

\subsubsection{Proof of Theorem~\texorpdfstring{\ref{thm:knapsack_fractional}}{\ref*{thm:knapsack_fractional}}} \label{sssc:knapsack_fractional}

In this section, we prove Theorem~\ref{thm:knapsack_fractional}. The proof is based on two algorithms: the Measured Continuous Greedy algorithm of~\cite{feldman2011unified}, and a variant of the Continuous Double Greedy of~\cite{buchbinder2014submodular}. We show that these algorithms imply the guarantee of the Theorem~\ref{thm:knapsack_fractional} for $c \leq 1/2$ and $c \geq 1/2$, respectively. Notice that Theorem~\ref{thm:knapsack_fractional} is trivial when $|\cN|$ is a constant (an algorithm that always outputs $\vzero$ obeys its requirements in this case). Thus, we only prove it below for a large enough ground set $\cN$.

We begin with Measured Continuous Greedy. This algorithm has a parameter named $T$ that determines the number of iterations performed by the algorithm.
\begin{proposition} \label{prop:measured_continuous_greedy_density}
When $\cN$ is large enough and $T$ is set to $\ln \eps^{-1} - 1/|\cN|$, Measured Continuous Greedy becomes the algorithm whose existence is guaranteed by Theorem~\ref{thm:knapsack_fractional} for $c \leq 1/2$.
\end{proposition}
\begin{proof}
By Theorem~\ref{thm:measured_continuous_greedy_monotone}, Measured Continuous Greedy has an approximation ratio of
\begin{align*}
	1 - e^{-T} - o(1)
	={} &
	1 - e^{\ln \eps + 1/|\cN|} - o(1)
	=
	1 - \eps \cdot e^{1/|\cN|} - o(1)\\
	={} &
	1 - \eps(1 + o(1)) - o(1)
	=
	1 - \eps - o(1)
	\enspace.
\end{align*}
Technically, Theorem~\ref{thm:measured_continuous_greedy_monotone} assumes that $T$ is a constant. However, as explained in Section~\ref{sec:preliminaries}, this only means that $T$ is required not to grow ``too fast'' with $|\cN|$ (in fact, the proof of \cite{feldman2011unified} for this theorem works as long as $T = o(\cN)$), and thus, assigning $\ln \eps^{-1} - 1/|\cN|$ to $T$ is fine.

Next, we need to prove that the infeasibility ratio of Measured Continuous Greedy is at most $\rho(c, \eps)$ when $T$ is set as specified in this proposition. We first prove it under the assumption that $c$ is bounded away from $0$. By Lemma~III.19 of~\cite{feldman2011unified}, if we denote $\delta \triangleq T(\lceil |\cN|^5 T \rceil)^{-1}$, then the output vector $\vy$ of Measured Continuous Greedy obeys
\begin{align*}
	\inner{\vp}{\vy}
	\leq{} &
	\frac{B}{c} \cdot [1 - e^{-cT} + O(\delta)T]
	\leq
	\frac{B}{c} \cdot [1 - \eps^c \cdot e^{-c/|\cN|} + O(\delta)\ln \eps^{-1}]\\
	\leq{} &
	\frac{B}{c} \cdot \Big[1 - \eps^c - \frac{\eps c}{|\cN|} + O(\delta)\ln \eps^{-1}\Big]
	\leq
	\frac{B}{c} \cdot (1 - \eps^c)
	=
	B \cdot \rho(c, \eps)
	\enspace,
\end{align*}
where the last inequality holds for a large enough $\cN$ since $\eps$ is constant, $c$ is bounded a way from $0$ and
\[
	\delta
	=
	\frac{T}{\lceil |\cN|^5 T \rceil}
	\leq
	\frac{T}{|\cN|^5 T}
	=
	|\cN|^{-5}
	\enspace.
\]

If $c$ is not bounded away from $0$, then the above calculation is not strong enough to prove the proposition. However, it turns out that the proof of Lemma III.19 of \cite{feldman2011unified} can be improved so that the term $O(\delta)T$ is replaced with $O(\delta)cT$ (we refer the reader to Appendix~\ref{app:improved_analysis_unified} for the details). Once this replacement is done, the above calculation works even when $c$ is not bounded away from $0$.
\end{proof}

We now present as Algorithm~\ref{alg:continuous_double_greedy} the variant of the Continuous Double Greedy algorithm of~\cite{buchbinder2014submodular} that we use. There are two main differences between this variant and the original algorithm of~\cite{buchbinder2014submodular}. The first difference is that in the variant the effect of $\ell$ on the treatment of each element $u \in \cN$ is scaled by the cost $p_u$ of $u$. This modification is necessary for supporting general knapsack constraints, rather than only cardinality constraints (as in~\cite{buchbinder2014submodular}). The second difference is that the variant only chooses a non-zero value for $\ell$ when a value of $0$ would result in $\inner{\vp}{\vx(t)}$ growing at a rate larger than $B \cdot \rho(c, \eps)$. This modification allows the algorithm to take advantage of the allowed infeasibility ratio.

As written, Algorithm~\ref{alg:continuous_double_greedy} is a continuous algorithms that runs from time $t = 0$ to time $t = 1$, and determines at every time $t$ within this interval the rate in which the vectors $\vx(t)$ and $\vy(t)$ should change. Of course, it is not possible to directly implement continuous algorithms. Another obstacle for the implementation of Algorithm~\ref{alg:continuous_double_greedy} is that it assumes direct access to the derivatives of the multilinear extension $F$ of the objective function $f$. Fortunately, the (by now standard) techniques of~\cite{calinescu2011maximizing} can be used to solve these problems at the cost of deteriorating the approximation ratio of the algorithm by a factor of $1 - o(1)$ and making it hold only with high probability, which does not affect the guarantee of Theorem~\ref{thm:knapsack_fractional}. For the sake of clarity, we ignore the need to apply these techniques, and directly analyze Algorithm~\ref{alg:continuous_double_greedy}.

\begin{algorithm}
\DontPrintSemicolon
\caption{\textsc{Variant of Continuous Double Greedy}$(f \colon 2^\cN \to \nnR, \vp, B, \eps)$} \label{alg:continuous_double_greedy}
Let $\vx(0) \gets \vzero$ and $\vx(1) \gets \characteristic_\cN$.\\
\For{every time $t \in [0, 1]$}
{
	Define $\va = \nabla F(\vz)|_{\vz = {\vx(t)}}$ and $\vb = - \nabla F(\vz)|_{\vz = {\vy(t)}}$, where $F$ is the multilinear extension of $f$.\\
	Define $\va'(\ell') = (\va - \ell' \cdot \vp) \vee \vzero$ and $\vb' = (\vb + \ell' \cdot \vp) \vee \vzero$.\\
	Let $\ell \geq 0$ be the minimum value such that
	\[
		\inner{\vp}{\frac{\va'(\ell)}{\va'(\ell) + \vb'(\ell)}}
		\leq
		B \cdot \rho(c, \eps)
		\enspace.
	\]\label{line:finding_ell}\\[-1mm]
	\lIf{$\ell = 0$}
	{
		Let $\vd \gets \frac{\va'(\ell)}{\va'(\ell) + \vb'(\ell)}$.\label{line:d_ell_zero}
	}
	\Else
	{
		Let $Z \gets \{u \in \cN \mid a'_u(\ell) + b'_u(\ell) = 0\}$.\\
		Let $r \geq 0$ be the minimum value such that $\big\langle\vp, \frac{\va'(\ell)}{\va'(\ell) + \vb'(\ell)} + r \cdot \characteristic_Z\big\rangle = B \cdot \rho(c, \eps)$.\label{line:r_select}\\
		Let $\vd \gets \frac{\va'(\ell)}{\va'(\ell) + \vb'(\ell)} + r \cdot \characteristic_Z$.\label{line:d_ell_non-zero}
	}
	Set $\frac{d\vx(t)}{dt} = \vd$, and $\frac{d\vy(t)}{dt} = \vd - \characteristic_\cN$.
}
\Return $\vx(1)$.
\end{algorithm}

Algorithm~\ref{alg:continuous_double_greedy} involves divisions between vectors. These divisions should be understood as coordinate-wise. Furthermore, when the denominator is $0$ in such a division for some coordinate, we define the output of the division to be zero as well for this coordinate.
To see that Algorithm~\ref{alg:continuous_double_greedy} is well-defined, we need the following two lemmata.

\begin{lemma} 
There is a value for $\ell$ that obeys the definition on Line~\ref{line:finding_ell} of Algorithm~\ref{alg:continuous_double_greedy}.
\end{lemma}
\begin{proof}
Notice that for a large enough values of $\ell$, the ratio $\frac{a'_u(\ell)}{a'_u(\ell) + b'_u(\ell)}$ is equal to $0$ for every element $u \in \cN$ whose cost $p_u$ is positive. Therefore, there are non-negative values of $\ell$ for which the inequality $\big\langle\vp' \frac{\va'(\ell)}{\va'(\ell) + \vb'(\ell)}\big\rangle = 0 \leq B \cdot \rho(c, \eps)$ holds. Thus, to prove the lemma, we only need to show that the set of $\ell \geq 0$ values for which the inequality on Line~\ref{line:finding_ell} of Algorithm~\ref{alg:continuous_double_greedy} holds has a minimum. One can observe that this will follow if we prove that the inner product $\big\langle\vp, \frac{\va'(\ell)}{\va'(\ell) + \vb'(\ell)}\big\rangle$ is a non-increasing function of $\ell$ that is continuous to the right at all points.

To prove that this is indeed the case, we need to analyze, for every $u \in \cN$, the properties of the ratio $\frac{a'_u(\ell)}{a'_u(\ell) + b'_u(\ell)}$ as a function of $\ell$.
\begin{itemize}
	\item First, we show that when $p_u > 0$ the ratio $\frac{a'_u(\ell)}{a'_u(\ell) + b'_u(\ell)}$ is continuous to the right for all values of $\ell$. For $\ell < a_u/p_u$, the ratio $\frac{a'_u(\ell)}{a'_u(\ell) + b'_u(\ell)}$ is continuous (to the left and the right) because the denominator $a'_u(\ell) + b'_u(\ell)$ is positive for all such values of $\ell$. For $\ell \geq a_u/p_u$, the ratio $\frac{a'_u(\ell)}{a'_u(\ell) + b'_u(\ell)}$ is equal to $0$ because $a'_u(\ell) = 0$ for such values of $0$, and therefore, this ratio is continues for all $\ell > a_u$, and continuous to the right at $\ell = a_\ell$.
	\item Second, we show that when $p_u > 0$ the ratio $\frac{a'_u(\ell)}{a'_u(\ell) + b'_u(\ell)}$ is a non-increasing function. 
	Like in the previous bullet, we note that for $\ell \geq a_u/p_u$, this ratio is $0$. For smaller values of $\ell$,
	\[
		\frac{a'_u(\ell)}{a'_u(\ell) + b'_u(\ell)}
		=
		\frac{a_u - \ell p_u}{a_u - \ell p_u + \max\{b_u + \ell p_u, 0\}}
		=
		\begin{cases}
			1 & \text{if $\ell \leq -b_u/p_u$} \enspace,\\
			\frac{a_u - \ell p_u}{a_u + b_u} & \text{if $\ell \in (-b_u/p_u, a_u/p_u)$} \enspace;
		\end{cases}
	\]
	which indeed proves that $\frac{a'_u(\ell)}{a'_u(\ell) + b'_u(\ell)}$ is a non-increasing function of $\ell$ because when the range $(-b_u/p_u, a_u/p_u)$ is non-empty (i.e., $a_u > -b_u$), the ratio $\frac{a_u - \ell p_u}{a_u + b_u}$ is a decreasing function of $\ell$ that decreases from $1$ for $\ell = -b_u/p_u$ to $0$ for $\ell = a_u/p_u$.
\end{itemize}
Since the vector $\vp$ is non-negative, the inner product $\big\langle\vp, \frac{\va'(\ell)}{\va'(\ell) + \vb'(\ell)}\big\rangle$ is also a non-increasing function of $\ell$ that is continuous to the right at all points, which completes the proof of the lemma as explained above.
\end{proof}

\begin{lemma} \label{lem:r_value}
Line~\ref{line:r_select} of Algorithm~\ref{alg:continuous_double_greedy} chooses a value $r \in [0, 1]$ whenever it is executed.
\end{lemma}
\begin{proof}
By the definition of $\ell$, $\big\langle\vp, \frac{\va'(\ell)}{\va'(\ell) + \vb'(\ell)}\big\rangle \leq B \cdot \rho(c, \eps)$. Below we prove that $\big\langle\vp, \frac{\va'(\ell)}{\va'(\ell) + \vb'(\ell)} + \characteristic_Z\big\rangle \geq B \cdot \rho(c, \eps)$. Since $\big\langle\vp, \frac{\va'(\ell)}{\va'(\ell) + \vb'(\ell)} + r \cdot \characteristic_Z\big\rangle$ is a continuous function of $r$, this will imply that the value chosen for $r$ by Line~\ref{line:r_select} is indeed within the range $[0, 1]$.

Since Algorithm~\ref{alg:continuous_double_greedy} has reached Line~\ref{line:r_select}, we are guaranteed that $\ell > 0$. Therefore, by the choice of $\ell$, for every $\hat{\ell} \in (0, \ell)$, it holds that
\begin{equation} \label{eq:ell_guarantee}
	\inner{\vp}{\frac{\va'(\hat{\ell}) \wedge \characteristic_{\cN \setminus Z}}{\va'(\hat{\ell}) + \vb'(\hat{\ell})} + \characteristic_Z}
	\geq
	\inner{\vp}{\frac{\va'(\hat{\ell})}{\va'(\hat{\ell}) + \vb'(\hat{\ell})}}
	>
	B \cdot \rho(c, \eps)
	\enspace.
\end{equation}
where the first inequality holds since $\frac{\va'(\hat{\ell})}{\va'(\hat{\ell}) + \vb'(\hat{\ell})} \leq \characteristic_\cN$ and $\vp$ is non-negative. Notice now that for every $u \in \cN \setminus Z$, the ratio $\frac{a'_u(\hat{\ell})}{a'_u(\hat{\ell}) + b'_u(\hat{\ell})}$ is a continuous function of $\hat{\ell}$ at the point $\ell$. Therefore,
\[
	\lim_{\hat{\ell} \to \ell} \frac{a'_u(\hat{\ell})}{a'_u(\hat{\ell}) + b'_u(\hat{\ell})}
	=
	\frac{a'_u(\ell)}{a'_u(\ell) + b'_u(\ell)}
	\enspace,
\]
which implies that we also have
\[
	\lim_{\hat{\ell} \to \ell} \inner{\vp}{\frac{\va'(\hat{\ell}) \wedge \characteristic_{\cN \setminus Z}}{\va'(\hat{\ell}) + \vb'(\hat{\ell})} + \characteristic_Z}
	=
	\inner{\vp}{\frac{\va'(\ell) \wedge \characteristic_{\cN \setminus Z}}{\va'(\ell) + \vb'(\ell)} + \characteristic_Z}
	=
	\inner{\vp}{\frac{\va'(\ell)}{\va'(\ell) + \vb'(\ell)} + \characteristic_Z}
	\enspace,
\]
where the second equality holds since $a'_u(\ell) = 0$ for every $u \in Z$. Since the inner product whose limit we have calculated is always larger than $B \cdot \rho(c, \eps)$ by Inequality~\eqref{eq:ell_guarantee}, the limit must also be at least $B \cdot \rho(c, \eps)$. Thus,
\[
	\inner{\vp}{\frac{\va'(\ell)}{\va'(\ell) + \vb'(\ell)} + \characteristic_Z}
	=
	\lim_{\hat{\ell} \to \ell} \inner{\vp}{\frac{\va'(\hat{\ell}) \wedge \characteristic_{\cN \setminus Z}}{\va'(\hat{\ell}) + \vb'(\hat{\ell})} + \characteristic_Z}
	\geq
	B \cdot \rho(c, \eps)
	\enspace,
\]
which completes the proof of the lemma.
\end{proof}

\begin{corollary} \label{cor:gap_and_feasibility}
At every time $t \in [0, 1]$, the vector $\vd$ belongs to $[0, 1]^\cN$ and obeys $\inner{\vp}{\vd} \leq B \cdot \rho(c, \eps)$,  and therefore:
\begin{itemize}
	\item at every time $t$, $\vzero \leq \vx(t) \leq \vx(t) + (1 - t) \cdot \characteristic_\cN = \vy(t) \leq \characteristic_\cN$.
	\item the output vector $\vx(1) = \vy(1)$ of Algorithm~\ref{alg:continuous_double_greedy} obeys $\inner{\vp}{\vx(1)} \leq B \cdot \rho(c, \eps)$, and thus, the infeasibility ratio of Algorithm~\ref{alg:continuous_double_greedy} is at most $\rho(c, \eps)$.
\end{itemize}
\end{corollary}
\begin{proof}
We need to consider two cases. If $\vd$ is set by Line~\ref{line:d_ell_zero} of Algorithm~\ref{alg:continuous_double_greedy}, then it clearly obeys $\vd \in [0, 1]^\cN$. Furthermore, in this case we also have $\inner{\vp}{\vd} \leq B \cdot \rho(c, \eps)$ by the definition of $\ell$.

Next, we consider the case in which $\vd$ is set by Line~\ref{line:d_ell_non-zero} of Algorithm~\ref{alg:continuous_double_greedy}. In this case $\inner{\vp}{\vd} = B \cdot \rho(c, \eps)$ by the definition of $r$, and $\vd$ is set to
\[
	\frac{\va'(\ell)}{\va'(\ell) + \vb'(\ell)} + r \cdot \characteristic_Z
	\enspace.
\]
This expression is a vector in $[0, 1]^\cN$ because (i) Lemma~\ref{lem:r_value} guarantees that $r \in [0, 1]$, and (ii) $\frac{\va'_u(\ell)}{\va'_u(\ell) + \vb'_u(\ell)} = 0$ for every $u \in Z$ by our definition of the value resulting from a division by $0$.
\end{proof}

It remains to analyze the approximation ratio of Algorithm~\ref{alg:continuous_double_greedy}. Towards this goal, we need to define $\opt(\vx, \vy) \triangleq (\characteristic_{\opt} \vee \vx) \wedge \vy$. One can observe that $\opt(\vx(0), \vy(0)) = \characteristic_{\opt}$, and $\opt(\vx(1), \vy(1)) = \vx(1)$ because $\vx(1) = \vy(1)$. Thus, one way to bound the approximation ratio of Algorithm~\ref{alg:continuous_double_greedy} is by bounding the rate by which $F(\opt(\vx(t), \vy(t)))$ decreases as a function of $t$. Lemma~\ref{lem:opt_derivative} does that. However, before getting to this lemma, we first need to prove the following observation. This observation is useful since $\vx(t) < \vy(t)$ for every $t \in [0, 1)$ by Corollary~\ref{cor:gap_and_feasibility}.

\begin{observation} \label{obs:opt_bounds}
For every two vectors $\vx, \vy \in [0, 1]^\cN$ that obey $\vx < \vy$ and element $u \in \cN$, if $u \not \in \opt$, then
\[
	\frac{\partial F(\opt(\vx, \vy))}{\partial x_u}
	\geq
	\frac{\partial F(\vy)}{\partial y_u}
	\qquad\text{and}\qquad
	\frac{\partial F(\opt(\vx, \vy))}{\partial y_u}
	=
	0
	\enspace,
\]
and if $u \in \opt$, then
\[
	\frac{\partial F(\opt(\vx, \vy))}{\partial x_u}
	=
	0
	\qquad\text{and}\qquad	
	\frac{\partial F(\opt(\vx, \vy))}{\partial y_u}
	\leq
	\frac{\partial F(\vx)}{\partial x_u}
	\enspace.
\]
\end{observation}
\begin{proof}
We prove the inequality only for $u \not \in \opt$. The proof for the other case is analogous. Since $u \not \in \opt$ and $x_u < y_u$, the value of $\opt_u(\vx, \vy)$ is equal to $x_u$, and it remains $x_u$ even if the value of $y_u$ changes by a small amount (as long as this amount is less than $y_u - x_u > 0$). Thus, $\frac{\partial F(\opt(\vx, \vy))}{\partial y_u} = 0$. On the other hand, changing $x_u$ by $\Delta$ will result in a change of $\opt_u(\vx, \vy)$ by $\Delta$ as long as $|\Delta| \leq y_u - x_u$. Therefore, by the submodularity of $f$,
\[
	\frac{\partial F(\opt(\vx, \vy))}{\partial x_u}
	=
	\left.\frac{\partial F(\vz)}{\partial z_u}\right|_{\vz = \opt(\vx, \vy)}
	\geq
	\frac{\partial F(\vy)}{\partial y_u}
	\enspace.
	\qedhere
\]
\end{proof}

\begin{lemma} \label{lem:opt_derivative}
If $c \geq 1/2$, then at every time $t \in [0, 1)$,
\[
	\frac{dF(\opt(\vx(t), \vy(t)))}{dt}
	\geq
	-\frac{\sqrt{\eps}}{2(1 - \sqrt{\eps})} \cdot \frac{dF(\vx(t))}{dt} - \frac{1 - \sqrt{\eps}}{2\sqrt{\eps}} \cdot \frac{dF(\vy(t))}{dt}
	\enspace.
\]
\end{lemma}
\begin{proof}
By the chain rule,
\begin{align*}
	\frac{dF(\opt(\vx(t), \vy(t)))}{dt}
	={} &
	\sum_{u \in \cN} \Big[\frac{d x_u(t)}{dt} \cdot \frac{\partial F(\opt(\vx(t), \vy(t)))}{\partial x_u(t)} + \frac{d y_u(t)}{dt} \cdot \frac{\partial F(\opt(\vx(t), \vy(t)))}{\partial y_u(t)}\Big]\\
	\geq{} &
	\sum_{u \not \in \opt} \frac{d x_u(t)}{dt} \cdot \frac{\partial F(\vy(t))}{\partial y_u(t)} + \sum_{u \in \opt} \frac{d y_u(t)}{dt} \cdot \frac{\partial F(\vx(t))}{\partial x_u(t)}\\
	={} &
	-\sum_{u \not \in \opt} \frac{d x_u(t)}{dt} \cdot b_u + \sum_{u \in \opt} \frac{d y_u(t)}{dt} \cdot a_u
	\enspace,
\end{align*}
where the inequality follows from Observation~\ref{obs:opt_bounds} because $\frac{d x_u(t)}{dt} \geq 0$ and $\frac{d y_u(t)}{dt} \leq 0$ for every $u \in \cN$. We now observe that
\begin{align*}
	\ell \cdot \Big[\sum_{u \not \in \opt} \frac{d x_u(t)}{dt} \cdot p_u + \sum_{u \in \opt} \frac{d y_u(t)}{dt} \cdot p_u\Big]
	={} &
	\ell \cdot \Big[\sum_{u \not \in \opt} d_u \cdot p_u + \sum_{u \in \opt} (d_u - 1) \cdot p_u\Big]\\
	={} &
	\ell \cdot [\inner{\vp}{\vd} - \inner{\vp}{\characteristic_{\opt}}]
	\geq
	\ell B \cdot (\rho(c, \eps) - 1)
	\enspace,
\end{align*}
where the inequality is trivial when $\ell = 0$, and holds when $\ell > 0$ because in this case $\inner{\vp}{\vd} = B \cdot \rho(c, \eps)$. Combining this inequality and the previous one now yields
\begin{multline} \label{eq:right_side_for_development}
	\ell B \cdot (\rho(c, \eps) - 1) - \frac{dF(\opt(\vx(t), \vy(t))}{dt}\\
	\leq
	\sum_{u \not \in \opt} \frac{d x_u(t)}{dt} \cdot (b_u + p_u \ell) - \sum_{u \in \opt} \frac{d y_u(t)}{dt} \cdot (a_u - p_u \ell)
	\enspace.
\end{multline}

Our next goal is to upper bound the right hand side of Inequality~\eqref{eq:right_side_for_development}. We bound the contribution of each element to this side separately. Specifically, for every element $u \in \cN$,
\begin{itemize}
	\item if $p_u \ell \geq a_u$ and $p_u \ell > -b_u$, then $\frac{d x_u(t)}{dt} \cdot (b_u + p_u \ell)$ = $0 \cdot (b_u + p_u \ell)$ and $-\frac{d y_u(t)}{dt} \cdot (a_u - p_u \ell) = a_u - p_u \ell \leq 0$.
	\item if $p_u \ell < a_u$ and $p_u \ell \leq -b_u$, then $\frac{d x_u(t)}{dt} \cdot (b_u + p_u \ell)$ = $b_u + p_u \ell \leq 0$ and $-\frac{d y_u(t)}{dt} \cdot (a_u - p_u \ell) = 0 \cdot (a_u - p_u \ell) = 0$.
	\item if $p_u \ell \geq a_u$ and $p_u \ell \leq -b_u$, then $\frac{d x_u(t)}{dt} \cdot (b_u + p_u \ell) \leq 0$ because $\frac{d x_u(t)}{dt} \geq 0$ and $b_u + p_u \ell \leq 0$. Similarly, $-\frac{d y_u(t)}{dt} \cdot (a_u - p_u \ell) \leq 0$ because $-\frac{d y_u(t)}{dt} \geq 0$ and $a_u - p_u \ell \leq 0$.
	\item if $p_u \ell < a_u$ and $p_u \ell > -b_u$, then $\frac{d x_u(t)}{dt} \cdot (b_u + p_u \ell) = \frac{a_u - p_u \ell}{a_u + b_u} \cdot (b_u + p_u \ell) = \frac{(a_u - p_u \ell)(b_u + p_u \ell)}{a_u + b_u}$ and $-\frac{d y_u(t)}{dt} \cdot (b_u + p_u \ell) = \big(1 - \frac{a_u - p_u \ell}{a_u + b_u}\big) \cdot (a_u - p_u \ell) = \frac{b_u + p_u \ell}{a_u + b_u} \cdot (a_u - p_u \ell) = \frac{(a_u - p_u \ell)(b_u + p_u \ell)}{a_u + b_u}$.
\end{itemize}
Plugging all these bounds into Inequality~\eqref{eq:right_side_for_development}, we get
\begin{align} \label{eq:right_side_for_development_2}
	\ell B \cdot (\rho(c, \eps) - 1) - \frac{dF(\opt(\vx(t), \vy(t))}{dt}
	\leq{} &
	\sum_{u \in \cN} \max\left\{0, \frac{(a_u - p_u \ell)(b_u + p_u\ell)}{a_u + b_u}\right\}\\\nonumber
	={} &
	\sum_{u \in \cN} \max\left\{0, \frac{s(a_u - p_u \ell)}{\sqrt{a_u + b_u}} \cdot \frac{s^{-1}(b_u + p_u\ell)}{\sqrt{a_u + b_u}}\right\}
	\enspace,
\end{align}
where $s$ is some positive value to be determined later. Notice that in this equality we assume, as usual, that the value of a division by $0$ is $0$.

Once again we need to upper bound the right hand side of the last inequality, and we do so by individually upper bounding the term in it corresponding to every element $u \in \cN$. Specifically, we claim that
\begin{equation} \label{eq:element_upper_bound}
	2\max\left\{0, \frac{s(a_u - p_u \ell)}{\sqrt{a_u + b_u}} \cdot \frac{s^{-1}(b_u + p_u\ell)}{\sqrt{a_u + b_u}}\right\}
	\leq
	s^2(a_u - p_u\ell) \cdot \frac{d x_u(t)}{dt} - s^{-2}(b_u + p_u\ell) \cdot \frac{d y_u(t)}{dt}
	\enspace.
\end{equation}
The right hand side of this inequality is always non-negative since $\frac{\partial x_u(t)}{dt} = 0$ whenever $p_u \ell > a_u$ and $\frac{\partial y_u(t)}{dt} = 0$ whenever $p_u \ell < -b_u$. Thus, Inequality~\eqref{eq:element_upper_bound} holds when its left hand side is $0$, and, it only remains to prove it in the case that
\[
	\frac{s(a_u - p_u \ell)}{\sqrt{a_u + b_u}} \cdot \frac{s^{-1}(b_u + p_u\ell)}{\sqrt{a_u + b_u}}
	>
	0
	\enspace.
\]
By the submodularity of $f$,
\[
	(a_u - p_u \ell) + (b_u + p_u \ell)
	=
	a_u + b_u
	=
	\frac{\partial F(\vx(t))}{\partial x_u(t)} - \frac{\partial F(\vy(t))}{\partial y_u(t)}
	\geq
	0
\]
because $\vx(t) \leq \vy(t)$. When combined with the previous inequality, this inequality implies that $a_u - p_u \ell$ and $b_u + p_u \ell$ are both positive in the case we consider, and thus, equal to $a'_u(\ell)$ and $b'_u(\ell)$, respectively. Hence, we get
\begin{align*}
	2\frac{s(a_u - p_u \ell)}{\sqrt{a_u + b_u}} \cdot \frac{s^{-1}(b_u + p_u\ell)}{\sqrt{a_u + b_u}}
	\leq{} &
	\frac{s^2(a_u - p_u \ell)^2}{a_u + b_u} + \frac{s^{-2}(b_u + p_u \ell)^2}{a_u + b_u}\\
	={} &
	s^2(a_u - p_u \ell) \cdot \frac{a'_u(\ell)}{a'_u(\ell) + b'_u(\ell)} + s^{-2}(b_u + p_u \ell) \cdot \Big[1 - \frac{a'_u(\ell)}{a'_u(\ell) + b'_u(\ell)}\Big]\\
	={} &
	s^2(a_u - p_u \ell) \cdot \frac{d x_u(t)}{dt} - s^{-2}(b_u + p_u \ell) \cdot \frac{dy_u(t)}{dt}
	\enspace,
\end{align*}
where the inequality holds since $2ab \leq a^2 + b^2$ for every two real numbers $a$ and $b$. This completes the proof of Inequality~\eqref{eq:element_upper_bound}.

Plugging Inequality~\eqref{eq:element_upper_bound} into Inequality~\eqref{eq:right_side_for_development_2}, we now get
\begin{align} \label{eq:final_with_ell}
	\ell B \cdot (\rho(c, \eps) - 1) -{}& \frac{dF(\opt(\vx(t), \vy(t))}{dt}
	\leq
	\frac{1}{2}\sum_{u \in \cN} \Big[s^2(a_u - p_u \ell) \cdot \frac{d x_u(t)}{dt} - s^{-2}(b_u + p_u \ell) \cdot \frac{dy_u(t)}{dt}\Big]\\\nonumber
	={} &
	\frac{s^2}{2} \cdot \frac{dF(\vx(t))}{dt} + \frac{s^{-2}}{2} \cdot \frac{dF(\vy(t))}{dt} - \frac{\ell}{2} \cdot [s^2 \cdot \inner{\vp}{\vd} + s^{-2} \cdot \inner{\vp}{\vd - \characteristic_\cN}]\\\nonumber
	={} &
	\frac{s^2}{2} \cdot \frac{dF(\vx(t))}{dt} + \frac{s^{-2}}{2} \cdot \frac{dF(\vy(t))}{dt} - \frac{\ell}{2} \cdot [(s^2 + s^{-2}) \cdot B \cdot \rho(c, \eps) - s^{-2} \cdot B/c]
	\enspace,
\end{align}
where the first equality holds by the chain rule, and the second equality holds since $\inner{\vp}{\vd} = B \cdot \rho(c, \eps)$ whenever $\ell > 0$.

Recall that $s$ is an arbitrary positive value. To make the last inequality useful, we would like to pick a value for $s$ that will make the terms involving $\ell$ on both sides of this inequality cancel each other out. In other words, $s$ should be a solution to the equation
\[
	\rho(c, \eps) - 1
	=
	-\frac{(s^2 + s^{-2}) \cdot \rho(c, \eps) - s^{-2}/c}{2}
	\Leftrightarrow
	s^4 \cdot \rho(c, \eps) + s^2 \cdot 2(\rho(c, \eps) - 1) + \rho(c, \eps) - 1/c
	=
	0
	\enspace.
\]
The solution for this inequality is
\begin{align*}
	s^2
	={} &
	\frac{2(1 - \rho(c, \eps)) + \sqrt{4(\rho(c, \eps) - 1)^2 - 4\rho(c, \eps) \cdot (\rho(c, \eps) - 1/c)}}{2\rho(c, \eps)}\\
	={} &
	\frac{1 - \rho(c, \eps) + \sqrt{1 + \rho(c, \eps) \cdot (1/c - 2)}}{\rho(c, \eps)}
	=
	\frac{1 - \rho(c, \eps) + \sqrt{1 + \frac{1 - 2(1 - c)\sqrt{\eps} - \eps(2c - 1)}{c} \cdot (1/c - 2)}}{\rho(c, \eps)}\\
	={} &
	\frac{c(1 - \rho(c, \eps)) + \sqrt{(1 - c)^2 + 2(1 - c)(2c - 1)\sqrt{\eps} + \eps(2c - 1)^2}}{c \cdot \rho(c, \eps)}\\
	={} &
	\frac{c(1 - \rho(c, \eps)) + (1 - c) + (2c - 1)\sqrt{\eps}}{c \cdot \rho(c, \eps)}
	=
	\frac{\sqrt{\eps} + \eps(2c - 1)}{1 - 2(1 - c)\sqrt{\eps} - \eps(2c - 1)}
	=
	\frac{\sqrt{\eps}}{1 - \sqrt{\eps}}
	\enspace,
\end{align*}
where the third equality holds by plugging in the value of $\rho(c, \eps)$ (for $c \geq 1/2$) into the single appearance of $\rho(c, \eps)$ under the square root, and the penultimate equality holds by plugging in this value again into the remaining appearances of $\rho(c, \eps)$. The lemma now follows by plugging this value of $s^2$ into Inequality~\eqref{eq:final_with_ell}.
\end{proof}

\begin{corollary}
For $c \geq 1/2$, the approximation ratio of Algorithm~\ref{alg:continuous_double_greedy} is at least $1 - \eps$.
\end{corollary}
\begin{proof}
Integrating the guarantee of Lemma~\ref{lem:opt_derivative} from $t = 0$ until $t = 1$ yields
\begin{multline*}
	F(\opt(\vx(1), \vy(1)) - F(\opt(\vx(0), \vy(0))
	\geq
	-\frac{\sqrt{\eps}}{2(1 - \sqrt{\eps})} \cdot [F(\vx(1)) - F(\vx(0))] \\- \frac{1 - \sqrt{\eps}}{2\sqrt{\eps}} \cdot [F(\vy(1)) - F(\vy(0))]
	\enspace.
\end{multline*}
Recall now that $\opt(\vx(1), \vy(1)) = \vx(1) = \vy(1)$, $F(\opt(\vx(0), \vy(0)) = F(\characteristic_{\opt}) = f(\opt)$, $F(\vx(0)) = F(\vzero) = f(\varnothing)$ and $F(\vy(0)) = F(\characteristic_\cN) = f(\cN)$. Plugging these observations into the above inequality, we get
\[
	F(\vx(1)) - f(\opt)
	\geq
	-\frac{\sqrt{\eps}}{2(1 - \sqrt{\eps})} \cdot [F(\vx(1)) - f(\varnothing)] - \frac{1 - \sqrt{\eps}}{2\sqrt{\eps}} \cdot [F(\vx(1)) - f(\cN)]
	\enspace,
\]
and rearranging this inequality gives
\begin{align*}
	F(x(1))
	\geq{} &
	\frac{f(\opt) + \frac{\sqrt{\eps}}{2(1 - \sqrt{\eps})} \cdot f(\varnothing) + \frac{1 - \sqrt{\eps}}{2\sqrt{\eps}} \cdot f(\cN)}{1 + \frac{\sqrt{\eps}}{2(1 - \sqrt{\eps})} + \frac{1 - \sqrt{\eps}}{2\sqrt{\eps}}}\\
	={} &
	\frac{2(\sqrt{\eps} - \eps) \cdot f(\opt) + \eps \cdot f(\varnothing) + (1 - 2\sqrt{\eps} + \eps) \cdot f(\cN)}{2(\sqrt{\eps} - \eps) + \eps + (1 - 2\sqrt{\eps} + \eps)}\\
	={} &
	2(\sqrt{\eps} - \eps) \cdot f(\opt) + \eps \cdot f(\varnothing) + (1 - 2\sqrt{\eps} + \eps) \cdot f(\cN)
	\enspace.
\end{align*}
We now observe that the non-negativity and monotonicity of $f$ imply that $f(\varnothing) \geq 0$ and $f(\cN) \geq f(\opt)$, respectively. Plugging these inequalities into the previous one complete the proof of the corollary.
\end{proof}

\subsubsection{Proof of Theorem~\texorpdfstring{\ref{thm:monotone_cardinality_inapproximability}}{\ref*{thm:monotone_cardinality_inapproximability}} for \texorpdfstring{$c \leq 1/2$}{c < 1/2}} \label{sssc:cardinality_inapproximability_small_c}

In this section, we prove Theorem~\ref{thm:monotone_cardinality_inapproximability} for $c \leq 1/2$. The proof of Theorem~\ref{thm:monotone_cardinality_inapproximability} is based on two functions $F(x, y)$ and $G(x, y)$ defined over the domain $\cD \triangleq [0, c] \times [0, 1 - c]$. The function $G(x, y)$ is defined by
\[
	G(x, y)
	\triangleq
	1 - (1 - x - y)^{1/c}
	\enspace.
\]
The function $F(x, y)$ has a more involved definition. Let $\delta' \in (0, c/4)$ be a constant to be determined later. Then, for $y < (1 - \delta')(1 - c)$, $F(x, y)$ is given by
\[
	F(x, y)
	\triangleq
	1 - \bigg(1 - \max\bigg\{0, \frac{x - \frac{c}{1 - c}y - \delta'}{c\big(1 - \frac{1}{1 - c}y - \delta'\big)}\bigg\}\bigg)\bigg(1 - \min\Big\{x + y, \frac{1}{1 - c}y + \delta'\Big\}\bigg)^{1/c}
	\enspace,
\]
and for larger values of $y$, it is defined as equal to $G(x, y)$.

\begin{observation} \label{obs:function_equality}
The function $F(x, y)$ 
is equal to $G(x, y)$ whenever $x \leq \frac{c}{1 - c}y + \delta'$.
\end{observation}
\begin{proof}
When $y \geq (1 - \delta')(1 - c)$, $F(x, y) = G(x, y)$ by definition; and when $y < (1 - \delta')(1 - c)$ and $x \leq \frac{c}{1 - c}y + \delta'$,
\begin{align*}
	F(x, y)
	={} &
	1 - \bigg(1 - \max\bigg\{0, \frac{x - \frac{c}{1 - c}y - \delta'}{c\big(1 - \frac{1}{1 - c}y - \delta'\big)}\bigg\}\bigg)\bigg(1 - \min\Big\{x + y, \frac{1}{1 - c}y + \delta'\Big\}\bigg)^{1/c}\\
	={} &
	1 - (1 - 0) \cdot (1 - (x + y))^{1/c}
	=
	G(x, y)
	\enspace.
	\qedhere
\end{align*}
\end{proof}

We now prove some properties of $G$. Notice that since $G$ is defined over the bounded domain $\cD$, it cannot have partial derivatives everywhere. Nevertheless, the following observation shows that it has partial derivatives wherever this is possible, and also proves some properties of these partial derivatives.
\begin{observation} \label{obs:G_properties}
The function $G$ is non-negative and upper bounded by $1$. It has a partial derivative with respect to $x$ at every point of $(x, y) \in \cD$ for which $x \in (0, c)$, and it has a partial derivative with respect to $y$ at every point of $(x, y) \in \cD$ for which $y \in (0, 1 - c)$. Moreover, whenever these derivatives are defined, they are non-negative and non-increasing in both $x$ and $y$ .
\end{observation}
\begin{proof}
It is not difficult to verify that $G$ is non-negative and upper bounded by $1$. For every point $(x, y)$ in $\cD$, the entire interval $\{(x', y) \mid x' \in (0, c)\}$ is included in $\cD$. If $x \in (0, c)$, the point $(x, y)$ is included inside this interval, and therefore, $G$ has a derivative with respect to $x$ at this point. Similarly, for every point $(x, y)$ in $\cD$, the entire interval $\{(x, y') \mid y' \in (0, 1 - c)\}$ is included in $\cD$. If $y \in (0, 1 - c)$, the point $(x, y)$ is included inside this interval, and therefore, $G$ has a derivative with respect to $y$ at this point.

To complete the proof of the observation, it remains to note that whenever the partial derivatives of $G$ with respect to $x$ and/or $y$ exist at a point $(x, y)$, they are equal to
\[
	\frac{(1 - x - y)^{1/c-1}}{c}
	\enspace.
\]
Since $c \in (0, 1/2]$, this expression is non-negative and non-increasing in $x$ and $y$.
\end{proof}

The following three lemmata prove together that $F$ also has the properties stated by Observation~\ref{obs:G_properties} for $G$.

\begin{lemma} \label{lem:F_non-negative}
The function $F$ is non-negative and upper bounded by $1$.
\end{lemma}
\begin{proof}
For $y \geq (1 - \delta')(1 - c)$, the non-negativity of $F(x, y)$ and the fact that $F(x, y) \leq 1$ both follow from Observation~\ref{obs:G_properties}. For smaller values of $y$,
\[
	F(x, y)
	=
	1 - \bigg(1 - \max\bigg\{0, \frac{x - \frac{c}{1 - c}y - \delta'}{c\big(1 - \frac{1}{1 - c}y - \delta'\big)}\bigg\}\bigg)\bigg(1 - \min\Big\{x + y, \frac{1}{1 - c}y + \delta'\Big\}\bigg)^{1/c}
	\enspace.
\]
This expression is non-negative and upper bounded by $1$ since the two terms in the product always take values between $0$ and $1$. To see that this is indeed the case, notice that
\[
	x - \frac{c}{1 - c}y - \delta'
	\leq
	c - \frac{c}{1 - c}y - \delta'
	\leq
	c\Big(1 - \frac{1}{1 - c}y - \delta'\Big)
\]
because $x \leq c$ in the domain $\cD$ of $F(x, y)$, and
\[
	0
	\leq
	\min\Big\{x + y, \frac{1}{1 - c}y + \delta'\Big\}
	\leq
	x + y
	\leq
	1
\]
because in the domain $\cD$ we are guaranteed that $x + y \in [0, 1]$ and $y$ is non-negative.
\end{proof}

\begin{lemma} \label{lem:F_derivative_x}
The function $F(x, y)$ has a partial derivative with respect to $x$ at every point $(x, y) \in \cD$ for which $x \in (0, c)$. Furthermore, when it is defined, this partial derivative is non-negative and non-increasing in $x$ and $y$.
\end{lemma}
\begin{proof}
Consider an arbitrary point $(x, y) \in \cD$ for which $x \in (0, c)$. If $y \geq (1 - \delta')(1 - c)$, then $F(x', y) = G(x', y)$ in every point $(x', y) \in \cD$. Thus, $F$ and $G$ have the same partial derivative with respect to $x$ at the point $(x, y)$. According to the proof of Observation~\ref{obs:G_properties}, this partial derivative exists, and is given by
\[
	\frac{\partial F(x, y)}{\partial x}
	=
	\frac{(1 - x - y)^{1/c-1}}{c}
	\enspace.
\]

Consider now the case that $y < (1 - \delta')(1 - c)$. Notice that the entire interval $\{(x', y) \mid x \in (0, c)\}$ is included inside $\cD$, and $(x, y)$ is a point insider this interval. In this case, the mathematical formula defining $F(x, y)$ for points in this interval involves $\min$ and $\max$ operators. Since both of these operators switch from outputting one argument to the next when $x = \frac{b}{1 - b}y + \delta'$, we get
\[
	\frac{\partial F(x, y)}{\partial x}
	=
	\begin{cases}
		\frac{(1 - x - y)^{1/c - 1}}{c} & \text{for $x < \frac{c}{1 - c}y + \delta'$} \enspace,\\
		\frac{\left(1 - \frac{1}{1 - c}y - \delta'\right)^{1/c - 1}}{c}& \text{for $x > \frac{c}{1 - c}y + \delta'$} \enspace.
	\end{cases}
\]
Notice that the expressions given above for the two cases are equal when $x = \frac{c}{1 - c}y + \delta'$. Thus, the partial derivative of $F(x, y)$ with respect to $x$ exists also at the point $(\frac{c}{1 - c}y, y)$, and its value can be calculated using the expressions of either of the above cases.

To summarize, the above discussion shows that $F(x, y)$ has a partial derivative with respect to $x$ at every point $(x, y) \in \cD$ for which $x \in (0, c)$. Furthermore, since $x \leq \frac{c}{1 - c}y + \delta'$ whenever $y \geq (1 - \delta')(1 - c)$, the value of this partial derivative is always given by
\[
	\frac{\partial F(x, y)}{\partial x}
	=
	\begin{cases}
		\frac{(1 - x - y)^{1/c - 1}}{c} & \text{for $x \leq \frac{c}{1 - c}y + \delta'$} \enspace,\\
		\frac{\left(1 - \frac{1}{1 - c}y - \delta'\right)^{1/c - 1}}{c}& \text{for $x \geq \frac{c}{1 - c}y + \delta'$} \enspace.
	\end{cases}
\]
The two expressions given for the above two cases are individually non-negative and non-increasing functions of $x$ and $y$ since $c \in (0, 1/2)$ and $x \leq c$. This immediately implies that the partial derivative of $F$ with respect to $x$ is non-negative whenever it is defined. Since the two expressions of the two cases are equal in the borderline between them (when $x = \frac{c}{1 - c}y + \delta'$), we also get that the partial derivative of $F$ with respect to $x$ is a non-increasing function of $x$ and $y$.
\end{proof}

\begin{lemma} \label{lem:F_derivative_y}
The function $F(x, y)$ has a partial derivative with respect to $y$ at every point $(x, y) \in \cD$ for which $y \in (0, 1-c)$. Furthermore, when it is defined, this partial derivative is non-negative and non-increasing in $x$ and $y$.
\end{lemma}
\begin{proof}
Consider an arbitrary point $(x, y) \in \cD$ for which $y \in (0, 1 - c)$. If $y < (1 - \delta')(1 - c)$, then $(x, y)$ is part of the interval $\{(x, y') \mid y' \in (0, (1 - \delta')(1 - c))\}$, which is completely included in $\cD$. The mathematical formula defining $F(x, y)$ for points in this interval involves $\min$ and $\max$ operators, and since both of these operators switch from outputting one argument to the next when $y = \frac{1 - c}{c}(x - \delta')$, we get
\begin{equation} \label{eq:cases_y_derivative}
	\frac{\partial F(x, y)}{\partial y}
	=
	\begin{cases}
		\frac{(c - x + \delta'(1 - c))\left(1 - \frac{1}{1 - c}y - \delta'\right)^{1/c - 2}}{c^2} & \text{for $y < \frac{1 - c}{c}(x - \delta')$} \enspace,\\
		\frac{(1 - x - y)^{1/c - 1}}{c} & \text{for $y > \frac{1 - c}{c}(x - \delta')$} \enspace.
	\end{cases}
\end{equation}
To see why the partial derivative given for the case of $y < \frac{1 - c}{c}(x - \delta')$ is correct, we notice that in this case
\begin{multline*}
	F(x, y)
	=
	1 - \bigg(1 - \frac{x - \frac{c}{1 - c}y - \delta'}{c\big(1 - \frac{1}{1 - c}y - \delta'\big)}\bigg)\left(1 - \frac{1}{1 - c}y - \delta'\right)^{1/c}\\
	=
	1 - \frac{c - x + \delta'(1 - c)}{c\big(1 - \frac{1}{1 - c}y - \delta'\big)}\left(1 - \frac{1}{1 - c}y - \delta'\right)^{1/c}
	=
	1 - \frac{c - x + \delta'(1 - c)}{c}\left(1 - \frac{1}{1 - c}y - \delta'\right)^{1/c - 1}
	\enspace,
\end{multline*}
and therefore, its partial derivative with respect to $y$ is
\begin{align*}
	&
	\frac{\frac{c - x + \delta'(1 - c)}{c}(1/c - 1)\left(1 - \frac{1}{1 - c}y - \delta'\right)^{1/c - 2}}{1 - c}
	=
	\frac{(c - x + \delta'(1 - c))\left(1 - \frac{1}{1 - c}y - \delta'\right)^{1/c - 2}}{c^2}
	\enspace.
\end{align*}
Notice that the expressions given for the the two cases of~\eqref{eq:cases_y_derivative} are equal when $y = \frac{1 - c}{c}(x - \delta')$ because for this value of $y$
\[
	\frac{1}{1 - c}y + \delta'
	=
	\frac{1}{c}(x - \delta') + \delta'
	=
	x + \frac{1 - c}{c}(x - \delta')
	=
	x + y
	\enspace,
\]
and
\[
	1 - \frac{1}{1 - c}y - \delta'
	=
	1 - \frac{1}{c}(x - \delta') - \delta'
	=
	\frac{c - x + \delta'(1 - c)}{c}
	\enspace.
\]
Thus, the partial derivative of $F$ with respect to $y$ exists also at the point $(x, \frac{1 - c}{c}(x - \delta'))$ when $\frac{1 - c}{c}(x - \delta') > 0$, and the value of this partial derivative can be calculated using the expressions of either of the two cases of~\eqref{eq:cases_y_derivative}.

Consider now the case of $(1 - \delta')(1 - c) \leq y < 1 - c$. In this case, because $x \leq c$, the point $(x, y)$ is included in the interval $\{(x, y') \mid y' \in (\max\{0, \frac{1 - c}{c}(x - \delta')\}, 1 - c)\}$. Furthermore, this interval is entirely included in the domain $\cD$, and by Observation~\ref{obs:function_equality}, $F(x, y') = G(x, y')$ throughout it. Thus, $F$ and $G$ share the same partial derivative with respect to $y$ at $(x, y)$, and by the proof of Observation~\ref{obs:G_properties}, this partial derivative exists and is given by
\[
	\frac{\partial F(x, y)}{\partial y}
	=
	\frac{(1 - x - y)^{1/c-1}}{c}
	\enspace.
\]

To summarize, the above discussion shows that $F(x, y)$ has a partial derivative with respect to $y$ at every point $(x, y) \in \cD$ for which $y \in (0, 1 - c)$. Furthermore, since $y \geq \frac{1 - c}{c}(x - \delta')$ whenever $y \geq (1 - \delta')(1 - c)$, the value of this partial derivative is always given by
\[
	\frac{\partial F(x, y)}{\partial y}
	=
	\begin{cases}
		\frac{(c - x + \delta'(1 - c))\left(1 - \frac{1}{1 - c}y - \delta'\right)^{1/c - 2}}{c^2} & \text{for $y \leq \frac{1 - c}{c}(x - \delta')$} \enspace,\\
		\frac{(1 - x - y)^{1/c - 1}}{c} & \text{for $y \geq \frac{1 - c}{c}(x - \delta')$} \enspace.
	\end{cases}
\]
The two expressions given for the above two cases are individually non-negative and non-increasing functions of $x$ and $y$ since $c \in (0, 1/2]$ and $x \leq c$. This immediately implies that the partial derivative of $F$ with respect to $y$ is non-negative whenever it is defined. Since the two expressions of the two cases are equal in the borderline between them (when $y = \frac{1 - c}{c}(x - \delta')$), we also get that the partial derivative of $F$ with respect to $y$ is a non-increasing function of $x$ and $y$.
\end{proof}

Recall now that $c$ is a rational constant, and therefore, can be presented as the ratio of two integer constants $h$ and $\ell$ (i.e., $c = h/\ell$). We need to construct a hard distribution $\cH_i$ over instances of the problem of maximizing a non-negative monotone submodular function over a cardinality constraint such that all the instances in the support of the distribution have a density of $c$. Our construction is parametrized by an integer $i \geq 1$ controlling the size of the instances in the support of the distribution $\cH_i$.

The ground set of all the instances in the support of the distribution $\cH_i$ we construct is an arbitrary set $\cN_i$ of size $i\ell$, and the cardinality constraint in all these instances allows the selection of up to $ih$ elements (and thus, the instances indeed have a density of $c$). To draw an instance from the distribution $\cH_i$, we choose a uniformly random subset $O_i \subseteq \cN_i$ of size $ih$, and then define the objective function of the drawn instance as
\[
	f_{O_i}(S)
	\triangleq
	F\bigg(\frac{|S \cap O_i|}{i\ell}, \frac{|S \setminus O_i|}{i\ell}\bigg)
	\quad
	\forall\; S \subseteq \cN_i
	\enspace.
\]
Notice that $|S \cap O_i| / (i\ell) \leq |O_i| / (i\ell) = (ih)/(i\ell) = h/\ell = c$, and $|S \setminus O_i| / (i\ell) \leq (|\cN_i| - |O_i|) / (i\ell) = (i\ell - ih) / (i\ell) = 1 - c$. Thus, the arguments passed to $F$ by the above definition are within the domain on which $F$ is defined. For analysis purposes, it is useful to also define
\[
	g_{O_i}(S)
	\triangleq
	G\bigg(\frac{|S \cap O_i|}{i\ell}, \frac{|S \setminus O_i|}{i\ell}\bigg)
	\quad
	\forall\; S \subseteq \cN_i
	\enspace.
\]
The following observation shows that, interestingly, the function $g_{O_i}(S)$ is independent of $O_i$. When we want to stress this fact below,
we use $g$ (without any subscript) to denote this function.
\begin{observation}
The function $g_{O_i}(S)$ is independent of $O_i$. 
\end{observation}
\begin{proof}
Observe that
\[
	g_{O_i}(S)
	=
	G\bigg(\frac{|S \cap O_i|}{i\ell}, \frac{|S \setminus O_i|}{i\ell}\bigg)
	=
	\bigg(1 - \frac{|S \cap O_i|}{i\ell} - \frac{|S \setminus O_i|}{i\ell}\bigg)^{1/c}
	=
	\bigg(1 - \frac{|S|}{i\ell}\bigg)^{1/c}
	\enspace,
\]
and the rightmost hand side is independent of $O_i$.
\end{proof}

The following lemma proves that the objective functions we have defined for the instances in the support of the distribution $\cH_i$ are valid objective functions for the problem we consider, and thus, $\cH_i$ is a distribution over valid instances of our problem. An analogous lemma can also be proved for the function $g$ by using Observation~\ref{obs:G_properties} instead of Lemmata~\ref{lem:F_non-negative}, \ref{lem:F_derivative_x} and~\ref{lem:F_derivative_y}.
\begin{lemma}
For every set $O_i \subseteq \cN_i$ of size $ih$, $f_{O_i}$ is a non-negative monotone submodular function.
\end{lemma}
\begin{proof}
The non-negativity of $f_{O_i}$ follows from the fact that $F$ is non-negative by Lemma~\ref{lem:F_non-negative}. Consider now any two sets $S \subseteq T \subseteq \cN_i$. Then,
\begin{align*}
	f_{O_i}(T) - f_{O_i}(S)
	={} &
	F\bigg(\frac{|T \cap O_i|}{i\ell}, \frac{|T \setminus O_i|}{i\ell}\bigg) - F\bigg(\frac{|S \cap O_i|}{i\ell}, \frac{|S \setminus O_i|}{i\ell}\bigg)\\
	={} &
	\bigg[F\bigg(\frac{|T \cap O_i|}{i\ell}, \frac{|T \setminus O_i|}{i\ell}\bigg) - F\bigg(\frac{|S \cap O_i|}{i\ell}, \frac{|T \setminus O_i|}{i\ell}\bigg)\bigg] \\&\qquad+ \bigg[F\bigg(\frac{|S \cap O_i|}{i\ell}, \frac{|T \setminus O_i|}{i\ell}\bigg) - F\bigg(\frac{|S \cap O_i|}{i\ell}, \frac{|S \setminus O_i|}{i\ell}\bigg)\bigg]\\
	={} &
	\int_{|S \cap O_i| / (i\ell)}^{|T \cap O_i| / (i\ell)} \frac{\partial F(x, |T \setminus O_i| / (i\ell))}{\partial x} \text{d}x
	+
	\int_{|S \setminus O_i| / (i\ell)}^{|T \setminus O_i| / (i\ell)} \frac{\partial F(|S \cap O_i| / (i\ell), y)}{\partial y} \text{d}y
	\geq
	0
	\enspace,
\end{align*}
where the last equality holds since Lemma~\ref{lem:F_derivative_x} shows that $F(x, y)$ has a partial derivative with respect to $x$ whenever $x \in (0, c)$ and Lemma~\ref{lem:F_derivative_y} shows that $F(x, y)$ has a partial derivative with respect to $y$ whenever $y \in (0, 1 - c)$. The inequality holds since Lemma~\ref{lem:F_derivative_x} and~\ref{lem:F_derivative_y} also show that these partial derivatives are non-negative. This proves that $f_{O_i}$ is monotone.

Let $u$ be an arbitrary element of $\cN_i \setminus T$. If $u \in O_i$, then
\begin{align*}
	f_{O_i}(u \mid S)
	={} &
	F\bigg(\frac{|S \cap O_i| + 1}{i\ell}, \frac{|S \setminus O_i|}{i\ell}\bigg) - F\bigg(\frac{|S \cap O_i|}{i\ell}, \frac{|S \setminus O_i|}{i\ell}\bigg)\\
	={} &
	\int_0^{(i\ell)^{-1}} \frac{\partial F(x, |S \setminus O_i| / (i\ell))}{\partial x}\bigg|_{x = |S \cap O_i|/(i\ell) + \tau} \text{d}\tau\\
	\geq{} &
	\int_0^{(i\ell)^{-1}} \frac{\partial F(x, |T \setminus O_i| / (i\ell))}{\partial x}\bigg|_{x = |T \cap O_i|/(i\ell) + \tau} \text{d}\tau\\
	={} &
	F\bigg(\frac{|T \cap O_i| + 1}{i\ell}, \frac{|T \setminus O_i|}{i\ell}\bigg) - F\bigg(\frac{|T \cap O_i|}{i\ell}, \frac{|T \setminus O_i|}{i\ell}\bigg)
	=
	f_{O_i}(u \mid T)
	\enspace,
\end{align*}
where the inequality holds since Lemma~\ref{lem:F_derivative_x} shows that the partial derivative of $F$ with respect to $x$ is a non-increasing function of $x$ and $y$ whenever $x \in (0, c)$. Similarly, if $u \not \in O_i$, then
\begin{align*}
	f_{O_i}(u \mid S)
	={} &
	F\bigg(\frac{|S \cap O_i|}{i\ell}, \frac{|S \setminus O_i| + 1}{i\ell}\bigg) - F\bigg(\frac{|S \cap O_i|}{i\ell}, \frac{|S \setminus O_i|}{i\ell}\bigg)\\
	={} &
	\int_0^{(i\ell)^{-1}} \frac{\partial F(|S \cap O_i| / (i\ell), y)}{\partial y}\bigg|_{y = |S \setminus O_i|/(i\ell) + \tau} \text{d}\tau\\
	\geq{} &
	\int_0^{(i\ell)^{-1}} \frac{\partial F(|T \cap O_i| / (i\ell), y)}{\partial y}\bigg|_{y = |T \setminus O_i|/(i\ell) + \tau} \text{d}\tau\\
	={} &
	F\bigg(\frac{|T \cap O_i|}{i\ell}, \frac{|T \setminus O_i| + 1}{i\ell}\bigg) - F\bigg(\frac{|T \cap O_i|}{i\ell}, \frac{|T \setminus O_i|}{i\ell}\bigg)
	=
	f_{O_i}(u \mid T)
	\enspace,
\end{align*}
where the inequality holds since Lemma~\ref{lem:F_derivative_y} shows that the partial derivative of $F$ with respect to $y$ is a non-increasing function of $x$ and $y$ whenever $y \in (0, 1 - c)$. This proves that $f_{O_i}$ is submodular.
\end{proof}

Our next goal is to show that $\cH_i$ is a hard distribution. The first step towards this goal is to lower bound and upper bound the value of the optimal solution under this distribution.
\begin{observation} \label{obs:f_optimal_bound}
The value of the optimal solution for every instance in the support of the distribution $\cH_i$ is at least $1 - 2\delta'/c$, and no more than $1$.
\end{observation}
\begin{proof}
Notice that $O_i$ is a feasible solution since its size is exactly $ih$. The value of this feasible solution according to the objective function $f_{O_i}$ is
\begin{align*}
	f_{O_i}(O_i)
	={} &
	F\bigg(\frac{|O_i \cap O_i|}{i\ell}, \frac{|O_i \setminus O_i|}{i\ell}\bigg)
	=
	F(c, 0)\\
	={} &
	1 - \bigg(1 - \max\bigg\{0, \frac{c - \frac{c}{1 - c} \cdot 0 - \delta'}{c\big(1 - \frac{1}{1 - c} \cdot 0 - \delta'\big)}\bigg\}\bigg)\bigg(1 - \min\Big\{c + 0, \frac{1}{1 - c} \cdot 0 + \delta'\Big\}\bigg)^{1/c}\\
	={} &
	1 - \bigg(1 - \max\bigg\{0, \frac{c - \delta'}{c(1 - \delta')}\bigg\}\bigg)(1 - \min\{c, \delta'\})^{1/c}\\
	\geq{} &
	1 - \bigg(1 - \frac{c - \delta'}{c(1 - \delta')}\bigg)(1 - 0)^{1/c}
	=
	\frac{c - \delta'}{c(1 - \delta')}
	=
	1 - \frac{\delta'(1 - c)}{c(1 - \delta')}
	\geq
	1 - \frac{2\delta'}{c}
	\enspace,
\end{align*}
where the last inequality holds since $\delta' < c / 4 < 1/2$. This completes the proof of the first part of the observation. The second part of the observation follows immediately from Lemma~\ref{lem:F_non-negative}.
\end{proof}

Consider now an arbitrary bicriteria deterministic algorithm $\alg$ for the problem we consider that uses a number of value oracle queries that is sub-exponential in the size of the ground set and has an infeasibility ratio of $\rho(c, \eps)$. We assume without loss of generality that $\alg$ returns one of the sets it makes a value oracle query on (if this is not the case, we can add such a query at the very end of $\alg$). Since $\alg$ is deterministic, when it is applied to an instance defined by the ground set $\cN_i$ and the objective function $g$, it makes queries on a fixed series of sets $Q_1, Q_2, \dotsc, Q_{m(i)}$. Notice that $m(i)$ must be a sub-exponential function of $i$ since $\alg$ makes a number of queries that is sub-exponential in the size $i\ell$ of the ground set $\cN_i$, and $\ell$ is a constant.

Let $\cE$ be the event that $f_{O_i}(Q_j) = g(Q_j)$ for every $j \in [m(i)]$.
\begin{lemma} \label{lem:follow_g_probability}
$\Pr[\cE] = 1 - o(1)$.
\end{lemma}
\begin{proof}
For every $j \in [m(i)]$, let $\bar{\cE}_j$ be the event that $f_{O_i}(Q_j) \neq g_{O_i}(O_j)$. Since $g$ and $g_{O_i}$ are different names for the same function, we get by the union bound that
\[
	\Pr[\cE]
	\geq
	1 - \sum_{j = 1}^{m(i)} \Pr[\bar{\cE}_j]
	\geq
	1 - m(i) \cdot \max_{j = 1}^{m(i)} \Pr[\bar{\cE}_j]
	\enspace.
\]
Below we prove that the probability of $\bar{\cE}_j$ is exponentially small in $i$ (for every $j \in [m(i)]$), which implies the lemma since $m(i)$ is a sub-exponential function.

Let $X = |O_i \cap Q_j|$. By Observation~\ref{obs:function_equality}, the event $\bar{\cE}_j$ can only happen when
\[
	\frac{X}{i\ell}
	=
	\frac{|O_i \cap Q_j|}{i\ell}
	>
	\frac{c}{1 - c}\cdot\frac{|Q_j \setminus O_i|}{i\ell} + \delta'
	=
	\frac{c}{1 - c}\cdot\frac{|Q_j| - X}{i\ell} + \delta'
	\enspace.
\]
Rearranging the last inequality, we get that it is equivalent to $X > c|Q_j| + \delta' i\ell(1 - c)$. Notice now that $X$ has a hypergeometric distribution, and $\mathbb{E}[X] = |O_i| \cdot |Q_j| / |N_i| = (ih) \cdot |Q_j| / (i\ell) = c|Q_j|$. Thus, using a bound given by \cite{skala2013hypergeometric} (based on results due to \cite{chvatal1979tail} and \cite{hoeffding1963probability}), we get
\begin{align*}
	\Pr[\bar{\cE}_j]
	\leq{} &
	\Pr[X > c|Q_j| + \delta' i\ell(1 - c)]\\
	={} &
	\Pr\bigg[X > c|Q_j| + \frac{\delta' |O_i|(1 - c)}{c}\bigg]
	\leq
	e^{-2\big(\frac{\delta'(1 - c)}{c}\big)^2|O_i|}
	=
	e^{-\frac{2\delta'^2(1 - c)ih}{c^2}}
	\enspace.
\end{align*}
Notice that the expression on the rightmost side of the last inequality is indeed exponentially small in $i$ since $c$, $\ell$ and $\delta'$ are constants and $c < 1$.
\end{proof}
\begin{corollary} \label{cor:deterministic_fail_probability}
When given an instance drawn from $\cH_i$, $\alg$ returns with probability $1 - o(1)$ a solution set $S$ such that $f_{O_i}(S) = g(S)$.
\end{corollary}
\begin{proof}
Given Lemma~\ref{lem:follow_g_probability}, to prove the corollary it suffices to show that whenever the event $\cE$ happens, $\alg$ returns a solution set $S$ such that $f_{O_i}(S) = g(S)$. Thus, let us assume in the rest of the proof that the event $\cE$ happens.

Since $\alg$ is a deterministic algorithm, when getting instance with either $f_{O_i}$ or $g$ as the objective function, it follows the same execution path as long as the queries to the value oracle that $\alg$ makes return the same answer for both objective functions. The event $\cE$ implies that this is the case for all the queries that $\alg$ makes when given $g$ as the objective, and thus, $\alg$ follows the same execution path when given either of the above objective functions (until this execution path terminates). This implies, in particular, that $\alg$ returns the same solution set when given either $f_{O_i}$ or $g$ as the objective function. Let us denote this set by $S$.

Recall that we assume that $\alg$ outputs one of the sets on which it queries the objective function. Since $S$ is the set that $\alg$ returns when given $g$ as the objective function, the set $S$ must be one of the sets $Q_1, Q_2, \dotsc, Q_{m(i)}$. However, the event $\cE$ guarantees that the equality $f_{O_i}(S) = g(S)$ holds for all these sets, which complete the proof of the corollary.
\end{proof}

We are now ready to upper bound the expected value of the solution $\alg$ returns given an instance drawn from the hard distribution we have defined.
\begin{lemma} \label{lem:ALG_expected_value}
When given an instance drawn from $\cH_i$, the expected value of the set returned by $\alg$ is at most $1 - \eps + o(1)$.
\end{lemma}
\begin{proof}
Let $\cE'$ be the event that $f_{O_i}(S) = g(S)$, where $S$ is the output set of $\alg$. By Corollary~\ref{cor:deterministic_fail_probability}, $\Pr[\cE'] \geq 1 - o(1)$. Thus, by the law of total expectation,
\begin{align*}
	\bE[f_{O_i}(S)]
	={} &
	\Pr[\cE'] \cdot \bE[f_{O_i}(S) \mid \cE'] + \Pr[\bar{\cE}'] \cdot \bE[f_{O_i}(S) \mid \bar{\cE}']\\
	\leq{} &
	\bE[g(S) \mid \cE'] + o(1)
	\leq
	1 - \eps + o(1)
	\enspace,
\end{align*}
where the first inequality follows from the inequality $f_{O_i}(S) \leq 1$, which holds by Observation~\ref{obs:f_optimal_bound}, and the second inequality is true since the fact that $\alg$ has an infeasibility ratio of $\rho(c, \eps)$ implies that $|S| \leq \rho(c, \eps) \cdot ih = (1 - \eps^c) \cdot i\ell$, which implies
\[
	g(S)
	=
	g_{O_i}(S)
	=
	1 - \bigg(1 - \frac{|O_i \cap S|}{i\ell} - \frac{|S \setminus O_i|}{i\ell}\bigg)^{1/c}
	=
	1 - \bigg(1 - \frac{|S|}{i\ell}\bigg)^{1/c}
	\leq
	1 - \eps
	\enspace.
	\qedhere
\]
\end{proof}
\begin{corollary} \label{cor:approximation_ratio}
The approximation ratio of any algorithm $\alg_R$ for maximizing a non-negative monotone submodular function subject to a cardinality constraints that makes a sub-exponential number of value oracle queries and has an infeasibility ratio of $\rho(c, \eps)$ is at most $1 - \eps + 4\delta'/c + o(1)$.
\end{corollary}
\begin{proof}
Since $\alg_R$ can be viewed as a distribution over deterministic algorithm that make a sub-exponential number of value oracle queries, by Yao's principle, Lemma~\ref{lem:ALG_expected_value} implies that there exists an instance $I$ in the support of $\cH_i$ such that, when given $i$, $\alg_R$ outputs a solution of expected value at most $1 - \eps + o(1)$. In contrast, Observation~\ref{obs:f_optimal_bound} shows that the instance $I$ has a solution of value at least $1 - 2\delta'/c$. Thus, the approximation ratio of $\alg_R$ is no better than
\[
	\frac{1 - \eps + o(1)}{1 - 2\delta'/c}
	\leq
	1 - \eps + 4\delta'/c + o(1)
	\enspace.
\]
where the inequality holds since $2\delta'/c \leq 1/2$.
\end{proof}

The part of Theorem~\ref{thm:monotone_cardinality_inapproximability} regarding $c \in (0, 1/2]$ now follows from Corollary~\ref{cor:approximation_ratio} by choosing a large enough value for $i$ so that $o(1) < \delta/2$, and setting $\delta'$ in our construction to $c\delta/8$ (which implies that $4\delta'/c = \delta/2$). Notice that since $\delta$ and $c$ are constants, so is $\delta'$.

\subsubsection{Proof of Theorem~\texorpdfstring{\ref{thm:monotone_cardinality_inapproximability}}{\ref*{thm:monotone_cardinality_inapproximability}} for \texorpdfstring{$c > 1/2$}{c > 1/2}} \label{sssc:cardinality_inapproximability_large_c}

In this section, we prove Theorem~\ref{thm:monotone_cardinality_inapproximability} for $c > 1/2$. The proof can be viewed as a generalization of the proof used by \cite{filmus2025separating} for the special case of single-criteria algorithms. Formally, the proof is based on the construction done in Section~\ref{sssc:cardinality_inapproximability_small_c}. For readability, we begin this section by repeating the necessary properties of this construction (or more accurately, the properties of the special case of this construction obtained by plugging $c = 1/2$). By Section~\ref{sssc:cardinality_inapproximability_small_c}, for every integer $i \geq 1$ and constant $\delta' \in (0, 1/8)$, there exists a distribution $\cH_i$ over instances of the problem of maximizing a non-negative monotone submodular function subject to a cardinality constraint such that
\begin{itemize}
	\item the ground set $\cN_i$ of all these instances has a size $2i$, and a feasible set is allowed to contain up to $i$ elements of this ground set.
	\item the optimal solution for all the instances in $\cN_i$ has a value between $1 - 4\delta'$ and $1$, and $1$ is also an upper bound on the value of any (even infeasible) set.
	\item for every deterministic algorithm $\alg$ that uses a sub-exponential number of value oracle queries, given an instance drawn from $\cH_i$, $\alg$ is guaranteed to output a set $S$ whose value (according to the objective function of the drawn instance) is $1 - \big(1 - \frac{|S|}{2i}\big)^2$ with probability at least $1 - o(1)$, where the $o(1)$ terms diminishes with $i$.
\end{itemize}

Fix now an arbitrary rational value $c \in (1/2, 1)$. Since $c$ is rational, it can be represented as the ratio $h / \ell$ for two positive integers $h$ and $\ell$ (notice that $h < \ell < 2h$). Algorithm~\ref{alg:drawing_c} describes a way to randomly draw an instance of the problem of maximizing a monotnone submodular function subject to a cardinality constraint of density $c$. Let $\cH'_i$ denote the distribution over instances defined by this algorithm. The value $r$ used in Algorithm~\ref{alg:drawing_c} is real number in the range $[0, 1]$ to be determined later.

\begin{algorithm} 
\caption{Drawing from the distribution $\cH'_i$.} \label{alg:drawing_c}
Draw an instance out of $\cH_{(\ell - h) i}$ consisting of the ground set $\cN_{(\ell - h)i}$ and an objective function $f$.\\
Define a new ground set $\cN'_i$ of size $\ell i$ by adding $(2h - \ell) i$ new elements to $\cN_{(\ell - h) i}$ (recall that $\cN_{(\ell - h) i}$ has a size of $2(\ell - h) i$).\\
Define a function $g\colon 2^{\cN'_i} \to \nnR$ as follows. For every $S \subseteq \cN'_i$,
\[
	g(S)
	\triangleq
	f(S \cap \cN_{(\ell - h)i}) + \frac{r}{(\ell - h)i} \cdot |S \setminus \cN_{(\ell - h)i}|
	\enspace.
\]\\
\Return an instance with the ground set $\cN'_i$ and the objective function $g$ in which a feasible set can contain up to $h i$ elements.
\end{algorithm}

\begin{observation}
In every instance drawn from $\cH'_i$, the cardinality constraint has a density of $c$ and the objective function $g$ is non-negative, mononotone and submodular.
\end{observation}
\begin{proof}
Since the size of the ground set $\cN'_i$ is $\ell i$, and every instance drawn from $\cH'_i$ allows the selection of up to $hi$ elements, the density of the cardinality constraint of this instance is $(hi) / (\ell i) = h / \ell = c$. The objective function $g$ of the instance is the sum of two functions. The first function is $r \cdot |S \setminus \cN_{(\ell - h)i}|$, which is a non-negative linear function, and thus, also monotone and submodular. The other function$f(S \cap \cN_{(\ell - h)i})$ is non-negative, monotone and submodular because $f$ has these properties. This completes the proof of the observation because non-negativity, monotonicity and submodularity are all close under addition.
\end{proof}

Let $\alg'$ be an arbitrary determinist bicriteria algorithm for maximizing a non-negative monotone submodular function subject to a cardinality constraint of density $c$ that uses a sub-exponential number of value oracle queries and has an infeasibility ratio of $\rho(c, \eps)$ for some value $\eps \in (0, 1)$. We would like to study the performance of $\alg'$ on an instance $\cI$ drawn from $\cH'_i$. Let $\cE$ be event that the output set $S$ returned by $\alg'$ when given $\cI$ obeys $f(S \cap \cN_{(\ell - h)i}) = 1 - \big(1 - \frac{|S|}{2i}\big)^2$. Note that the calculation of $S \cap \cN_{(\ell - h)i}$ can be viewed as a deterministic algorithm that gets as input an instance drawn out of $\cH_{(\ell - h)i}$ and outputs a subset of $\cN_{(\ell - h)i}$. This calculation involves only a sub-exponential number of value oracle queries to $f$ (because each value oracle query to $g$ used by $\alg'$ can be evaluated using a single value oracle query to $f$), and thus, by the properties of $\cH_{(\ell - h)i}$, it must hold that $\Pr[\cE] = 1 - o(1)$.

\begin{lemma} \label{lem:conditioned_bound}
Conditioned on $\cE$, the output set $S$ of $\alg'$ obeys $g(S) \leq 1 + r\big[r - 2 + \frac{c}{1 - c} \cdot \rho(c, \eps)\big]$.
\end{lemma}
\begin{proof}
The condition $\cE$ implies that
\begin{align*}
	g(S)
	={} &
	f(S \cap \cN_{(\ell - h)i}) + \frac{r}{(\ell - h)i} \cdot |S \setminus \cN_{(\ell - h)i}|\\
	={} &
	1 - \bigg(1 - \frac{|S \cap \cN_{(\ell - h)i}|}{2(\ell - h)i}\bigg)^2 + \frac{r}{(\ell - h)i} \cdot [|S| - |S \cap \cN_{(\ell - h)i}|]\\
	\leq{} &
	1 - \bigg(1 - \frac{|S \cap \cN_{(\ell - h)i}|}{2(\ell - h)i}\bigg)^2 + \frac{r}{(\ell - h)i} \cdot [\rho(c, \eps) \cdot h i - |S \cap \cN_{(\ell - h)i}|]
	\enspace,
\end{align*}
where the inequality holds since the infeasibility ratio of $\alg'$ is $\rho(c, \eps)$ by definition. Therefore, if we denote
\[
	h(x)
	\triangleq
	1 - \bigg(1 - \frac{x}{2(\ell - h)i}\bigg)^2 + \frac{r}{(\ell - h)i}[\rho(c, \eps) \cdot h i - x]
	\enspace,
\]
then
\begin{align*}
	g(S)
	\leq{} &
	h(|S \cap \cN_{(\ell - h)i}|)
	\leq
	\max_{x \geq 0} h(x)
	=
	h(2i(\ell - h)(1 - r))\\
	={} &
	1 - r^2 + \frac{r}{(\ell - h)i} \cdot [\rho(c, \eps) \cdot h i - 2i(\ell - h)(1 - r)]
	=
	1 + r\bigg[r - 2 + \frac{c}{1 - c} \cdot \rho(c, \eps)\bigg]
	\enspace,
\end{align*}
where the first equality holds since the derivative of $h$ is $\frac{1 - x/[2(\ell - h)i] - r}{(\ell - h)i}$, which is a decreasing function of $x$ taking the value $0$ at $x = 2i(\ell - h)(1 - r)$.
\end{proof}
\begin{corollary} \label{cor:ALG_prime_upper_bound}
The output set $S$ of $\alg'$ obeys $\bE[g(S)] \leq 1 + r\big[r - 2 + \frac{c}{1 - c} \cdot \rho(c, \eps)\big] + o(1)$. Notice that since $\alg'$ is a deterministic algorithm, the expectation is over the distribution $\cH'_{i}$ from which the input instance for $\alg'$ is drawn.
\end{corollary}
\begin{proof}
Recall that there are only $(2h - \ell)i$ elements in $\cN'_i$ that do not belong to $\cN_{(\ell - h)i}$. Therefore, it always holds that
\[
	g(S)
	=
	f(S \cap \cN_{(\ell - h)i}) + \frac{r}{(\ell - h)i} \cdot |S \setminus \cN_{(\ell - h)i}|
	\leq
	1 + \frac{r(2h - \ell)}{\ell - h}
	\leq
	1 + \frac{2h - \ell}{\ell - h}
	\enspace.
\]
Hence, by the law of total expectation and Lemma~\ref{lem:conditioned_bound},
\begin{align*}
	\bE[g(S)]
	={} &
	\Pr[\cE] \cdot \bE[g(S) \mid \cE] + \Pr[\bar{\cE}] \cdot \bE[g(S) \mid \bar{\cE}]\\
	\leq{} &
	1 + r\bigg[r - 2 + \frac{c}{1 - c} \cdot \rho(c, \eps)\bigg] + o(1) \cdot \bigg[1 + \frac{2h - \ell}{\ell - h}\bigg]\\
	={} &
	1 + r\bigg[r - 2 + \frac{c}{1 - c} \cdot \rho(c, \eps)\bigg] + o(1)
	\enspace,
\end{align*}
where the last equality holds since $o(1)$ diminishes as $i$ grows, while $h$ and $\ell$ are constants independent of $i$.
\end{proof}

We now denote by $\alg'_R$ an arbitrary (possibly randomized) bicriteria algorithm for maximizing a non-negative monotone submodular function subject to a cardinality constraint of density $c$ that uses a sub-exponential number of value oracle queries and has an infeasibility ratio of $\rho(c, \eps)$. By Yao's principle, the last corollary implies that there exists an instance $\cI$ in the support of the distribution of $\cH'_i$ on which $\alg'_R$ outputs a set $S$ such that $\bE[g(S)] \leq 1 + r\big[r - 2 + \frac{c}{1 - c} \cdot \rho(c, \eps)\big] + o(1)$. To get from this inequality a bound on the approximation ratio of $\alg'_R$, we need to lower bound the value of the optimal solution of $\cI$.
\begin{observation} \label{obs:prime_opt_bound}
Every instance in the support of $\cH'_i$ has a feasible solution whose value is at least $1 - 4\delta' + \frac{r(2c - 1)}{1 - c}$, and thus, the approximation ratio of $\alg_R$ is at most $\frac{1 + r[r - 2 + \frac{c}{1 - c} \cdot \rho(c, \eps)] + o(1)}{1 - 4\delta' + r(2c - 1) / (1 - c)}$.
\end{observation}
\begin{proof}
By the definition of $\cH_{(\ell - h)i}$, there is a set $S \subseteq \cN_{(\ell - h)i}$ of size $(\ell - h)i$ such that $f(S) \geq 1 - 4\delta'$. Adding to $S$ the $(2h - \ell)i$ elements of $\cN'_i \setminus \cN_{(\ell - h)i}$, we get a set $S'$ of size $hi$. Notice that $S'$ is a feasible solution for any instance of $\cH'_i$. Furthermore, the value of $S'$ is
\[
	g(S')
	=
	f(S' \cap \cN_{(\ell - h)i}) + \frac{r}{i(\ell - h)} \cdot |S' \setminus \cN_{(\ell - h)i}|
	=
	f(S) + \frac{r}{i(\ell - h)} \cdot (2h - \ell)i
	\geq
	1 - 4\delta' + \frac{r(2c - 1)}{1 - c}
	\enspace.
	\qedhere
\]
\end{proof}


At this point, we set $r = \sqrt{\eps}$. 
Notice that $r$ indeed belongs to the range $[0, 1]$, as was promised above.

\begin{lemma} \label{lem:competitive_delta}
The competitive ratio of $\alg'_R$ is no better than $1 - \eps + 8\delta' + o(1)$.
\end{lemma}
\begin{proof}
Plugging the values of $r$ and $\rho(c, \eps)$ into the bound given by Observation~\ref{obs:prime_opt_bound}, this bound becomes
\begin{align*}
	&
	\frac{1 + \sqrt{\eps}\big[\sqrt{\eps} - 2 + \frac{1 - 2(1 - c)\sqrt{\eps} - \eps(2c - 1)}{1 - c}\big] + o(1)}{1 - 4\delta' + \sqrt{\eps}(2c - 1) / (1 - c)}\\
	={} &
	\frac{1 - c + \sqrt{\eps}\big[(1 - c)(\sqrt{\eps} - 2) + 1 - 2(1 - c)\sqrt{\eps} - \eps(2c - 1)\big] + o(1)}{1 - c - 4\delta'(1 - c) + \sqrt{\eps}(2c - 1)}\\
	={} &
	\frac{(1 - \eps)[1 - c + \sqrt{\eps}(2c - 1)] + o(1)}{1 - c + \sqrt{\eps}(2c - 1) - 4\delta'(1 - c)}\\
	\leq{} &
	\frac{(1 - \eps)[1 - c + \sqrt{\eps}(2c - 1)]}{1 - c + \sqrt{\eps}(2c - 1) - 4\delta'(1 - c)} + o(1)
	\leq
	1 - \eps + 8\delta' + o(1)
	\enspace,
\end{align*}
where the two inequalities hold since the fact that $\delta' \leq 1/8$ implies that the denominator is at least $(1 - c)/2$, which is positive and independent of $i$.
\end{proof}

The part of Theorem~\ref{thm:monotone_cardinality_inapproximability} regarding $c \in (1/2, 1)$ follows from Lemma~\ref{lem:competitive_delta} by choosing a large enough value for $i$ so that $o(1) < \delta/2$, and setting $\delta'$ to $\delta/16$ (which implies that $8\delta' = \delta/2$).

\subsection{Simple Greedy Algorithms} \label{ssc:simple_greedy_monotone}

In this section, we present simple greedy algorithms for bicriteria maximization of monotone submodular functions subject to cardinality, knapsack and matroid constraints. The first such algorithm that we study is a density-based greedy algorithm for maximization subject to a knapsack (and thus, also cardinality) constraint. This algorithm appears as Algorithm~\ref{alg:greedy_density}.

\begin{algorithm}[ht] 
\caption{\textsc{Density-Based Greedy}$(f \colon 2^\cN \to \nnR, c, B, \eps)$} \label{alg:greedy_density}
\DontPrintSemicolon
\lIf{$c(\cN) \leq B \cdot \ln \eps^{-1}$}{\Return $\cN$.\label{line:early_termination}}
Let $S_0 \gets \{u \in \cN \mid c(u) = 0\}$ and $i \gets 0$.\\
\While{$c(S_i) < B \cdot \ln \eps^{-1}$\label{line:greedy_density_loop}}
{
	Update $i \gets i + 1$.\\
	Let $u_i$ be an element of $\cN \setminus S_{i - 1}$ maximizing $\frac{f(u_i \mid S_{i - 1})}{c(u_i)}$.\\
	Let $S_i \gets S_{i - 1} + u_i$.
}
\Return $S_i$.
\end{algorithm}

The properties of Algorithm~\ref{alg:greedy_density} are summarized by the following theorem.
\begin{theorem} \label{thm:monotone_knapsack_greedy}
For every $\eps \in (0, 1)$, Algorithm~\ref{alg:greedy_density} is a polynomial time $(1 - \eps, 1 + \ln \eps^{-1})$-bicriteria approximation algorithm for the problem of maximizing a non-negative monotone submodular function subject to a knapsack constraint. In the special case of a cardinality constraint, the infeasibility ratio of Algorithm~\ref{alg:greedy_density} improves to $\lceil \ln \eps^{-1} \rceil$.
\end{theorem}

To prove Theorem~\ref{thm:monotone_knapsack_greedy}, we need to analyze the approximation and infeasibility ratios of Algorithm~\ref{alg:greedy_density}. Each one of the following two lemmata analyzes one of these ratios.
\begin{lemma} \label{lem:monotone_knapsack_approximation}
The approximation ratio of Algorithm~\ref{alg:greedy_density} is at least $1 - \eps$.
\end{lemma}
\begin{proof}
If $c(\cN) \leq B \cdot \ln \eps^{-1}$, then Algorithm~\ref{alg:greedy_density} return $\cN$, and by the monotonicity of $f$ we have $f(\cN) \geq f(OPT) \geq (1 - \eps) \cdot f(\opt)$. Thus, in the rest of the proof, we may assume that $c(\cN) > B \cdot \ln \eps^{-1}$. Notice that this assumption implies that the loop on Line~\ref{line:greedy_density_loop} of Algorithm~\ref{alg:greedy_density} terminates because each iteration of this loops adds another element of $\cN$ to $S_i$ (and thus, if the loop makes $n$ iterations, then after the $n$-th iteration $c(S_n) = c(\cN) > B \cdot \ln \eps^{-1}$, which prevents the loop from starting another iteration). Let $\ell$ be the number of iterations performed by the loop on Line~\ref{line:greedy_density_loop}. By the condition of this loop, $c(S_\ell) \geq B \cdot \ln \eps^{-1}$.

We now prove by induction that for every integer $0 \leq i \leq \ell$, $f(S_i) \geq (1 - e^{-c(S_i) / B}) \cdot f(\opt)$. Notice that this will complete the proof the lemma since it implies that
\[
	f(S_\ell)
	\geq
	(1 - e^{-c(S_\ell)/B}) \cdot f(\opt)
	\geq
	(1 - e^{-\ln \eps^{-1}}) \cdot f(\opt)
	=
	(1 - \eps) \cdot f(\opt) \enspace.
\]
For $i = 0$, the inequality $f(S_i) \geq (1 - e^{-c(S_i) / B}) \cdot f(\opt)$ holds since the non-negativity of $f$ guarantees that
\[
	f(S_0)
	\geq
	0
	=
	(1 - e^{-0 / B}) \cdot f(\opt)
	=
	(1 - e^{-c(S_0) / B}) \cdot f(\opt)
	\enspace.
\]
Next, we need to prove the inequality $f(S_i) \geq (1 - e^{-c(S_i) / B}) \cdot f(\opt)$ for $i \in [\ell]$ given the induction hypothesis that this inequality holds for $i - 1$. Since every element of $\opt \setminus S_{i - 1}$ is a possible candidate to be chosen as $u_i$, we have
\begin{align*}
	f(S_i)
	={} &
	f(S_{i - 1}) + f(u_i \mid S_i)
	\geq
	f(S_{i - 1}) + c(u_i) \cdot \max_{u \in \opt \setminus S_{i - 1}} \mspace{-9mu} \frac{f(u \mid S_{i - 1})}{c(u)}\\
	\geq{} &
	f(S_{i - 1}) + c(u_i) \cdot \frac{\sum_{u \in \opt \setminus S_{i - 1}} f(u \mid S_{i - 1})}{\sum_{u \in \opt \setminus S_{i - 1}} c(u)}
	\geq
	f(S_{i - 1}) + c(u_i) \cdot \frac{f(\opt \mid S_{i - 1})}{B}\\
	\geq{} &
	f(S_{i - 1}) + c(u_i) \cdot \frac{f(\opt) - f(S_{i - 1})}{B}
	=
	\Big(1 - \frac{c(u_i)}{B}\Big) \cdot f(S_{i - 1}) + \frac{c(u_i)}{B} \cdot f(\opt)
	\enspace,
\end{align*}
where the third inequality follows from $f$'s submodularity, and the last inequality holds by $f$'s monotonicity. Using the induction hypothesis, the previous inequality implies
\begin{align*}
	f(S_i)
	\geq{} &
	\Big(1 - \frac{c(u_i)}{B}\Big) (1 - e^{-c(S_{i - 1}) / B}) \cdot f(\opt) + \frac{c(u_i)}{B} \cdot f(\opt)\\
	={} &
	\Big(1 - e^{-c(S_{i - 1}) / B}\Big(1 - \frac{c(u_i)}{B}\Big)\Big) \cdot f(\opt)
	\geq
	(1 - e^{-c(S_{i - 1}) / B}\cdot e^{- c(u_i) / B}) \cdot f(\opt)\\
	={} &
	(1 - e^{-c(S_{i - 1} + u_i) / B}) \cdot f(\opt)
	=
	(1 - e^{-c(S_i) / B}) \cdot f(\opt)
	\enspace,
\end{align*}
which completes the proof by induction.
\end{proof}

\begin{lemma} \label{lem:greedy_monotone_infeasibility}
The infeasibility ratio of Algorithm~\ref{alg:greedy_density} is at most $1 + \ln \eps^{-1}$. In the special case of a cardinality constraint, this infeasibility ratio improves to $\lceil \ln \eps^{-1} \rceil$.
\end{lemma}
\begin{proof}
If $c(\cN) \leq B \cdot \ln \eps^{-1}$, then any solution that Algorithm~\ref{alg:greedy_density} outputs clearly has an infeasibility ratio that is upper bounded by $\ln \eps^{-1}$. Thus, in the rest of the proof, we may assume that $c(\cN) > B \cdot \ln \eps^{-1}$. As explained in the proof of Lemma~\ref{lem:monotone_knapsack_approximation}, given this assumption the loop on Line~\ref{line:greedy_density_loop} of Algorithm~\ref{alg:greedy_density} terminates after a finite number $\ell$ of iterations, and then Algorithm~\ref{alg:greedy_density} outputs $S_\ell$.

Since Algorithm~\ref{alg:greedy_density} did not terminate after $\ell - 1$ iterations, we must have $c(S_{\ell - 1}) < B \cdot \ln \eps^{-1}$ (notice that $\ell \geq 1$ because $c(S_0) = 0 < B \cdot \ln \eps^{-1}$). Thus, the output set $S_\ell$ of Algorithm~\ref{alg:greedy_density} obeys
\[
	c(S_\ell)
	=
	c(S_{\ell - 1}) + c(u_i)
	<
	B \cdot \ln \eps^{-1} + B
	=
	B \cdot (1 + \ln \eps^{-1})
	\enspace.
\]
In the special case of a cardinality constraint, $c(S_\ell)$ and $B$ are integers and $c(u_i) = 1$, and thus, the last inequality improves to
\[
	c(S_\ell)
	=
	c(S_{\ell - 1}) + c(u_i)
	\leq
	[\lceil B \cdot \ln \eps^{-1} \rceil - 1] + 1
	=
	\lceil B \cdot \ln \eps^{-1} \rceil
	\leq
	B \cdot \lceil \ln \eps^{-1} \rceil
	\enspace.
	\qedhere
\]
\end{proof}

This completes the proof of Theorem~\ref{thm:monotone_knapsack_greedy}. The next simple greedy algorithm is designed for matroid constraints, and is given as Algorithm~\ref{alg:greedy_matroid}. Intuitively, each iteration of this algorithm augments the solution $S_i$ using a set $T_{i, j}$ constructed by a single execution of the standard greedy algorithm for single-criteria maximization of a monotone submodular function subject to a matroid constraint (due to \cite{nemhauser1978analysis}).

\begin{algorithm}[ht] 
\caption{\textsc{Iterative Matroid Greedy}$(f \colon 2^\cN \to \nnR, \cM = (\cN, \cI), \eps)$} \label{alg:greedy_matroid}
\DontPrintSemicolon
Let $S_0 \gets \varnothing$.\\
\For{$i = 1$ \KwTo $\lceil \log_2 \eps^{-1} \rceil$}
{
	Let $T_{i, 0} \gets \varnothing$ and $j \gets 0$.\\
	\While{there exists an element $u \in \cN \setminus (S_{i - 1} \cup T_{i, j})$ for which $T_{i, j} + u \in \cI$\label{line:greedy_matroid_loop}}
	{
		Update $j \gets j + 1$.\\
		Let $u_{i, j}$ be an element maximizing $f(u_{i, j} \mid S_{i - 1} \cup T_{i, j - 1})$ among the elements $u$ that obeyed the condition on Line~\ref{line:greedy_matroid_loop}.\\
		Let $T_{i, j} \gets T_{i, j - 1} + u_{i, j}$.
	}
	Let $S_i \gets S_{i - 1} \cup T_{i, j}$.
}
\Return $S_{\lceil \log_2 \eps^{-1} \rceil}$.
\end{algorithm}

The properties of Algorithm~\ref{alg:greedy_matroid} are summarized by the following theorem.
\begin{theorem} \label{thm:monotone_matroid_greedy}
For every $\eps \in (0, 1)$, Algorithm~\ref{alg:greedy_matroid} is a polynomial time $(1 - \eps, \lceil \log_2 \eps^{-1} \rceil)$-bicriteria approximation algorithm for the problem of maximizing a non-negative monotone submodular function subject to a matroid constraint.
\end{theorem}
\begin{proof}
Let us denote by $j(i)$ the final value of $j$ in the $i$-th iteration of the loop on Line~\ref{line:greedy_matroid_loop} of Algorithm~\ref{alg:greedy_matroid}, and let $\ell$ denote $\lceil \log_2 \eps^{-1} \rceil$. Given this notation, the output of Algorithm~\ref{alg:greedy_matroid} is equal to $\cup_{i = 1}^{\ell} T_{i, j(i)}$. Thus, to prove that the infeasibility ratio of Algorithm~\ref{alg:greedy_matroid} is bounded by $\ell = \lceil \log_2 \eps^{-1} \rceil$, it suffices to show that the set $T_{i, j(i)}$ is independent in the matroid $\cM$ for every $i \in [\ell]$. This is indeed the case because each iteration of the loop on Line~\ref{line:greedy_matroid_loop} of Algorithm~\ref{alg:greedy_matroid} chooses an element $u_{i, j}$ whose addition to $T_{i, j - 1}$ keeps this set independent.

It remains to prove the approximation ratio of Algorithm~\ref{alg:greedy_matroid}. Below, we prove by induction that for every integer $0 \leq i \leq \ell$, $f(S_i) \geq (1 - 2^{-i}) \cdot f(\opt)$. Notice that this will complete the proof of the lemma since this inequality implies, in particular, that the output set $S_\ell$ of Algorithm~\ref{alg:greedy_matroid} obeys
\[
	f(S_\ell)
	\geq
	(1 - 2^{-\ell}) \cdot f(\opt)
	\geq
	(1 - 2^{-\log_2 \eps^{-1}}) \cdot f(\opt)
	=
	(1 - \eps) \cdot f(\opt)
	\enspace.
\]
The base case of the induction holds since the non-negaitivity of $f$ proves that $f(S_0) \geq 0 = (1 - 2^{-0}) \cdot f(\opt)$. Next, we need to prove the inequality $f(S_i) \geq (1 - 2^{-i}) \cdot f(\opt)$ for $i \in [\ell]$ given the induction hypothesis that this inequality holds for $i - 1$. 

Let $\opt' = \opt \setminus S_{i}$, and let $B_1$ and $B_2$ be arbitrary extensions of $\opt'$ and $T_{i, j(i)}$ to bases of $\cM$. By Corollary~\ref{cor:perfect_matching_two_bases}, there exists a bijection function $h\colon B_1 \to B_2$ such that $(B_2 - h(v)) + v \in \cI$ for every $v \in B_1$. Consider now an arbitrary element $v \in \opt'$. If $h(v) \not \in T_{i, j(i)}$, then by the down-closedness property of matroids, $T_{i, j(i)} + v \in \cI$. However, this contradicts the fact that Algorithm~\ref{alg:greedy_matroid} exited the loop on Line~\ref{line:greedy_matroid_loop} after $j(i)$ iterations because $v \in \cN \setminus S_i = \cN \setminus (S_{i - 1} \cup T_{i, j(i)})$. Thus, we are guaranteed that $h(v) \in T_{i, j(i)}$, which implies that there exists some value $j_v \in [j(i)]$ such that $h(v) = u_{i, j_v}$. Since $T_{i, j_v - 1} \subseteq B_2 - u_{i, j_v}$, by the down-closedness property of matroids, we also get $T_{i, j_v - 1} + v \in \cI$. Together with the fact that $v \not \in S_i \supseteq T_{i, j(i)}$, the inclusion of $T_{i, j_v - 1} + u$ in $\cI$ means that Algorithm~\ref{alg:greedy_matroid} could have chosen $v$ when it chose $u_{i, j_v}$, and therefore, 
\[
	f(u_{i, j_v} \mid S_{i - 1} \cup T_{i, j_v - 1})
	\geq
	f(v \mid S_{i - 1} \cup T_{i, j_v - 1})
	\geq
	f(v \mid S_{i - 1} \cup T_{i, j(i)})
	=
	f(v \mid S_i)
	\enspace,
\]
where the second inequality follows from the submodularity of $f$.

Summing up the last inequality over all elements $v \in \opt'$, we get
\begin{multline*}
	f(\opt) - f(S_{i})
	\leq
	f(\opt \cup S_{i}) - f(S_{i})
	\leq
	\sum_{v \in \opt'} \mspace{-9mu} f(v \mid S_i)\\
	\leq
	\sum_{v \in \opt'} \mspace{-9mu} f(u_{i, j_v} \mid S_{i - 1} \cup T_{i, j_v - 1})
	\leq
	\sum_{j = 1}^{j(i)} f(u_{i, j} \mid S_{i - 1} \cup T_{i, j - 1})
	=
	f(S_i) - f(S_{i - 1})
	\enspace,
\end{multline*}
where the first and second inequalities hold by the monotonicity and submodularity of $f$, respectively, and the last inequality follows also from the monotonicity of $f$ since the fact that $h$ is a bijection implies that $j_v$ takes a different value for every $v \in \opt'$. Rearranging the last inequality now yields
\[
	f(S_i)
	\geq
	\frac{f(S_{i - 1}) + f(\opt)}{2}
	\geq
	\frac{(1 - 2^{1-i}) \cdot f(\opt) + f(\opt)}{2}
	=
	(1 - 2^{-i}) \cdot f(\opt)
	\enspace,
\]
where the second inequality holds by the induction hypothesis.
\end{proof}

We note that the bicriteria approximation guarantee stated in Theorem~\ref{thm:monotone_matroid_greedy} is strictly worse than the bicriteria approximation guarantee that can be obtained for the same problem by combining Corollary~\ref{cor:monotone_downclosed} and Lemma~\ref{lem:rounding_matroid}. However, the algorithm obtained by this combination is much more involved than Algorithm~\ref{alg:greedy_matroid}.

%% file: GeneralFunctions.tex
\section{General Submodular Functions} \label{sec:general}

This section includes our results for bicriteria maximization of general submodular functions. It begins with Section~\ref{ssc:additional_preliminaries}, which states some known results and standard definitions that we use later. Then, Section~\ref{ssc:bicriteria_general} presents our results for general down-closed solvable convex set constraints. A simple combinatorial algorithm obtaining the same guarantee for cardinality, knapsack and matroid constraints can be found in Section~\ref{ssc:combinatorial_general}. Finally, Section~\ref{ssc:inapproximability_general} presents two inapproximability results. One inapproximability result showing that no constant bicriteria approximation guarantee can be obtained for non-down-closed convex sets, and one inapproximability result bounding the bicriteria approximation guarantees that can be obtained even for a simple cardinality constraint.

\subsection{Additional Preliminaries} \label{ssc:additional_preliminaries}

In this section, we present some known results from the field of submodular optimization. The first of these results is a variant of a 
lemma originally proven by \citet{feige2011maximizing}. 
\begin{lemma}[Lemma 2.2 of \citet{buchbinder2014submodular}]\label{lemma:prob-lemma}
  Let $f\colon 2^{\cN} \to \nnR$ be a non-negative submodular function. Denote by $A(p)$
  a random subset of $A$ in which each element appears
  with probability at most $p$ (not necessarily independently).
  Then, $\ex{f(A(p))} \ge (1 - p) \cdot f( \varnothing )$. 
\end{lemma}

We now recall that in the unconstrained submodular maximization problem, the goal is to find an arbitrary subset maximizing a given non-negative submodular function. As was mentioned in Section~\ref{sec:warmup}, \citet{buchbinder2015tight} designed an algorithm called \textsc{DoubleGreedy} for this problem that has the following guarantee.
\begin{theorem} \label{thm:discretedoublegreedy}
  Let $f\colon 2^\cN \to \nnR$ be a non-negative submodular function. 
  There exists a (randomized) algorithm
  that finds set $X$, such that
  $f(X) \ge (1/2) \cdot \max_{S \subseteq \cN} f(S) + (1/4) \cdot f( \varnothing )$. 
\end{theorem}

The result of~\cite{buchbinder2015tight} was improved and generalized by multiple works. For example, \citet{buchbinder2018deterministic} showed how to derandomize the \textsc{DoubleGreedy} algorithm. Another version of \textsc{DoubleGreedy} yields the result that appears below as Theorem~\ref{thm:doublegreedy}. This result applies to DR-submoduar functions, which are a class of continuous functions generalizing multilinear extensions of submodular set functions. Following is the formal definition of DR-submoduar functions (for differentiable functions).
\begin{definition}
A differentiable function $F\colon [0, 1]^\cN \to \bR$ is DR-submodular if the inequality $\nabla F(\vx) \geq \nabla F(\vy)$ holds for every two vectors $\vx \leq \vy$ in $[0, 1]^\cN$.
\end{definition}

An important example of differentiable DR-submodular functions are the multilinear extensions of submodular set functions~\citep{bian2017guaranteed}.
\begin{theorem}[Theorem 2.5 of \citet{buchbinder2024constrained}] \label{thm:doublegreedy}
  There exists a polynomial time algorithm that given a differentiable non-negative DR-submodular function
  $F\colon [0,1]^{\cN} \to \nnR$ and parameter $\eps > 0$,
  outputs a vector $\vx \in [0,1]^{\cN}$
  such that
  \[ \max_{r \ge 0} \bigg[ \bigg( \frac{2r}{(r + 1)^2} - O( \eps ) \bigg) \cdot F( \vo ) + \frac{1}{(r + 1)^2}\cdot F( \vzero ) + \frac{r^2}{(r + 1)^2} \cdot F ( \vone_\cN ) \bigg] \enspace,\]
  where $\vo$ is any vector in $[0,1]^{\cN}$. 
\end{theorem}

We conclude this section with the following useful properties of DR-submodular functions.
\begin{lemma}[Lemma~2.1(3) of \citet{buchbinder2024constrained}] \label{lem:positive_direction}
Let $F\colon [0,1]^\cN \to \nnR$ be a non-negative differentiable DR-submodular function. Then, $\inner{\nabla F(\vx)}{\vy} \geq F(\vx+\vy)-F(\vx)$ for every $\vx \in [0,1]^\cN$ and $\vy \geq 0$ such that $\vx+\vy \leq 1$.
\end{lemma}

\begin{lemma}[Lemma~2.2 of \citet{buchbinder2024constrained}] \label{lem:closure}
Given a non-negative DR-submdoualr function $F\colon [0,1]^\cN \to \nnR$ and vector $\vy \in [0, 1]^\cN$, the functions $G_\psum(\vx) \triangleq F(\vx \psum \vy)$ and $G_\hprod(\vx) \triangleq F(\vx \hprod \vy)$ are both non-negative DR-submodular functions.
\end{lemma}

\subsection{A \texorpdfstring{$(1/2-\eps, O( \eps^{-1} ))$}{(1/2-epsilon, O(1/epsilon))}-bicriteria algorithm} \label{ssc:bicriteria_general}

In this section, we prove the following theorem.

\begin{theorem} \label{thm:general_down-closed_convex_body}
For every constant value $\eps \in (0, 1/2)$, there exists a polynomial-time $(1/2 - \eps, O(\eps^{-1})$-bicriteria algorithm for maximizing a non-negative submodular function $f\colon 2^\cN \to \nnR$ subject to a down-closed solvable convex set $\cP \subseteq [0, 1]^N$.
\end{theorem}

More specifically, Section~\ref{ssc:guided-mcg} presents and analyzes a sub-routine termed \textsc{Guided-Measured-Continuous-Greedy} (Algorithm~\ref{alg:guided-mcg}), and then, Section~\ref{ssc:full_algorithm} describes an algorithm that uses this sub-routine and has the properties guaranteed by Theorem~\ref{thm:general_down-closed_convex_body}. To simplify the exposition of our algorithm, we assume direct access to the multilinear extension $F$ of $f$ (rather than just to the set function $f$), and we write the sub-routine \textsc{Guided-Measured-Continuous-Greedy} in the form of a continuous time algorithm. Both these things prevent direct implementation of our algorithm. However, as was previously mentioned in Section~\ref{sssc:knapsack_fractional}, the (by now standard) techniques of~\cite{calinescu2011maximizing} can be used to solve this issue at the cost of deteriorating the approximation ratio of our algorithm by a factor of $1 - o(1)$ and making the algorithm randomized. For the sake of clarity, in the rest of the section we ignore the need to apply these techniques, and directly analyze the algorithms in their written form.

\subsubsection{Guided Measured Continuous Greedy} \label{ssc:guided-mcg}
In this section, we present and analyze a sub-routine termed \textsc{Guided-Measured-Continuous-Greedy}. This sub-routine is a version of the Measured
Continuous Greedy algorithm (Algorithm~\ref{alg:guided-mcg}) that is guided by a vector $\va \in [0, 1]^\cN$
(which it avoids).\footnote{The idea of guiding Measured Continuous Greedy by avoiding an input vector can be traced back to the work of~\citet{buchbinder2019constrained}.} In addition to the vector $\va$, Algorithm~\ref{alg:guided-mcg}
has three additional parameters:
the multilinear extension $F$ of a submodular
set function $f$, a stopping time $T \geq 0$, and a down-closed
solvable convex body $\cP$ representing the constraint. 
\begin{algorithm}[ht] 
\caption{\textsc{Guided-Measured-Continuous-Greedy}$(F, T, \va, \cP)$}\label{alg:guided-mcg}
\DontPrintSemicolon
Let $\vy{(0)} \gets \vzero$. \\
\For{every time $t$ between $0$ and $T$}
{
  Let $\vx{(t)} \gets \arg \max_{\vx \in \cP } \inner{\vx \hprod (\vone_\cN - \va \psum \vy{(t)})}{\nabla F(\vy{(t)})}$.\\
  Increase $\vy{(t)}$ at a rate $\frac{ \diff{ \vy{(t)} }}{\diff{t} } = \vx{(t)} \hprod (\vone_\cN - \va \psum \vy{(t)}) $.\\
}
\Return $\vy{(T)}$.
\end{algorithm}

We begin the analysis of Algorithm~\ref{alg:guided-mcg} with the following observation, which states properties that the vector $\vy(t)$ has, and are not related to the objective function.
\begin{observation} \label{obs:sane}
  At any time $t \in [0,T]$,
  $\vy{(t)} \le \vone_\cN - e^{-t(\vone_\cN - \va)} \leq t(\vone_\cN - \va)$ and $\vy{(t)} / t \in \cP$. 
\end{observation}
\begin{proof}
Since $\vx(t) \in \cP \subseteq [0, 1]^\cN$,
\[
	\frac{\diff{\vy(t)}}{\diff{t}}
	=
	\vx{(t)} \hprod (\vone_\cN - \va \psum \vy{(t)})
	\leq
	\vone_\cN - \va \psum \vy{(t)}
	=
	(\vone_\cN - \va) \hprod (\vone_\cN - \vy(t))
	\enspace.
\]
Thus, $\vy(t)$ is upper bounded by $\vone_\cN - e^{-t(\vone_\cN - \va)}$. 
To see why the second inequality of the lemma holds as well, notice that, for every $t \in [0, T]$,
  \[ \frac{\vy{(t)}}{t} = \frac{1}{t} \int_0^t \vx{(\tau)} \hprod (\vone_\cN - \va \psum \vy{(\tau)}) \text{d}\tau \le \frac{1}{t}\int_0^t \vx{( \tau )} \text{d} \tau \enspace. \]
  Since $\cP$ is convex and the vectors
  $\vx{( \tau )}$ belong to $\cP$, the right hand
  side of the above inequality belongs to $\cP$; and
  by down-closedness of $\cP$, so does $\vy{(t)}/t$. 
\end{proof}

Next, we prove two results about the rate in which the value $F(\vy(t))$ of the solution maintained by Algorithm~\ref{alg:guided-mcg} increases as a function of $t$. The first of these results is given by the next observation, which shows that $F(\vy(t))$ increases at a non-increasing rate.

\begin{observation} \label{obs:monotone-derivative}
  For every $t_1,t_2 \in [0,T]$ such that $t_1 \le t_2$,
  $\frac{ \diff{ F(\vy{(t_2)}) }}{\diff{t}} \le \frac{ \diff{ F(\vy{(t_1)})} }{\diff{t}}$. 
\end{observation}
\begin{proof}
Let $\vx'$ a vector $\vx$ maximizing $\max_{\vx \in \cP} \inner{ \nabla F( \vy{(t_2)} ) }{ \vx \hprod (\vone -\va \psum \vy{(t_2)}) }$. Notice that for every $u \in \cN$, $x'_u$ must be $0$ whenever $\frac{\partial F(\vz)}{\partial z_u}\big|_{\vz = \vy(t_2)} < 0$ because otherwise decreasing $x'_u$ to $0$ would increase the inner product $\inner{ \nabla F( \vy{(t_2)} ) }{ \vx' \hprod (\vone -\va \psum \vy{(t_2)}) }$ while keeping $\vx'$ in the convex set $\cP$ (recall that $\cP$ is down-closed). This observation justifies the first inequality in the following display math because $\vy(t_1) \leq \vy(t_2)$.
		\begin{align*}
			\max_{\vx \in \cP} \inner{ \nabla F( \vy{(t_2)} ) }{ \vx \hprod (\vone -\va \psum \vy{(t_2)}) }
			={} &
			\inner{ \nabla F( \vy{(t_2)} ) }{ \vx' \hprod (\vone -\va \psum \vy{(t_2)}) }\\
			={} &
			\sum_{u \in \cN} \left.\frac{\partial F(\vz)}{\partial z_u}\right|_{\vz = \vy(t_2)} \cdot x'_u \cdot (1 - a_u \psum y_u(t_2))\\
			\leq &
			\sum_{u \in \cN} \left.\frac{\partial F(\vz)}{\partial z_u}\right|_{\vz = \vy(t_2)} \cdot x'_u \cdot (1 - a_u \psum y_u(t_1))\\
			={} &
			\inner{ \nabla F( \vy{(t_2)} ) }{ \vx' \hprod (\vone -\va \psum \vy{(t_1)}) }\\
			\leq{} &
			\max_{\vx \in \cP} \inner{ \nabla F( \vy{(t_2)} ) }{ \vx \hprod (\vone -\va \psum \vy{(t_1)}) }
			\enspace,
		\end{align*}
Using the last inequality, we now get
    \begin{align*}
      \frac{ \diff{ F( \vy{(t_2)} ) } }{\diff{t} } &= \Big\langle \nabla F( \vy{(t_2)} ) , \frac{\diff{\vy{(t_2)}}}{\diff{t} } \Big\rangle \\
															 &= \inner{\nabla F( \vy{(t_2)} ) }{ \vx(t_2) \hprod (\vone -\va \psum \vy{(t_2)}) } \\
                               &= \max_{\vx \in \cP} \inner{ \nabla F( \vy{(t_2)} ) }{ \vx \hprod (\vone -\va \psum \vy{(t_2)}) } \\
															 &\leq \max_{\vx \in \cP} \inner{ \nabla F( \vy{(t_2)} ) }{ \vx \hprod (\vone -\va \psum \vy{(t_1)}) } \\
															 &\leq \max_{\vx \in \cP} \inner{ \nabla F( \vy{(t_1)} ) }{ \vx \hprod (\vone -\va \psum \vy{(t_1)}) } \\
															 &= \inner{\nabla F( \vy{(t_1)} ) }{ \vx(t_1) \hprod (\vone -\va \psum \vy{(t_2)}) } \\
                               &= \Big\langle \nabla F( \vy{(t_1)} ) , \frac{ \diff{ \vy{(t_1)} } }{\diff{t} } \Big\rangle\\
                               &= \frac{ \diff{ F( \vy{(t_1)} ) } }{\diff{t} }\enspace,
    \end{align*}
    where the first and last equalities hold by the chain rule, the third and fourth equalities hold due to the method used by Algorithm~\ref{alg:guided-mcg} to choose $\vx(t)$, and the second inequality holds because the DR-submodularity of $F$ and the inequality $\vy(t_1) \leq \vy(t_2)$ imply together $\nabla F(\vy(t_1)) \geq \nabla F(\vy(t_2))$.
  \end{proof}
	
The next observation lower bounds the rate in which $F(\vy(t))$ increases for a particular value of $t$.
\begin{observation} \label{obs:guided-deriv-bound}
  For any vector $\vb \in [0, 1]^\cN$ obeying $\va \leq \vb \leq \vone_\cN$, 
  vector $\vz \in \cP$ and value $t \in [0,T]$,
  it holds that
  \[ \frac{ \diff{ F( \vy{(t)} ) } }{\diff{t} } \ge \FF{ \vy{(t)} \psum ((\vone_\cN - \vb ) \hprod \vz )} - \FF{ \vy{(t)}} \enspace.\]
\end{observation}
\begin{proof}
  Let $\vz'$ be the vector obeying $(\vone_\cN - \va) \hprod \vz' = (\vone_\cN - \vb) \hprod \vz$. Since $\vb \geq \va$, it must hold that $\vz' \leq \vz$, and thus, $\vz \in \cP$ by the down-closedness of $\cP$. Observe now that
  \begin{align*}
    \frac{\diff{ F(\vy{(t)})} }{\diff{t} } &= \inner{ \nabla F( \vy{(t)} ) }{ \vx{(t)} \hprod (\vone_\cN - \va \psum \vy{(t)}) } \\
                            &\geq \inner{ \nabla F( \vy{(t)} ) }{ \vz' \hprod (\vone_\cN - \va \psum \vy{(t)}) } \\
														&= \inner{\nabla F( \vy{(t)} ) }{ \vz' \hprod (\vone_\cN - \va) \hprod (\vone_\cN - \vy(t))}\\
														&\geq F(\vy{(t)} \psum (\vz' \hprod (\vone_\cN - \va))) - F(\vy(t))\\
														&= F(\vy{(t)} \psum ((\vone_\cN - \vb) \hprod \vz)) - F(\vy(t))
														\enspace,
  \end{align*}
  where the first equality hold by the chain rule, the first
  inequality holds by the greedy selection of $\vx{(t)}$ since
  $\vz' \in \cP$, and the last inequality holds by
  Lemma~\ref{lem:positive_direction}.
\end{proof}

Combining the two above results, we can derive a lower bound on the value of the output set of Algorithm~\ref{alg:guided-mcg} for $T = 2$ (which is the only value of $T$ that our full algorithm from Section~\ref{ssc:full_algorithm} uses).
  \begin{lemma} \label{lemma:gmcg-main} 
  If $T = 2$, then, for any vector $\vb \in [0, 1]^\cN$ obeying $\va \leq \vb \leq \vone_\cN$ and
  vector $\vz \in \cP$, it holds that
    \[ \frac{1}{2}\FF{ {\vy}{(2)} } \ge \FF{{\vy}{(2)} \oplus ((\vone_\cN - \vb) \hprod \vz) } - \FF{ {\vy}{(2)}} \enspace.\]
  \end{lemma}
  \begin{proof}
	By the non-negativity of $F$,
    \begin{align*}
      \FF{ {\vy}{(2)} } \ge F( \vy{(2)} ) - F( \vy{(0)} ) &= \int_0^2 \frac{\diff{ F(\vy{(t)})}}{\diff{t}} \text{d}t \\
																															&\geq 2 \cdot \frac{\diff{ F(\vy{(2)}) }}{\diff{t}}\\
                                                               &\ge 2[ \FF{{\vy}{(2)} \oplus ((\vone_\cN - \vb) \hprod \vz) } - \FF{ {\vy}{(2)}}].
    \end{align*}
    where the second
    inequality follows from Observation~\ref{obs:monotone-derivative},
    and the third inequality from Observation~\ref{obs:guided-deriv-bound}.
  \end{proof}
	
  \subsubsection{Bicriteria Algorithm} \label{ssc:full_algorithm}
	
	In this section, we prove Theorem~\ref{thm:general_down-closed_convex_body} by presenting Algorithm~\ref{alg:gen-bicriteria}, and showing that this algorithm has the properties guaranteed by the theorem. Algorithm~\ref{alg:gen-bicriteria} uses \textsc{Guided-Measured-Continuous-Greedy} (Algorithm~\ref{alg:guided-mcg}) from Section~\ref{ssc:guided-mcg} as a sub-routine.
	
  \begin{algorithm}[ht] 
\caption{\textsc{General-Bicriteria-Alg}$(F, \cP, \eps)$}\label{alg:gen-bicriteria}
\DontPrintSemicolon
Let $\vg \gets \vzero$ and $\ell = \lceil \frac{1}{\eps} \rceil$. \\
\For{$i = 1$ to $\ell$}
{
  Let $\vg^{(i)} \gets $\textsc{Guided-Measured-Continuous-Greedy}$(F, 2, \vg , \cP)$. \\
  Update $\vg \gets ( \vg + \vg^{(i)}) \wedge \vone_\cN$.  \\
}
\For{$i = 1$ to $\ell$}
{
  Let $G^{(i)}$ be the function $\vx \mapsto \FF{ \vg^{(i)} \psum ( \vg \hprod \vx ) }$. \\
  Let $\vd^{(i)} \gets \textsc{DoubleGreedy}(G^{(i)})$.
}
Let $i^*$ be any index $i \in [\ell]$ maximizing $\FF{ \vg^{(i)} \psum ( \vg \hprod \vd^{(i)} )}$.\\
\Return $\vg^{(i^*)} \psum ( \vg \hprod \vd^{(i^*)} )$.
\end{algorithm}

We begin the analysis Algorithm~\ref{alg:gen-bicriteria} with the following lemma, which is a consequence of Lemma~\ref{lemma:prob-lemma}.
\begin{lemma} \label{lemm:prob-vectors}
  Let $X = \{ \vx^{(1)}, \vx^{(2)}, \ldots, \vx^{(m)} \}$ be a finite set of $m$ vectors in $[0,1]^{\cN}$,
  and let $g\colon 2^\cN \to \nnR$ be a non-negative submodular set function whose multilinear extension is denoted by $G$.
  Then
  \[ \frac{1}{m} \sum_{i=1}^m G( \vx^{(i)} ) \ge \bigg( 1 - \frac{\|\sum_{i=1}^m \vx^{(i)}\|_\infty}{m} \bigg) \cdot G( \vzero ) \enspace.\]
\end{lemma}
\begin{proof}
  Let $\vs$ be a uniformly at random sample from $X$. Then,
  \begin{equation} \label{eq:expectation_expansion}
    \ex{ g(\RSet(\vs)) } = \sum_{i=1}^m \ex{ g(\RSet(\vs)) \mid \vs = \vx^{(i)} } \cdot \Pr \left[ \vs = \vx^i \right]
                = \frac{1}{m}\sum_{i=1}^m G \left( \vx^{(i)} \right)\enspace. 
  \end{equation}
  Moreover, for each $u \in \cN$, $\Pr[ u \in \RSet(\vs) ] = \frac{1}{m}\sum_{i=1}^m x^{(i)}_u \leq \frac{1}{m}\|\sum_{i = 1}^m \vx^{(i)}\|_\infty$.
  Therefore, by Lemma~\ref{lemma:prob-lemma},
  \[
		\ex{ g(\RSet(\vs)) } \ge \bigg( 1 - \frac{\|\sum_{i = 1}^m \vx^{(i)}\|_\infty}{m} \bigg) \cdot g( \varnothing ) =
  \bigg( 1 - \frac{\|\sum_{i = 1}^m \vx^{(i)}\|_\infty}{m}\bigg) \cdot G( \vzero ) \enspace.
	\]
	The lemma now follows by combining this inequality with Equation~\eqref{eq:expectation_expansion}.
\end{proof}

In the following, we use Lemma~\ref{lemm:prob-vectors} with each vector $\vx^{(i)}$ being $\vg^{(i)}$. The next observation bounds the value of the infinity-norm in this use.

\begin{observation} \label{obs:g} 
It holds that $\|\sum_{i=1}^\ell \vg^{(i)}\|_\infty \le 2$.
\end{observation}
\begin{proof}
Fix an arbitrary element $u \in \cN$. We show below that $\sum_{i=1}^\ell g_u^{(i)} \leq 2$. Let $i'$ be the last value in $\{0, 1, \dotsc, \ell\}$ such that $\sum_{i = 1}^{i'} g^{(i)}_u < 1$. If $i' = \ell$, then we are done. Otherwise, we make two observations:
\begin{itemize}
	\item The vector $\va$ passed to execution number $i' + 1$ of \textsc{Guided-Measured-Continuous-Greedy} has the value $\sum_{i = 1}^{i'} g_u^{(i)}$ in the coordinate corresponding to element $u$, and therefore, by Observation~\ref{obs:sane}, $g^{(i' + 1)}_u \leq 2(1 - a_u) = 2(1 - \sum_{i - 1}^{i'} g_u^{(i)})$.
	\item The vector $\va$ passed to later executions of \textsc{Guided-Measured-Continuous-Greedy} has the value $1$ in the coordinate corresponding to element $u$, and therefore, by Observation~\ref{obs:sane}, $g^{(i)}_u \leq 2(1 - a_u) = 0$ for every $i > i'+ 1$.
\end{itemize}
Combining the above observations yields
\[
	\sum_{i=1}^\ell g_u^{(i)}
	=
	\sum_{i=1}^{i'} g_u^{(i)} + g_u^{(i')}
	\leq
	\sum_{i=1}^{i'} g_u^{(i)} + 2\bigg(1 - \sum_{i=1}^{i'} g_u^{(i)}\bigg)
	\leq
	2
	\enspace.
	\qedhere
\]
\end{proof}

\begin{corollary} \label{cor:low_overlap}
For every vector $\vz \in [0, 1]^\cN$,
\[
	\frac{1}{\ell} \sum_{i=1}^\ell F( \vz \psum \vg^{(i)} )
	\geq
	\left( 1 - \frac{2}{\ell }\right) \cdot F( \vz )
	\enspace.
\]
\end{corollary}
\begin{proof}
Let $g$ be the submodular set function $g(S) = \FF{ \vz \psum \vone_S}$,
  and let $G$ be its multilinear extension. Note that, for every vector $\vs \in [0, 1]^\cN$, it holds that
	\[
		G(\vs) = \bE[g(\RSet(\vs))] = \bE[F(\vz \psum \vone_{\RSet(\vs)})]
		=
		\bE[f(\RSet(\vz \psum \vone_{\RSet(\vs)}))]
		=
		\bE[f(\RSet(\vz \psum \vs))]
		=
		F(\vz \psum \vs)
		\enspace.
	\]
	Thus, by Lemma~\ref{lemm:prob-vectors} and Observation~\ref{obs:g},
	\[
		\frac{1}{\ell} \sum_{i=1}^\ell F( \vz \psum \vg^{(i)} )
		=
		\frac{1}{\ell} \sum_{i=1}^\ell G( \vg^{(i)} )
		\ge
		\left( 1 - \frac{2}{\ell }\right) \cdot G( \vzero )
		=
		\left( 1 - \frac{2}{\ell }\right) \cdot F( \vz )
		\enspace.
		\qedhere
	\]
\end{proof}

We need one more technical lemma about DR-submodular functions.
\begin{lemma} \label{lem:dr-submodularity-addition}
Given a differentiable DR-submodular function $G\colon [0, 1]^\cN \to \bR$ and two vectors $\va, \vb \in [0, 1]^\cN$ obeying $\va + \vb \leq \vone_\cN$, it holds that $F(\va) + F(\vb) \geq F(\vzero) + F(\va + \vb)$.
\end{lemma}
\begin{proof}
By the chain rule and the DR-submodularity of $F$,
	\begin{align*}
		F(\va) + F(\vb)
		={} &
		2F(\vzero) + \int_0^1 \langle \va, \nabla F(r \cdot \va) \rangle \text{d}r
		+ \int_0^1 \langle \vb , \nabla F(r \cdot \vb)\rangle \text{d}r\\
		\geq{} &
		2F(\vzero) + \int_0^1 \langle \va , \nabla F(r \cdot (\va + \vb))\rangle \text{d}r
		+ \int_0^1 \langle \vb , \nabla F(r \cdot (\va + \vb)) \rangle \text{d}r\\
		={} &
		2F(\vzero) + \int_0^1 \langle \va + \vb , \nabla F(r \cdot (\va + \vb))\rangle \text{d}r
		=
		F(\vzero) + F(\va + \vb)
		\enspace.
		\qedhere
	\end{align*}
\end{proof}

We are now ready to prove our guarantee for Algorithm~\ref{alg:gen-bicriteria}. Notice that one can get Theorem~\ref{thm:general_down-closed_convex_body} from Theorem~\ref{thm:gen-bicriteria} by simply dividing the value of $\eps$ passed to Algorithm~\ref{alg:gen-bicriteria} by the multiplicative constant that is hidden by the big $O$ notation in $O(\eps)$.

\begin{theorem} \label{thm:gen-bicriteria}
  For any $\eps \in (0, 1/2)$, Algorithm~\ref{alg:gen-bicriteria} returns a vector $\vr$ such that
  $\frac{\vr}{2(\eps^{-1} + 2)} \in \cP$ and,
	for any $\vz \in \cP$, $F( \vr ) \ge (1/2 - O(\eps) )F( \vz )$.
\end{theorem}
\begin{proof}
  By Observation \ref{obs:sane}, for each $i \in [\ell]$, $\vg^{(i)}/2 \in \cP$.
  Therefore, because $\cP$ is convex, $\frac{\vg^{(i^*)} + \vg}{2(\ell + 1)} = \frac{\sum_{i=1}^\ell (\vg^{(i)} / 2) + (\vg^{(i^*)} / 2)}{\ell + 1} \in \cP$. Since $\cP$ is down-closed, this implies that $\frac{\vr}{2(\eps^{-1} + 2)}$ also belongs to $\cP$ because
	\[
		\frac{\vr}{2(\eps^{-1} + 2)}
		\leq
		\frac{\vr}{2(\ell + 1)}
		=
		\frac{\vg^{(i^*)} \psum (\vg \hprod \vd^{(i^*)})}{2(\ell + 1)}
		\leq
		\frac{\vg^{(i^*)} + \vg}{2(\ell + 1)}
		\enspace.
	\]

The rest of the proof is devoted to lower bounding $F(\vr)$. Let $\hat \vz \triangleq (\vone_\cN - \vg) \hprod \vz$ and
  $\dot \vz \triangleq \vg \hprod \vz$.
	Each time Algorithm~\ref{alg:gen-bicriteria} invokes the sub-routine \textsc{Guided-Measured-Continuous-Greedy}, it passes $\vg$ as the value for the parameter $\va$ of this sub-routine. Since $\vg$ only increases during the execution of Algorithm~\ref{alg:gen-bicriteria}, its final value is (coordinate-wise) at least as large as the value passed as $\va$ in each such invocation. Thus, by Lemma~\ref{lemma:gmcg-main},
  for each $i \in [\ell]$,
  \begin{equation} \label{eq:guided-mcg}
    \frac{1}{2} F( \vg^{(i)} ) \ge F( \vg^{(i)} \psum ((\vone_\cN - \vg) \hprod \vz ) - F( \vg^{(i)} )
		=
		F( \vg^{(i)} \psum \hat \vz ) - F( \vg^{(i)} ) \enspace. 
  \end{equation}

	Consider now the function $G^{(i)}( \vx )  \triangleq \FF{ \vg^{(i)} \psum (\vg \hprod \vx ) }$ defined by Algorithm \ref{alg:gen-bicriteria} for some value $i \in [\ell]$. By Lemma~\ref{lem:closure}, $G^{(i)}$ is non-negative and DR-submodular. Thus, we can use
	Theorem \ref{thm:doublegreedy} to lower bound the value of the vector $d^{(i)}$ obtained by invoking \textsc{DoubleGreedy}$(G^{(i)})$. Specifically, setting $r = 1$ in the guarantee of this theorem yields
	\begin{multline*}
		F(\vg^{(i)} \psum (\vg \hprod \vd^{(i)}))
		=
		G^{(i)}(\vd^{(i)})
		\geq
		\bigg(\frac{1}{2} - O(\eps)\bigg) \cdot G^{(i)}(\vz) + \frac{1}{4} \cdot G^{(i)}(\vzero) + \frac{1}{4} \cdot G^{(i)}(\vone_\cN)\\
		\geq
		\bigg(\frac{1}{2} - O(\eps)\bigg) \cdot F(\vg^{(i)} \psum (\vg \hprod \vz)) + \frac{1}{4} \cdot F(\vg^{(i)})
		=
		\bigg(\frac{1}{2} - O(\eps)\bigg) \cdot F(\vg^{(i)} \psum \dot \vz) + \frac{1}{4} \cdot F(\vg^{(i)})
		\enspace,
	\end{multline*}
	where the second inequality follows from the non-negativity of $F$. Rearranging this inequality, and using the non-negativity of $F$ again, gives
  \begin{equation} \label{eq:usm}
    (2 + O(\eps)) \cdot F(\vg^{(i)} \psum (\vg \hprod \vd^{(i)}))
		\geq
		F(\vg^{(i)} \psum \dot \vz) + \frac{1 + O(\eps)}{2} \cdot F(\vg^{(i)})
		\geq
		F(\vg^{(i)} \psum \dot \vz) + \frac{1}{2} \cdot F(\vg^{(i)})
		\enspace.
  \end{equation}

  We are now ready to prove the promised lower bound on $F(\vr)$.
  \begin{align*}
    ( 1 - 2\eps )\FF{ \vz } &\le \bigg( 1 - \frac{2}{\ell }\bigg)\FF{ \vz }  \\
                                                &\le \frac{1}{\ell} \sum_{i=1}^\ell F( \vz \psum \vg^{(i)} ) \\
                                                &\le \frac{1}{\ell} \sum_{i=1}^\ell [F(\hat \vz \psum \vg^{(i)}) + F(\dot \vz \psum \vg^{(i)}) - F(\vg^{(i)})] \\
                                                &\le \frac{1}{\ell} \sum_{i=1}^\ell \Big[\frac{1}{2} F( \vg^{(i)} ) + F(\dot \vz \psum \vg^{(i)})\Big] \\
                                                &\le \frac{1}{\ell} \sum_{i=1}^\ell [(2 + O( \eps )) \cdot F( \vg^{(i)} \psum ( \vg \hprod \vd^{(i)} ))]\\
																								&\le (2 + O( \eps )) \cdot F( \vg^{(i^*)} \psum ( \vg \hprod \vd^{(i^*)} ))
																								=(2 + O( \eps )) \cdot \FF{ \vr }
																								\enspace,
  \end{align*}
	where the first inequality holds since $\eps = \frac{1}{\eps^{-1}} \geq \frac{1}{\ell}$, the second inequality follows from Corollary~\eqref{cor:low_overlap}, the third inequality follows from Lemma~\ref{lem:dr-submodularity-addition} by setting $\va = \dot \vz$ and $\vb = \hat{\vz}$ since $F(\cdot \psum g^{(i)})$ is a DR-submodular function by Lemma~\ref{lem:closure}, the fourth inequality follows from Inequality~\eqref{eq:guided-mcg}, the penultimate inequality follows from Inequality~\eqref{eq:usm}, and the last inequality holds by the choice of $i^*$.
\end{proof}

\subsection{Combinatorial and Deterministic Algorithms for Knapsack and Matroid Constraints} \label{ssc:combinatorial_general}

In this section, we present combinatorial and/or deterministic
bicriteria algorithms with a bicriteria approximation ratio of $(1/2 - \eps, O( \eps^{-1} ))$ for both the matroid-constrained and
knapsack-constrained submodular maximization problems.
Algorithm~\ref{alg:comb-gen-knapsack} summarizes all these 
algorithms. As in the warmup in Section \ref{sec:warmup},
$\ell$ disjoint sets $A_1, \ldots, A_{\ell}$
are greedily constructed. 
The greedy algorithm used for this purpose as a subroutine depends on
the constraint type. For matroid constraints,
\textsc{Iterative Matroid Greedy} (Algorithm \ref{alg:greedy_matroid}) is invoked,
while for knapsack constraints, $\densitygreedy$ (Algorithm \ref{alg:greedy_density})
is used. 
Then, a variant of \textsc{DoubleGreedy} is used to best extend
each set $A_i$ using elements of $A = \bigcup_{i \in [\ell]} A_i$.
The choice of the exact algorithm used to
implement \textsc{DoubleGreedy} determines whether the overall algorithm
is deterministic or combinatorial: the original \textsc{DoubleGreedy}
of \citet{buchbinder2015tight} is randomized and combinatorial,
while the derandomized version of \citet{buchbinder2018deterministic}
is deterministic but requires solving linear programs. For simplicity, we assume in the analysis below that the deterministic algorithm is used. If this is not the case, then appropriate expectation signs should be added throughout the analysis.

 To be more concrete, let $f\colon 2^\cN \to \nnR$ be the non-negative submodular objective function over a ground set $\cN$. If Algorithm~\ref{alg:comb-gen-knapsack} is invoked for a knapsack constraints, it gets as input a positive vector $\vp$ of element prices\footnote{In this section, we assume for simplicity that the price of every element is strictly positive. It is possible to drop this assumption, but this would requires some extra care in both the algorithm and the analysis. Since this extra care is relatively straight forward, and would make understanding of our main new ideas more difficult, we prefer to avoid it.} and a budget $B \geq 0$, and its goal is to find a set that (approximately) maximizes $\max_{S \subseteq \cN : \sum_{e \in S} p_e \le B} f(S)$. Otherwise, if Algorithm~\ref{alg:comb-gen-knapsack} gets a matroid $\cM = (\cN, \cI)$, its goal is to find an independent set of this matroid that (approximately) maximizes $f$.
Algorithm~\ref{alg:comb-gen-knapsack} often uses the notation $f|_S$ to denote the restriction of $f$ to a ground set $S \subseteq \cN$.
Additionally, Algorithm~\ref{alg:comb-gen-knapsack} implicitly assumes that dummy elements are added
to the ground set $\cN$,
in the following fashion.
A dummy element $d$ does not change the
value of any set $S$;
that is $f(S) = f(S \setminus \{ d \} )$.
In the case of a
knapsack constraint $B$,
we set the price of the dummy element $d$ to $p_d = B$.
In the case of a matroid constraint,
adding a dummy element $d$ to a feasible set
does not violating
feasibility if the size of the resulting set does not exceed the rank of the matroid. One may verify that the
augmented set system obtained in this way
is still a matroid. The number of dummy elements that are added should be large enough
to ensure that each element selected by the greedy algorithm used
has non-negative marginal gain (Observation~\ref{obs:comb:dummies} formally determines this number).
To use Algorithm~\ref{alg:comb-gen-knapsack}, one should add the appropriate number of dummy elements to the ground set, execute the algorithm, and then simply delete
any dummy elements from the set returned by the algorithm. Notice that such deletion does not affect the value of the output set.

\begin{algorithm}[ht]
   \caption{\textsc{Combinatorial $(1/2 - \eps, O(1/\eps ))$-Bicriteria Algorithm for Knapsack/Matroid Constraint}} \label{alg:comb-gen-knapsack}
   \DontPrintSemicolon
   \textbf{Input:} function $f \colon 2^\cN \to \nnR$, error parameter $\eps > 0$, and \textbf{either} matroid $\cM = (\cN, \mathcal I )$; \textbf{or} both knapsack budget $B$ and element prices $\vp$ \\
   Let $A \gets \varnothing$ and $\ell = \lceil \frac{1}{2\eps } \rceil$.\\
	 Add in sufficiently many dummy elements as specified by Observation~\ref{obs:comb:dummies}. \\
   \For{ $i = 1$ to $\ell$ }
   {
     \textcolor{teal}{[The flavor of combinatorial greedy algorithm is selected based upon the constraint. Parameters are set to produce a solution
       infeasible by a factor of at least $2$ and at most $3$.]} \\
     \If{ \textit{knapsack constraint} }
     {
       $A_i \gets \densitygreedy ( f |_{\cN \setminus A}, \vp, 2B, 1/e )$ \textcolor{teal}{(Algorithm \ref{alg:greedy_density})}\\
     } \Else {
       $A_i \gets \textsc{Iterative Matroid Greedy} ( f |_{\cN \setminus A}, \cM, 1/4 )$ \textcolor{teal}{(Algorithm \ref{alg:greedy_matroid})}\\
     }
     $A \gets A \cup A_i$} 
   \For{ $i = 1$ to $\ell$ }
   {
     Let $g_i$ be function $g_i(S) = f( A_i \cup S )$.\\

     $D_i \gets \textsc{DoubleGreedy} ( g_i |_{A} ) $}
   \Return the set in $\{D_i \cup A_i \mid i \in [\ell]\}$ maximizing $f$.
 \end{algorithm}

\paragraph{Analysis of Algorithm~\ref{alg:comb-gen-knapsack}.} In the following,
let $A,A_i, D_i$ (for every $i \in [\ell]$) denote their respective values at termination
of Algorithm \ref{alg:comb-gen-knapsack}. 
\begin{observation} \label{obs:comb:dummies}
  For the case of the knapsack constraint,
  add $2\ell$ dummy elements to the ground set.
  For the case of the matroid constraint, add
  $2\ell k $ dummy elements, where $k$ is the
  rank of the matroid. Note that at most $O(n/\eps)$ dummy elements are added in either case.
	
  For every $i \in [\ell]$ and element $a \in A_i$,
  let $A_{i,a}$
  be the subset of $A_i$ already chosen by
  the greedy algorithm when $a$ is selected.
  Given that the above dummy elements have been added, it is guaranteed to hold that
  \begin{enumerate}
    \item the marginal gain $\marge{a}{ A_{i,a} } \ge 0$.
    \item the sets $A_i$ are large in the following
sense. For a knapsack constraint, $c( A_i ) \ge 2B$;
and for a matroid constraint, $A_i$ is the union of
two disjoint bases of the matroid.
\end{enumerate}
\end{observation}
\begin{proof}
  The greedy algorithm employed by Algorithm
  \ref{alg:comb-gen-knapsack} depends on the constraint. For either
  version, observe that at least one dummy element always remains as a candidate
  for the greedy algorithm because:
  for a knapsack constraint, at most two dummy elements are selected
  into each set $A_i$, and for a matroid constraint, at most $2k$ dummy
  elements are selected into each set $A_i$.
	This immediately implies $1$.
	In the case of a knapsack constraint, it also means that Line~\ref{line:early_termination} of Algorithm~\ref{alg:greedy_density} is never invoked, and thus, $c(A_i) \ge 2B$. In the case of a matroid constraint, it means that the loop on Line~\ref{line:greedy_matroid_loop} of Algorithm~\ref{alg:greedy_matroid} does not terminate before $T_{i, j}$ becomes a base, and thus, $A_i$ is the union of two disjoint bases of the matroid. Thus, $2$ holds as well.
\end{proof}

Observe that $A_i \cap A_j = \varnothing$ for all $i \neq j \in [\ell]$,
and $A = \cup_{i \in [\ell]} A_i$. We often use below the notation $A_{i, a}$ introduced by Observation~\ref{obs:comb:dummies}. 
When $i$ is understood
from the context in this notation, we sometimes drop it and write $A_a$. 
As usual, let $\opt \subseteq \cN$ be an optimal solution to
the problem, and let $\hat O \triangleq O \setminus A$, $\dot O \triangleq O \cap A$. 
The proofs for knapsack and matroid constraints
differs only slightly; in each case, we partition
$A_i$ into two parts.
For a matroid constraint, the partition is
entirely natural: $A_i = E_1\cupdot E_2$,
where $E_1$ and $E_2$ are disjoint bases of the matroid,
and every element of $E_1$ is selected before those
in $E_2$ (such a partition exists by the definition of Algorithm~\ref{alg:greedy_matroid}). For a knapsack constraint, a little
bit more work is needed. 
For the purposes of the analysis only,
we may \textit{split}
an item $s \in \cN$ into two items.
Formally, this is defined as follows.
\begin{definition} \label{def:split}
  Let $f$ be a non-negative submodular function and $\vp$ be a positive price vector over $\cN$.
  Given an element $s \in \cN$ and
  value $\gamma \in (0,1)$,
  a $\gamma$-split of $s$ produces a new function $f^*$ and price vector $\vp^*$ as follows.
  Let $\cN^* = \cN \cup \{\sigma_1, \sigma_2 \} \setminus s $.
  The function $f^*\colon 2^{\cN^*} \to \nnR$ is defined for
  every set $S \subseteq \cN^*$ by
  \[f^*(S) = f( \widehat S ) + \gamma \cdot f(s \mid \widehat S ) \cdot \mathbb{I} \left[ \sigma_1 \in S \right] + (1 - \gamma) \cdot f(s \mid \widehat S ) \cdot \mathbb{I} \left[ \sigma_2 \in S \right] \enspace, \]
  where $\widehat{S} \triangleq S \setminus \{\sigma_1, \sigma_2 \}$, and $\mathbb I$ takes the value $1$ if its argument is true, and the value $0$ otherwise. 
  The price vector $\vp^*$ is defined by $p^*_x = p_x$ for every $x \in \cN^* \setminus \{\sigma_1, \sigma_2\}$,
  $p^*_{\sigma_1} = \gamma p_s$, and $p^*_{ \sigma_2 } = (1 - \gamma ) p_s$.
\end{definition}
The next observation relates the marginal gain of a split item to the marginal gain of the original item.
It follows directly from Definition~\ref{def:split}.
\begin{observation} \label{obs:split-1}
  Consider a $\gamma$-split of $s \in \cN$. 
  For all sets $S \subseteq \cN^*$,
  1) if $\sigma_1 \not \in S$, then
  $\margen{f^*}{\sigma_1}{S} = \gamma \margen{f}{s}{\widehat S}$;
  and 2) if $\sigma_2 \not\in S$, then $\margen{f^*}{\sigma_2}{S} = (1 - \gamma ) \margen{f}{s}{\widehat S}$.
\end{observation}
The next lemma shows that the function $f^*$ resulting from a $\gamma$-split is submodular like the original function $f$. 
\begin{lemma} \label{lemm:split-submod}
  Consider a $\gamma$-split of $s \in \cN$
  into $\{ \sigma_1, \sigma_2 \}$. 
  The resulting function $f^*$ is submodular.
\end{lemma}
\begin{proof}
  We use in this proof the following characterization of submodularity (due to \citep{zotero-3825}):
  $f^*$ is submodular if and only if,
  for all sets $S \subseteq \cN^*$ and elements
  $x_1, x_2 \in \cN^* \setminus S$,
  \begin{equation} \label{eq:sm1}
    f^*( S + x_1 ) + f^*( S + x_2 ) \ge f^*( S + x_1 + x_2 ) + f^*( S ) \enspace.
  \end{equation}
  Notice that Inequality~\eqref{eq:sm1} can also be written as
  $\margen{f^*}{ x_2 }{ S + x_1 } \le \margen{f^*}{ x_2 }{ S }$.
  To proceed, fix a set $S \subseteq \cN^*$ and elements $x_1, x_2 \in \cN^* \setminus S$, and recall that
	$\widehat{T}$ denotes the set $T$ without the items $\sigma_1$ and $\sigma_2$ resulting from the split. 
  We consider below several cases.

  \textbf{Case $x_2 \in \{ \sigma_1 ,\sigma_2 \}$.}
  Suppose $x_2 = \sigma_1$.
  \begin{align*}
    \margen{f^*}{ x_2}{ S + x_1 } &= \gamma \cdot \margen{f}{s}{ \widehat{S + x_1} } \\
    &\le \gamma \cdot \margen{f}{s}{ \widehat{S} } = \margen{f^*}{ x_2}{ S }\enspace,
  \end{align*}
  where the first and last equalities are by Observation \ref{obs:split-1},
  and the inequality is from the submodularity of $f$. 
  The argument for $x_2 = \sigma_2$ is exactly analogous, and we omit it.

  \textbf{Case $x_2 \not \in \{\sigma_1, \sigma_2\}$ and $x_1 \in \{ \sigma_1, \sigma_2 \}$.}
  Suppose $x_1 = \sigma_1$ (the argument for $x_1 = \sigma_2$ is analogous).
  Then
  \begin{align*}
    \margen{f^*}{ x_2}{ S + \sigma_1 } &= f^*(S + \sigma_1 + x_2) - f^*( S + \sigma_1 ) \\
                                       &= \margen{f^*}{ \sigma_1 }{ S + x_2} + f^*( S + x_2 ) - [ \margen{f^*}{\sigma_1}{S} + f^*(S) ] \\
                                       &= [ \margen{f^*}{ \sigma_1 }{ S + x_2} - \margen{f^*}{\sigma_1}{S} ] + \margen{f^*}{x_2}{S} \\
                                       &\le \margen{f^*}{x_2}{S}\enspace,
  \end{align*}
  where the inequality follows from the previous case.
	
	\textbf{Case $x_1, x_2 \not \in \{\sigma_1, \sigma_2\}$.}
  We consider four separate subcases here, depending on the intersection of $S$ with
  $\{ \sigma_1, \sigma_2 \}$.

  \textbf{Subcase $\sigma_1, \sigma_2 \in S$.}
  Then
  \begin{align*}
  \margen{f^*}{ x_2}{ S + x_1 } &= f^*(S + x_1 + x_2) - f^*( S + x_1 ) \\
                                &= f(\hat S + x_1 + x_2 + s) - f( \hat S + x_1 + s ) \\
                                &\le f(\hat S + x_2 + s) - f( \hat S + s ) \\
                                &= \margen{f^*}{x_2}{S} \enspace,
  \end{align*}
  where the second and third equalities are from the definition of $f^*$, and the inequality is by the submodularity of $f$.

  \textbf{Subcase $S \cap \{ \sigma_1, \sigma_2 \} = \{ \sigma_1 \}$.}
  Then
  \begin{align*}
    \margen{f^*}{ x_2}{ S + x_1 } &= f^*(S + x_1 + x_2) - f^*( S + x_1 ) \\
                                &= \margen{f^*}{ \sigma_1 }{ \hat S + x_1 + x_2 } + f^*( \hat S + x_1 + x_2 )
                                  - \margen{f^*}{ \sigma_1 }{ \hat S + x_1 } - f^*( \hat S + x_1 ) \\
                                &= \gamma \left[ \margen{f}{s}{ \hat S + x_1 + x_2 } - \margen{f}{s}{\hat S + x_1 } \right] + \margen{f^*}{x_2}{\hat S + x_1 } \\
                                &= \gamma \left[ \margen{f}{x_2}{\hat S + x_1 + s} - \margen{f}{x_2}{ \hat S + x_1 } \right] + \margen{f^*}{x_2}{\hat S + x_1 } \\
                                &= \gamma \margen{f}{x_2}{\hat S + x_1 + s} + (1 - \gamma) \margen{f}{x_2}{\hat S + x_1 } \\
                                &\le \gamma \margen{f}{x_2}{\hat S + s} + (1 - \gamma) \margen{f}{x_2}{\hat S } \\
                                &= \margen{f^*}{ x_2}{ S } \enspace, 
  \end{align*}
  where the first four equalities are from the definition of marginal gain, the definition of $f^*$, and the assumption on $S$;
  the fifth equality is by the fact that none of the sets involved contain $\sigma_1$ or $\sigma_2$, so $f = f^*$ on these sets;
  the inequality is from submodularity of $f$, and the last equality is by similar reasoning as the first four equalities.

  \textbf{Subcase $S \cap \{ \sigma_1, \sigma_2 \} = \{ \sigma_2 \}$.} This subcase is exactly analogous to the previous subcase.

  \textbf{Subcase $S \cap \{ \sigma_1, \sigma_2 \} = \varnothing$.} In this subcase, none of the sets involved contain $\sigma_1$
  or $\sigma_2$, so it immediately follows from the submodularity of $f$. 
\end{proof}
Having defined as studied the split operation, we proceed to formally define
the partition of $A_i$ mentioned above. 
\begin{definition} \label{def:partition-ai}
  For every $i \in [ \ell ]$,
  we partition $A_i$ into two disjoint parts:
  $A_i = E_1 \cupdot E_2$, as follows.
  For a matroid constraint,
  $E_1$ is the base selected first
  by the greedy algorithm, and $E_2$
  is the base selected second.
  For a knapsack constraint, let
  $A_i = \{a_1, \ldots, a_m \}$, where the elements of $A_i$ are ordered here by the order in which they were chosen by the greedy algorithm.
  Let $j$ be the largest index such
  that $\sum_{r = j}^m p_{a_r} \ge B$.
  If the inequality is strict,
  split item $a_j$ into two items, $\sigma_1$ and $\sigma_2$,
  such that $p_{\sigma_2} + \sum_{r = j + 1}^m p_{a_r} = B$, and
  set $E_1 = \{ a_1, \ldots, a_{j-1}, \sigma_1 \}$,
  and $E_2 = \{ \sigma_2, a_{j+1}, \ldots, a_m \}$.
  Otherwise, if no split is needed, simply let
  $E_2 = \{ a_j, a_{j+1}, \ldots, a_m \}$, and
  $E_1 = A_i \setminus E_2$.  
  Since $\sum_{a \in A_i} p_a \ge 2B$, it must hold that
  $\sum_{a \in E_1} p_a \ge B = \sum_{a \in E_2} p_a$.
\end{definition}
Our next goal is to prove that $E_1$ holds at least half of the greedy value, which is formally stated by Lemma~\ref{lem:greedy_non-monotone}. Intuitively, this is
a consequence of submodularity and the greedy selection. However, to formally prove it, we first need to define a mapping $\theta\colon E_2 \to E_1$ such that for every element $a \in E_1$ it holds that $\sum_{u \in \theta^{-1}(a)} f(u \mid A_u) \leq f(a \mid A_a)$. For matroid constraints, the mapping $\theta$ can be simply chosen as the mapping $h$ whose existence is guaranteed by Corollary~\ref{cor:perfect_matching_two_bases}. Notice that this mapping guarantees that for every $a \in E_1$, there is a sole element $\theta^{-1}(a) \in E_2$, and furthermore, $E_1 + \theta^{-1}(a) - a$ is independent in the matroid. Thus, $\theta^{-1}(a)$ was a candidate for selection by the greedy algorithm when $a$ was select. Since $a$ ended up being selected, we get
\[
	\sum_{u \in \theta^{-1}(a)} f(u \mid A_u)
	=
	f(\theta^{-1}(a) \mid A_{\theta^{-1}(a)})
	\leq
	f(\theta^{-1}(a) \mid A_a)
	\leq
	f(a \mid A_a)
	\enspace,
\]
where the first inequality follows from the submodularity of $f$ since $a \in E_1$ was selected before $\theta^{-1}(a) \in E_2$.

  For a knapsack constraint, defining $\theta$ requires us to split
  items again. We also need to explain the meaning of the notation $A_a$ in the presence of splits.
	Recall that $A_a$ contains the elements of $A_i$ that precede $a$. Originally, before the first split,
	the elements of $A_i$ are sorted by the order in which they are added to $A_i$ by the greedy algorithm.
	Whenever an element $s$ is split, the elements $\sigma_1$ and $\sigma_2$ that replace $s$ take its
	place in the order of $A_i$, and $\sigma_1$ precedes $\sigma_2$.
  The following definition now describes the process used to construct $\theta$ for a knapsack constraint. 
  \begin{definition} \label{def:theta_construction}
    We define $\theta$ through the
    following process. Initially $\theta$ is undefined for all elements of $E_2$, and we iteratively define it, while maintaining the invariant that $\sum_{u \in \theta^{-1}(a)} p_u \leq p_a$ for every $a \in E_1$.
    In every iteration, let $u$ be the last element in $E_2$
    for which $\theta$ is undefined. If there is no such element, terminate.
    Otherwise, let $a$ be the last element
    in $E_1$ for which
    $\sum_{u \in \theta^{-1}(a)} p_u < p_a$.
    Such an element must exist by our invariant since $\sum_{u \in E_2} p_u \le \sum_{a \in E_1} p_a$
    and currently $u \not \in \theta^{-1}(a)$ for any $a \in E_1$.
		If setting $\theta(u) = a$ preserves the invariant $\sum_{u \in \theta^{-1}(a)} p_u \leq p_a$, then the definition of $\theta$ is extended in this way. Otherwise, the process first splits $a$ into $\sigma_1, \sigma_2$ using a $(1 - \frac{p_a - \sum_{u' \in\theta^{-1}(a)} p_{u'}}{p_u})$-split, and then sets $\theta(\sigma_2) = a$, while leaving $\theta(\sigma_1)$ undefined for the time being. In the last case $\sigma_1$ and $\sigma_2$ take the place of $u$ in $A_i$ and $E_2$, and one can verify that the invariant inequality $\sum_{u \in \theta^{-1}(a)} p_u \leq p_a$ becomes tight for the element $a$.
		%
%
%
  \end{definition}
	 The iterative process described by the last definition is guaranteed to terminate since $E_2$ is originally finite, and every time an element of it is split, this results in in the invariant inequality of one of the finitely many elements of $E_1$ becoming tight.
  For convenience, we continue
  to denote by $f$
  the resulting function after all of the splits; recall that, by
  Lemma~\ref{lemm:split-submod}, it remains submodular.
	
  \begin{observation} \label{obs:integrity-of-splits}
	For every element $a \in E_1$, $\sum_{u \in \theta^{-1}(a)} f(u \mid A_u) \leq f(a \mid A_a)$.
  \end{observation}
  \begin{proof}
    Consider a $\gamma$-split of some element $s$ into elements $\sigma_1$ and $\sigma_2$.
    By Observation \ref{obs:split-1},
    $\marge{\sigma_1}{A_{\sigma_1}} = \gamma \marge{\sigma_1}{A_s}$
    and $\marge{\sigma_2}{A_{\sigma_2}} = (1 - \gamma) \marge{\sigma_2}{A_s}$. Since we also have $p_{\sigma_1} = \gamma p_s$ and $p_{\sigma_2} = (1 - \gamma)p_s$ by the definition of a split, we get
		\[
			\frac{f(\sigma_1 \mid A_{\sigma_1})}{p_{\sigma_1}}
			=
			\frac{f(\sigma_2 \mid A_{\sigma_2})}{p_{\sigma_2}}
			=
			\frac{f(s \mid A_{s})}{p_{s}}
			\enspace.
		\]
		In other words, if we think about the ratio $\frac{f(s \mid A_{s})}{p_{s}}$ as the ``density'' of element $s$, then a split operation creates two new elements that inherit the density of the spitted element. By repeating this argument multiple times, we get that, if $u$ is an element of the final set $A_i$ derived from an element $s(u)$ of the original set $E_2$ through a series of zero or more splits, then $u$ and $s(u)$ have the same density. Thus,
		\begin{align*}
			\sum_{u \in \theta^{-1}(u)} f(u \mid A_u)
			={} &
			\sum_{u \in \theta^{-1}(u)} p_u \cdot \frac{f(s(u) \mid A_{s(u)})}{p_{s(u)}}\\
			\leq{} &
			\sum_{u \in \theta^{-1}(u)} p_u \cdot \frac{f(s(u) \mid A_{a})}{p_{s(u)}}
			\leq
			\sum_{u \in \theta^{-1}(u)} p_u \cdot \frac{f(a \mid A_{a})}{p_a}
			\leq
			f(a \mid A_{a})
			\enspace,
		\end{align*}
		where the first inequality follows from the submodularity of $f$, the second inequality holds since $s(u)$, as an element of the original set $E_2$ was a candidate that Algorithm~\ref{alg:greedy_density} could have selected instead of the element $a \in E_1$,\footnote{There is also a possibility that either $s(u)$, $a$ or both were derived from the single split done during the definition of $E_1$ and $E_2$. However, the inequality holds also in this case for similar reasons.} and the last inequality holds due to the invariant maintained by the construction of $\theta$ according to Definition~\ref{def:theta_construction} because the existence of the dummy elements guarantees that $f(a \mid A_{a}) \geq 0$.
  \end{proof}
	
	We are now ready to prove Lemma~\ref{lem:greedy_non-monotone}.
	
	\begin{lemma} \label{lem:greedy_non-monotone}
  For every $i \in [\ell]$, 
	$f( E_1 ) \ge f(A_i) / 2$. 
\end{lemma}
\begin{proof}
By the properties of the mapping $\theta$ (Observation~\ref{obs:integrity-of-splits} in the case of a knapsack constraint, and a similar claim proved earlier for a matroid constraint), we get
\begin{align*}
	f(E_1)
	={} &
	f(\varnothing) + \sum_{a \in E_1} \marge{ a }{ A_a }
	\geq
	f(\varnothing) + \sum_{a \in E_1} \sum_{u \in \theta^{-1}(a)} \marge{ u }{ A_u }\\
	={} &
	f(\varnothing) + \sum_{u \in E_2} \marge{ u }{ A_u }
	=
	f(\varnothing) + [f(A_i) - f(E_1)]
	\enspace,
\end{align*}
where the last equality uses the fact that splitting an element of a set $S$ does not affect the value of this set. 
The lemma now follows by rearranging this inequality, and observing that $f(\varnothing) \geq 0$ since $f$ is non-negative.
\end{proof}
The next part of the analysis of Algorithm~\ref{alg:comb-gen-knapsack} maps $\hat O$,
the optimal elements that were not selected
by any of the $\ell$ runs of the greedy algorithm,
to the elements in $E_2$.
Using such a mapping, we get the following
observation.
\begin{observation} \label{prop:ohat}
  For all $i \in [\ell]$,
  $f( \hat O \cup A_i ) - f(A_i) \le f(A_i)/2.$
\end{observation}
\begin{proof}
  We define a mapping
  $\theta\colon \hat O \to E_2$ such that $\sum_{u \in \theta^{-1}(a)} f(u \mid A_u) \leq f(a \mid A_a)$ for every $a \in E_2$.
	For knapsack constraints, this can be done using the same method
	described by Definition~\ref{def:theta_construction} since
	$\sum_{a \in E_2} p_a = B \ge \sum_{u \in \opt} p_u \geq \sum_{u \in \hat{O}} p_u$.
	For a matroid constraint, one can get the necessary mapping by first adding to $\hat{O}$ enough
  dummy elements to make it a base like $E_2$, and then invoking Corollary~\ref{cor:perfect_matching_two_bases}.
	Using the mapping $\theta$, we now get
  \begin{align*}
    f( \hat O \cup A_i ) - \ff{ A_i } &\le \sum_{u \in \hat O} \marge{u}{ A_i } \\
                                        &= \sum_{a \in E_2} \sum_{u \in \theta^{-1}(a)} \marge{u}{ A_i } \\
                                        &\le \sum_{a \in E_2} \sum_{u \in \theta^{-1}(a)} \marge{u}{ A_u } \\
                                        &\le \sum_{a \in E_2} \marge{a}{A_a} \\
                                        &= \ff{A_i } - \ff{ E_1 } \\
                                        &\le \ff{ A_{i} } - \ff{A_i} / 2 = \ff{A_i}/2,
  \end{align*}
  where the first and second inequalities are from submodularity,
  and the last inequality follows from Lemma~\ref{lem:greedy_non-monotone}.
\end{proof}
\begin{theorem} \label{thm:general_knapsack_matroid} 
  For every $\eps \in (0, 1/2)$, Algorithm~\ref{alg:comb-gen-knapsack} is a $(1/2 - \eps, O( 1 / \eps ))$-bicriteria approximation algorithm for the problems of maximizing a non-negative submodular function subject to  either a matroid or a knapsack constraint.
\end{theorem}
\begin{proof}
  Let us begin by proving the infeasibility ratio.
  For a knapsack constraint,
  each $A_i$ has a total price of at most $3B$ by Lemma~\ref{lem:greedy_monotone_infeasibility}, and
  thus, the total price of $A$ is at most $3 \ell B$. 
  Since the solution produced by Algorithm~\ref{alg:comb-gen-knapsack} is a subset of $A$, it follows that its price is
  at most $O(1 / \eps ) \cdot B$. For a matroid constraint $M$,
  each $A_i$ is the union of two disjoint bases; and hence, the
  set $A$ is the union of $2 \ell$ bases, which implies
	$\vone_A / (2\ell) \in \cP_M$. Since the output set of Algorithm~\ref{alg:comb-gen-knapsack} is a subset of $A$,
	the infeasibility ratio of Algorithm~\ref{alg:comb-gen-knapsack} for this kind of constraints
	is at most $2 \ell \le 1/\eps + 2$.

  It remains to bound the approximation ratio of Algorithm~\ref{alg:comb-gen-knapsack}. Towards this goal, note that using the fact that the $A_i$'s
  are pairwise disjoint and repeated application of submodularity, one can get
	\[
		\sum_{i \in [\ell]} \ff{ \opt \cup A_i } \ge ( \ell - 1 ) f(\opt) + f( \opt \cup A )
		\geq
		( \ell - 1 ) f(\opt)
		\enspace.
	\]
  Thus, by an averaging argument, there must exist a value $i \in \ell$ such that
  $f( \opt \cup A_i ) \ge \left( 1 - \frac{1}{\ell } \right) f(\opt)$.
  By Theorem~\ref{thm:discretedoublegreedy}, 
  \begin{equation} \label{eq:D_i_guarantee}
    \ff{ D_i \cup A_i }
		=
		g_i(D_i)
		\geq
		\frac{1}{2} g_i(\dot O) + \frac{1}{4} g_i(\varnothing)
		=
		\frac{1}{2} f( \dot O \cup A_i ) + \frac{1}{4} \ff{A_i}
		\enspace,
  \end{equation}
  and thus,
  \begin{align*}
    (1 - 2 \eps ) f(\opt) &\le \Big( 1 - \frac{1}{\ell } \Big) f(\opt) \\
    &\le f(\opt \cup A_i ) \\
                        &\le f( \hat O \cup A_i  )  - \ff{ A_i } + f( \dot O \cup A_i ) \\
                        &\le \frac{f(A_i)}{2} + \Big[2f( D_i \cup A_i ) - \frac{f(A_i)}{2}\Big] \\
                        &= 2f( D_i \cup A_i )
												\enspace,
  \end{align*}
  where the third inequality follows from submodularity,
  and the fourth inequality follows from Observation~\ref{prop:ohat}
  and Inequality~\eqref{eq:D_i_guarantee}. 
\end{proof}

\subsection{Inapproximability Results} \label{ssc:inapproximability_general}

In this section, we prove two inapproximability results about bicriteria optimization of general submodular functions. The first result is given by the next theorem, and shows that maximization of general submodular functions subject to quite simple non-down-closed constraints does not admit bicriteria approximation with constant approximation and infeasibility ratios. Previously, this result was shown by \cite{vondrak2013symmetry} for single criteria optimization, and we notice that its proof extends quite easily to bicriteria optimization.
\begin{theorem} \label{thm:hardness_non_down-closed}
Any algorithm for maximizing a non-negative submodular function over the bases of a partition matroid\footnote{A partition matroid is defined by a partition of a ground set $\cN$ into \emph{disjoint} sets $\cN_1, \cN_2, \dotsc, \cN_\ell$ and an integer value $r_i$ between $0$ and $|\cN_i|$ for every $i \in [\ell]$. A set $S$ is independent in the partition matroid if it contains at most $r_i$ elements of each set $\cN_i$ (i.e., $|S \cap \cN_i| \leq r_i$ for every $i \in [\ell]$). One can verify that a partition matroid is indeed a matroid. Furthermore, a set $S$ is a base of this matroid if $|S \cap \cN_i| = r_i$ for every $i \in [\ell]$.} that guarantees $(\alpha, \beta)$-bicriteria approximation for constants $\alpha \in (0, 1]$ and $\beta \geq 1$ must use an exponential number of value oracle queries.
\end{theorem}

Theorem~\ref{thm:hardness_non_down-closed} shows that the restriction of our algorithm from Section~\ref{ssc:bicriteria_general} to down-closed convex sets is unavoidable. For down-closed convex sets, we prove the following inapproximability result.
\begin{theorem} \label{thm:hardness_general}
Let
\[
	\nu(\beta)
	=
	\begin{cases}
		2e^{-\beta/2}(1 - e^{-\beta/2}) & \text{if $\beta \leq 2 \ln 2$} \enspace,\\
		1/2 & \text{if $\beta \geq 2 \ln 2$} \enspace.
	\end{cases}
\]
Any algorithm for maximizing a non-negative submodular function over a cardinality constraint that guarantees $(\nu(\beta) + \delta, \beta)$-bicriteria approximation for constants $\beta \geq 0$ and $\delta > 0$ must use an exponential number of value oracle queries.
\end{theorem}
Theorem~\ref{thm:hardness_general} interpolates between two known results. For single-criteria optimization ($\beta = 1$), Theorem~\ref{thm:hardness_general} recovers an inapproximability due to \cite{qi2022maximizing} (who extended a result of \cite{gharan2011submodular} that applied to partition matroids, but not to cardinality constraints). On the other exreme, for $\beta \geq 2 \ln 2$, the hardness of Theorem~\ref{thm:hardness_general} follows from the hardness proved by \cite{feige2011maximizing} for unconstrained submodular maximization (as was previously observed by \cite{crawford2023scalable}). 

Unfortunately, there is a large gap between the inapproximability proved by Theorem~\ref{thm:hardness_general} and our results. This mirrors a similar gap that exists in the literature between the state-of-the-art single-criteria approximation algorithm for maximizing a non-negative submodular function subject to either a cardinality or a matroid constraint that guarantees $0.401$-approximation \citep{buchbinder2024constrained}, and the above mentioned inapproximability results of \cite{gharan2011submodular} and \cite{qi2022maximizing}. The following proposition explains why it might be difficult to close this gap using stronger inapproximability results.
\begin{proposition} \label{prop:double_measured_continuous_greedy}
Given a non-negative submodular function $f \colon 2^\cN \to \nnR$ and a down-closed convex set $\cP$. If there exist two disjoint optimal solutions $\opt_1$ and $\opt_2$ for the problem of maximizing $f$ subject to the constraint that the characteristic vector of the solution must be in $\cP$, then for any two constants $\beta \geq 1$ and $\delta > 0$, a variant of the Measured Continuous Greedy algorithm of \cite{feldman2011unified} can be used to obtain $(\nu(\beta) - \delta, \beta)$-bicriteria integer approximation for the problem of maximizing the multilinear extension $F$ of $f$ over the convex set $\cP$.

Moreover, the above is the case even if $\opt_2$ is not an optimal solution, as long as it ``imitates'' $\opt_1$ in the sense that, for every vector $\vx \in [0, 1]^\cN$ encountered during the execution of the variant of Measured Continuous Greedy, $\inner{\nabla F(\vx)}{\characteristic_{\opt_1}} = \inner{\nabla F(\vx)}{\characteristic_{\opt_2}}$.
\end{proposition}

Given Proposition~\ref{prop:double_measured_continuous_greedy}, to improve over Theorem~\ref{thm:hardness_general}, it will be necessary to construct hard instances that do not include solutions $\opt_1$ and $\opt_2$ of the kind mentioned in Proposition~\ref{prop:double_measured_continuous_greedy}. Currently, we are not aware of any known such constructions with non-monotone objective functions.\footnote{Even if one is interested also in monotone objective functions, the only such constructions that we are aware of are the ones used to prove inapproximability results for maximizing a monotone submodular function subject to a cardinality constraint with density $c > 1/2$ (like the construction from Section~\ref{sssc:cardinality_inapproximability_large_c}). However, even these constructions are obtained by starting with a construction that has solutions $\opt_1$ and $\opt_2$ of the kind mentioned in Proposition~\ref{prop:double_measured_continuous_greedy}, and then enriching it with free elements that can always be added to any solution (and thus, belong to every optimal solution since $f$ is monotone in these constructions).}

We defer the proof of Proposition~\ref{prop:double_measured_continuous_greedy} to Section~\ref{sssc:double_measured_continuous_greedy}, and concentrate from this point on proving our inapproximability results (Theorems~\ref{thm:hardness_non_down-closed} and~\ref{thm:hardness_general}). The work horse behind the proofs of these results is Theorem~\ref{thm:symmetry_gap}, which is presented below. However, before presenting Theorem~\ref{thm:symmetry_gap} itself, we first need to state a few definitions.

\begin{definition}[Definition~1.5 of~\cite{vondrak2013symmetry}]
Given a set function $f\colon 2^\cN \to \bR$ and a constraint set $\cF \subseteq 2^\cN$, we call the optimization problem instance $\max\{f(S) \mid S \in \cF\}$ \emph{strongly symmetric} with respect to a group of permutations $\cG$ on $\cN$ if $f(S) = f(\sigma(S))$ for all $S \subseteq \cN$ and $\sigma \in \cG$, and $S \in \cF \Leftrightarrow S' \in \cF$ whenever $\bE_{\sigma \in \cG} [\characteristic_{\sigma(S)}] = \bE_{\sigma \in \cG}[\characteristic_{\sigma(S')}]$.
\end{definition}

\begin{definition}[Generalization of Definition~1.6 of~\cite{vondrak2013symmetry}]
Let $f\colon 2^\cN \to \bR$ be a set function and $\cF \subseteq 2^\cN$ a constraint set such that the instance $\max\{f(S) \mid S \in \cF\}$ is strongly symmetric with respect to a group $\cG$ of permutations over $\cN$. Let $F(x)$ be the multilinear extension of $f$, and $\cP(\cF)$ be the convex closure of the characteristic vectors of the sets in $\cF$. Let $\bar{\vx} = \bE_{\sigma \in \cG}[\sigma(\vx)]$. The symmetry gap of $\max\{f(S) \mid S \in \cF\}$ for infeasibility ratio $\beta \geq 0$ is defined as $\gamma_\beta = \overline{\optvalue}_\beta / \optvalue$, where
\[
	\optvalue = \max\{F(\vx) \mid \vx \in \cP(\cF)\}
	\enspace,
\]
and
\[
	\overline{\optvalue}_r = \max\{F(\bar{\vx}) \mid \exists_{\vy^{(1)} \in \cP(\cF), \vy^{(2)}/\beta \in \cP(\cF)}\; \vy^{(1)} \leq \vx \leq \vy^{(2)}\}
	\enspace.
\]
\end{definition}

\begin{definition}[Definition~1.7 of~\cite{vondrak2013symmetry}]
Given a ground set $\cN$, a constraint set $\cF \subseteq 2^\cN$ and a positive integer $k$, we say that $\tilde{\cF} \subseteq 2^{\cN \times [k]}$ is a refinement of $\cF$ if
\[
	\tilde{\cF}
	=
	\{\tilde{S} \subseteq \cN \times[k] \mid (x_1, x_2, \dotsc, x_{|\cN|}) \in \cP(\cF)\text{, where }x_j = \tfrac{1}{k}|\tilde{S} \cap (\{j\} \times [k])|\}
	\enspace.
\]
\end{definition}

\begin{theorem}[Generalization of Theorem~1.8 of~\cite{vondrak2013symmetry}] \label{thm:symmetry_gap}
Let $\max\{f(S) \mid S \in \cF\}$ be an instance of an optimization problem that has a non-negative submodular objective function $f$, and furthermore, is strongly symmetric and has a symmetry gap $\gamma_\beta$ with respect to a group $\cG$ of permutations and a constant infeasibility ratio $\beta$. Let $\cC$ be the class of instances $\max\{\tilde{f}(S) \mid S \in \tilde{F}\}$, where $\tilde{F}$ is a refinement of $F$ and $\tilde{f}$ is non-negative, monotone if $f$ is monotone, and submodular. Then, for every $\delta > 0$, any (even randomized) $((1 + \delta)\gamma_\beta, \beta)$-bicriteria approximation algorithm for the class $\cC$ would require exponentially many value queries to $\tilde{f}(S)$.
\end{theorem}

We do not prove Theorem~\ref{thm:symmetry_gap} here because the proof of Theorem~1.8 of \cite{vondrak2013symmetry} immediately extends to proving our more general Theorem~\ref{thm:symmetry_gap}. To use Theorem~\ref{thm:symmetry_gap}, we need to present particular families of instances. The first such family was originally suggested by \cite{vondrak2013symmetry}, and is parameterized by an integer parameter $n \geq 1$. An instance of this family consists of a ground set $\cN \triangleq A \cup B$, where $A \triangleq \{a_i \mid i \in [n]\}$ and $B \triangleq \{b_i \mid i \in [n]\}$, a constraint set $\cF \triangleq \{S \subseteq \cN \mid |S \cap A| = 1, |S \cap B| = n - 1\}$, and an objective function $f \colon 2^\cN \to \nnR$ given by
\[
	f(S)
	\triangleq
	\sum_{i = 1}^n \characteristic[a_i \in S \wedge b_i \not \in S]
	\qquad
	\forall\;S \subseteq \cN
	\enspace.
\]
Observe that $f$ is submodular since it is the cut function of a directed graph consisting of $n$ disjoint arcs (one arc from $a_i$ to $b_i$ for every $i \in [n]$). Furthermore, the above instance is invariant to permutations that permute $A$ and $B$ in the same way (i.e., they map $a_i$ to $a_{\sigma(i)}$ and $b_i$ to $b_{\sigma(i)}$ for some permutation $\sigma$ of $[n]$). Let $\cG$ be the group of such permutations. One can verify that the feasibility (i.e., membership in $\cF$) of a set $S \subseteq \cN$ can be determined based only on $\bE_{\sigma \in \cG}[\characteristic_{\sigma(S)}]$, and therefore, the instance $\max\{f(S) \mid S \in \cF\}$ is strongly symmetric with respect to $\cG$.

\begin{lemma}
The symmetry gap of the above instance $\max\{f(S) \mid S \in \cF\}$ with respect to $\cG$ and an infeasibility ratio $\beta \geq 1$ is at least $\beta/n$.
\end{lemma}
\begin{proof}
Let $S = \{a_1\} \cup (B - \{b_1\})$. Clearly, $S$ is feasible in the above instance, and $f(S) = 1$. Thus, $\optvalue \geq F(\characteristic_S) = f(S) = 1$. Consider now arbitrary vectors $\vy^{(1)} \in \cP(\cF)$, $\vy^{(2)}/\beta \in \cP(\cF)$ and $\vy^{(1)} \leq \vx \leq \vy^{(2)}$. By the definition of $\cG$, there exists values $a, b \in [0, 1]$ such that the symmetrized version $\bar{\vx}$ of $\vx$ obeys
\[
	\bar{x}_{a_i} = a
	\quad\text{and}\quad
	\bar{x}_{b_i} = b
	\qquad
	\forall\; i \in [n]
	\enspace.
\]
Furthermore, since $\vx \geq \vy^{(1)} \in \cP(\cF)$, it must hold that $b \geq \frac{n - 1}{n}$. Similarly, since $\vx/\beta \leq \vy^{(2)}/\beta \in \cP(\cF)$, it must hold that $a \leq \frac{\beta}{n}$. Thus,
\[
	F(\vx)
	=
	na(1 - b)
	\geq
	n \cdot \frac{\beta}{n} \cdot \Big(1 - \frac{n - 1}{n}\Big)
	=
	\frac{\beta}{n}
	\enspace.
\]
This implies that $\overline{\optvalue}_\beta \leq \beta/n$, which completes the proof of the lemma since the symmetry gap is $\overline{\optvalue}_\beta / \optvalue$.
\end{proof}

By combining the last lemma with Theorem~\ref{thm:symmetry_gap}, we get the following corollary.
\begin{corollary} \label{cor:non-down-closed}
For every constant integer $n \geq 1$ and constants $\delta > 0, \beta \geq 1$, every algorithm that guarantees $((1 + \delta)\beta/n, \beta)$-bicriteria approximation for the problem of maximizing non-negative submodular functions over refinements of $\cF$ must use an exponential number of value oracle queries.
\end{corollary}

Notice now that $\cF$ represents a partition matroid base constraint, and so are its refinements. Thus, Theorem~\ref{thm:hardness_non_down-closed} follows from Corollary~\ref{cor:non-down-closed} by simply choosing a large enough $n$ so that $\alpha \leq (1 + \delta)\beta/n$.

Next, we present the family of instances that we use to prove Theorem~\ref{thm:hardness_general}. This family is due to \cite{qi2022maximizing} (up to minor changes), and is parametrized by two values: an integer $n \geq 1$ and a real value $\kappa \in [0, 1]$. An instance of this family consists of a ground set $\cN \triangleq \{a, b\} \cup A \cup B$, where $A \triangleq \{a_{i, j} \mid i, j \in [n]\}$ and $B \triangleq \{b_{i, j} \mid i, j \in [n]\}$, a constraint set $\cF \triangleq \{S \subseteq \cN \mid |S| \leq n + 1\}$, and an objective function $f \colon 2^\cN \to \nnR$ given by $f(S) \triangleq \kappa g(S) + (1 - \kappa) \cdot h(S)$, where
\[
	g(S)
	\triangleq
	\begin{cases}
		1 & \text{if $|S \cap \{a, b\}| = 1$} \enspace, \\
		0 & \text{otherwise} \enspace,
	\end{cases}
	\qquad
	\forall\;S \subseteq \cN
	\enspace.
\]
and
\begin{multline*}
	h(S)
	\triangleq
	\characteristic[a \not \in S] \cdot \bigg[1 - \prod_{i = 1}^n \bigg(1 - \frac{|S \cap \{a_{i, j} \mid j \in [n]\}|}{n}\bigg)\bigg]\\
	+
	\characteristic[b \not \in S] \cdot \bigg[1 - \prod_{i = 1}^n \bigg(1 - \frac{|S \cap \{b_{i, j} \mid j \in [n]\}|}{n}\bigg)\bigg]
	\qquad
	\forall\;S \subseteq \cN
	\enspace.
\end{multline*}

\begin{observation}
The function $f$ is non-negative and submodular.
\end{observation}
\begin{proof}
One can verify that $g$ is non-negative and submodular, and $h$ is non-negative. Below we explain why $h$ is also submodular, which implies the observation since $f$ is a convex combination of $g$ and $h$.

To show that $h$ is submodular, we need to show that for every element $u \in \cN$, the marginal contribution $h(u \mid S - u)$ can only decrease when elements are added to $S$. For $u = a$ this is true since
\[
	h(a \mid S - a)
	=
	\prod_{i = 1}^n \bigg(1 - \frac{|S \cap \{a_{i, j} \mid j \in [n]\}|}{n}\bigg) - 1
	\enspace,
\]
and for $u = a_{i', j'}$ (for any $i', j' \in [n]$) this is true since
\[
	h(a_{i', j'} \mid S - a_{i', j'})
	=
	\begin{cases}
		\frac{1}{n}\prod_{i \in [n] \setminus \{i'\}} \Big(1 - \frac{|S \cap \{a_{i, j} \mid j \in [n]\}|}{n}\Big) & \text{if $a \not \in S$} \enspace,\\
		 0 & \text{if $a \in S$} \enspace.
	\end{cases}
\]
For $u = b$ and $u = b_{i', j'}$ one can get analogous expressions that show that, in these cases as well, $h(u \mid S - u)$ can only decrease when elements are added to $S$.
\end{proof}

The optimal value for the above instance is $\optvalue \geq 1$ because $\{a\} \cup \{b_{1, j} \mid j \in [n]\}$ is a feasible solution for it of value $1$. Additionally, we observe that the above instance is symmetric under permutations of the following three kinds.
\begin{itemize}
	\item The permutation that switches between $a$ and $b$ and between $a_{i, j}$ with $b_{i, j}$ (for every $i, j \in [n]$).
	\item Every permutation that maps $a_{i, j}$ and $b_{i, j}$ to $a_{\sigma(i), j}$ and $b_{\sigma(i), j}$, respectively, for some permutation $\sigma$ of $[n]$ (and maps $a$ and $b$ to themselves).
	\item Every permutation that maps $a_{i, j}$ and $b_{i, j}$ to $a_{i, \sigma(j)}$ and $b_{i, \sigma(j)}$, respectively, for some permutation $\sigma$ of $[n]$ (and maps $a$ and $b$ to themselves).
\end{itemize}
Let $\cG$ be the group of permutations that can be obtained by concatenating permutations of the above three kinds. Notice that for every set $S \subseteq 2^\cN$, it is possible to determine whether $S \in \cF$ by simply checking whether $|\bE_{\sigma \sim \cG}[\characteristic_{\sigma(S)}]| \leq n + 1$ because $|\sigma(S)| = S$ for every permutation $\sigma$. Thus, the above instance is strongly symmetric with respect to $\cG$.

The last property of the above instance that we need in order to invoke Theorem~\ref{thm:symmetry_gap} is a bound on its symmetry gap. To get this bound, we need to choose a value for $\kappa$. Specifically, we choose $\kappa = \min\{1, e^{\beta/2} - 1\}$. The following lemma then gives the bound.
\begin{lemma}
For every vector $\vx/\beta \in \cP(\cF)$, $F(\bar{\vx}) \leq \nu(\beta) + \frac{3\beta\max\{1, \beta\}}{n}$, where $\bar{\vx} = \bE_{\sigma \in \cG}[\sigma(\vx)]$ and $f$ is the multilinear extension of $F$. Since $\optvalue \geq 1$, as was observed above, the symmetric gap $\gamma_\beta$ of $\max\{f(S) \mid S \in \cF\}$ with respect to $\cG$ and the infeasibility ratio $\beta$ is at least $\nu(\beta) + \frac{3\beta\max\{1, \beta\}}{n}$.
\end{lemma}
\begin{proof}
By the definition of $\cG$, the coordinates of the vector $\bar{\vx}$ can take only two different values. Specifically, the coordinates corresponding to the elements $a$ and $b$ must both take the same value $p \in [0, 1]$, and the other coordinates (corresponding to the elements of the forms $a_{i, j}$ and $b_{i, j}$) must all take the same value $q \in [0, 1]$. Furthermore, since $\vx/\beta \in \cP(\cF)$, it must also hold that $2p + 2n^2q \leq \beta(n + 1)$.

Notice that the multilinear extension $H$ of $h$ is given by
\begin{align*}
	H(\vx)
	={} &
	(1 - x_a) \cdot \bE_{S \sim \RSet(\vx)} \bigg[1 - \prod_{i = 1}^n \bigg(1 - \frac{|S \cap \{a_{i, j} \mid j \in [n]\}|}{n}\bigg)\bigg]\\
	&\mspace{180mu}+
	(1 - x_b) \cdot \bE_{S \sim \RSet(\vx)} \bigg[1 - \prod_{i = 1}^n \bigg(1 - \frac{|S \cap \{b_{i, j} \mid j \in [n]\}|}{n}\bigg)\bigg]\\
	={} &
	(1 - x_a) \cdot \bigg[1 - \prod_{i = 1}^n \bigg(1 - \frac{\bE_{S \sim \RSet(\vx)} [|S \cap \{a_{i, j} \mid j \in [n]\}|]}{n}\bigg)\bigg]\\
	&\mspace{180mu}+
	(1 - x_b) \cdot \bE_{S \sim \RSet(\vx)} \bigg[1 - \prod_{i = 1}^n \bigg(1 - \frac{\bE_{S \sim \RSet(\vx)} [|S \cap \{b_{i, j} \mid j \in [n]\}|]}{n}\bigg)\bigg]\\
	={} &
	(1 - x_a) \cdot \bigg[1 - \prod_{i = 1}^n \bigg(1 - \frac{\sum_{j \in [n]} x_{a_{i, j}}}{n}\bigg)\bigg]
	+
	(1 - x_b) \cdot \bE_{S \sim \RSet(\vx)} \bigg[1 - \prod_{i = 1}^n \bigg(1 - \frac{\sum_{j \in [n]} x_{b_{i, j}}}{n}\bigg)\bigg]
	\enspace,
\end{align*}
where the second equality holds because if we define $X_i \triangleq |S \cap \{a_{i, j} \mid j \in [n]\}|$, then the random variables $X_1, X_2, \dotsc, X_n$  obtained are independent of each other, and a similar claim holds for random variables defined as $Y_i \triangleq |S \cap \{b_{i, j} \mid j \in [n]\}|$. Thus,
\begin{align*}
	F(\bar{\vx})
	={} &
	2\kappa p(1 - p) + 2(1 - \kappa)(1 - p)[1 - (1 - q)^n]\\
	\leq{} &
	2\kappa p(1 - p) + 2(1 - \kappa)(1 - p)\bigg[1 - \bigg(1 - \frac{\beta(n + 1)}{2n^2}\bigg)^n\bigg]\\
	\leq{}&
	2\kappa p(1 - p) + 2(1 - \kappa)(1 - p)\bigg[1 - e^{-\frac{\beta(n + 1)}{2n}}\bigg(1 - \frac{\beta^2(n + 1)^2}{4n^3}\bigg)\bigg]\\
	\leq{} &
	2\kappa p(1 - p) + 2(1 - \kappa)(1 - p)(1 - e^{-\beta/2}) + \frac{3\beta\max\{1, \beta\}}{n}
\end{align*}
where the first inequality holds since $q \leq \frac{\beta(n + 1) - 2p}{2n^2} \leq \frac{\beta(n + 1)}{2n^2}$, and the last inequality holds since
\[
	e^{-\frac{\beta(n + 1)}{2n}}
	=
	e^{-\beta/2} - \int_{\beta/2}^{\beta/2 + \beta/(2n)} e^{-x} \text{d}x
	\geq
	e^{-\beta/2} - \frac{\beta e^{-\beta/2}}{2n}
	\geq
	e^{-\beta/2} - \frac{\beta}{2n}
	\enspace,
\]
and
\[
	\frac{\beta^2(n + 1)^2}{4n^3}
	\leq
	\frac{4\beta^2 n^2}{4n^3}
	=
	\frac{\beta^2}{n}
	\enspace.
\]

We now need to distinguish between two regimes. If $\beta \geq \frac{1}{2}\ln 2$, then $\kappa = 1$, and therefore,
\[
	F(\bar{\vx})
	\leq
	2p(1 - p) + \frac{3\beta\max\{1, \beta\}}{n}
	\leq
	\frac{1}{2} + \frac{3\beta\max\{1, \beta\}}{n}
	=
	\nu(\beta) + \frac{3\beta\max\{1, \beta\}}{n}
	\enspace.
\]
The other case we need to consider is the case of $\beta \leq \frac{1}{2}\ln 2$. In this case $\kappa = e^{\beta/2} - 1$, and thus,
\begin{align} \label{eq:average_bound}
	F(\bar{\vx})
	\leq{} &
	2p(e^{\beta/2} - 1)(1 - p) + 2(2 - e^{\beta/2})(1 - p)(1 - e^{-\beta/2}) - \frac{3\beta\max\{1, \beta\}}{n}\\\nonumber
	={} &
	2(e^{\beta/2} - 1)(1 - p)(p + 2e^{-\beta / 2} - 1) - \frac{3\beta\max\{1, \beta\}}{n}
	\enspace.
\end{align}
The derivative of the rightmost side of this inequality with respect to $p$ is
$
	2(e^{\beta/2} - 1)(2 - 2e^{-\beta/2} - 2p)
$.
Notice that this derivative is a decreasing function of $p$, which takes the value $0$ when $p = 1 - e^{-\beta/2}$.
Thus, the right hand side of Inequality~\eqref{eq:average_bound} is maximized for this value of $p$, which implies
\begin{align*}
	F(\bar{\vx})
	\leq{} &
	2(e^{\beta/2} - 1)[1 - (1 - e^{-\beta/2})][(1 - e^{-\beta/2}) + 2e^{-\beta / 2} - 1] - \frac{3\beta\max\{1, \beta\}}{n}\\
	={} &
	2(e^{\beta/2} - 1) \cdot e^{-\beta/2} \cdot e^{-\beta / 2} - \frac{3\beta\max\{1, \beta\}}{n}\\
	={} &
	2e^{-\beta / 2}(1 - e^{-\beta/2}) - \frac{3\beta\max\{1, \beta\}}{n}
	=
	\nu(\beta) - \frac{3\beta\max\{1, \beta\}}{n}
	\enspace.
	\qedhere
\end{align*}
\end{proof}

Plugging the above results into Theorem~\ref{thm:symmetry_gap}, we get the following corollary.
\begin{corollary}
For every constant integer $n \geq 1$ and constants $\delta > 0, \beta \geq 0$, every algorithm that guarantees $((1 + \delta')(\nu(\beta) + \frac{3\beta\max\{1, \beta\}}{n}), \beta)$-bicriteria approximation for the problem of maximizing non-negative submodular functions over refinements of $\cF$ must use an exponential number of value oracle queries.
\end{corollary}

Theorem~\ref{thm:hardness_general} follows from the last corollary because the refinements of $\cF$ are cardinality constraints, and because of the next observation.
\begin{observation}
For any fixed choice of values for $\beta \geq 0$ and $\delta > 0$, there exists a large enough $n$ for which it holds that $(1 + \delta)(\nu(\beta) + \frac{3\beta\max\{1, \beta\}}{n}) \leq \nu(\beta) + \delta$.
\end{observation}
\begin{proof}
Notice that $\nu(\beta) \leq 1/2$ for every $\beta$. Therefore,
\begin{align*}
	(1 + \delta)\bigg(\nu(\beta) + \frac{3\beta\max\{1, \beta\}}{n}\bigg)
	\leq{} &
	\nu(\beta) + \frac{3\beta\max\{1, \beta\}}{n} + \delta \cdot \bigg(\frac{1}{2} + \frac{3\beta\max\{1, \beta\}}{n}\bigg)\\
	\leq{} &
	\nu(\beta) + \frac{\delta}{4} + \delta \cdot \bigg(\frac{1}{2} + \frac{1}{4}\bigg)
	=
	\nu(\beta) + \delta
	\enspace,
\end{align*}
where the second inequality holds for $n \geq 12\beta\max\{1, \beta\} \cdot \max\{1, \delta^{-1}\}$.
\end{proof}

\subsubsection{Proof of Proposition~\texorpdfstring{\ref{prop:double_measured_continuous_greedy}}{\ref*{prop:double_measured_continuous_greedy}}} \label{sssc:double_measured_continuous_greedy}
In this section, we prove Proposition~\ref{prop:double_measured_continuous_greedy} by
analyzing a variant of Measured Continuous Greedy (Algorithm \ref{alg:mcg-disjoint})
designed for a scenario in which multiple disjoint optimal solutions exist.
Let $f \colon 2^\cN \to \nnR$ be a
non-negative submodular function,
 $\cP \subseteq [0, 1]^\cN$ be a down-closed convex set,
 and $\opt$ be a set
 whose characteristic vector $\characteristic_{\opt}$
 is in $\cP$.
In the rest of this section, we make the following assumption, which informally
states that there are $\ell$ optimal solutions.
\begin{assumption} \label{ass:disjoint-opts} 
There exist a positive integer $\ell \in \bN$ (known to the algorithm) and $\ell$ vectors
$\vz^{(1)}$,$\vz^{(2)}$,\ldots, $\vz^{(\ell)} \in \cP$
such that $F( \vz^{(i)} ) \ge f(\opt )$ for each $i \in [\ell]$,
and $\sum_{i \in [\ell ]} \vz^{(i)} \le \vone_\cN$.
\end{assumption}
Proposition~\ref{prop:double_measured_continuous_greedy} follows from the analysis in
this section for the special case of $\ell = 2$ and vectors $\vz^{(1)}$ and $\vz^{(2)}$ that are integral. An important weakened version of Assumption~\ref{ass:disjoint-opts} merely promises that the disjoint solutions imitate an optimal solution; that is,
for each $\vz_i$, $\inner{\nabla F(\vx)}{ \vz_i } \ge \inner{\nabla F(\vx)}{\characteristic_{\opt}}$ for every vector $\vx$ encountered
during the execution of Algorithm~\ref{alg:mcg-disjoint}. The analysis we present in this section can be adapted to work also with this weakened version of Assumption~\ref{ass:disjoint-opts}, which is necessary for proving the exact statement of Proposition~\ref{prop:double_measured_continuous_greedy}. However, for simplicity, we explicitly consider only the stated version of Assumption~\ref{ass:disjoint-opts} in the rest of the section.
 
The input for Algorithm~\ref{alg:mcg-disjoint} consists of the multi-linear extension $F$ of $f$, a constant value $T \geq 0$, the constraint convex set $\cP$ and the value $\ell$ from Assumption~\ref{ass:disjoint-opts}.

\begin{algorithm}[ht] 
\caption{\textsc{Measured Continuous Greedy for multiple \opt s}$(F, T, \cP, \ell)$}\label{alg:mcg-disjoint}
\DontPrintSemicolon
Let $\vy{(0)} \gets \vzero$. \\
\For{every time $t$ between $0$ and $T$}
{
  Let $\vx{(t)} \gets \arg \max_{\vx \in [0,1]^{\cN} \mid \vx / \ell \in \cP} \inner{\vx \hprod (\vone_\cN - \vy{(t)}) }{ \nabla F(\vy{(t)})}$.\\
  Increase $\vy{(t)}$ at a rate $\frac{d\vy{(t)}}{dt} = \frac{1}{\ell} \vx{(t)} \hprod (\vone_\cN - \vy{(t)})$.\\
}
\Return ${\vy}{(T)}$.
\end{algorithm}

We begin the analysis of Algorithm~\ref{alg:mcg-disjoint} with the following observation.
\begin{observation} \label{obs:violation_multiple}
For every $t \in (0, T]$, $\vy(t) / t \in \cP$.
\end{observation}
\begin{proof}
Notice that $\vy(\tau) \leq \vone_\cN$ for every $\tau \in [0, t]$. Thus,
  \[ \frac{\vy{(t)}}{t} = \frac{1}{t} \int_0^t \frac{1}{\ell } \vx{( \tau )} \hprod (\vone_\cN - \vy{(t)}) \text{d} \tau \leq \frac{1}{t} \int_0^t \frac{\vx{( \tau )}}{\ell} \text{d}\tau \enspace. \]
  Since $\cP$ is convex and the vectors
  $\vx{( \tau )}/\ell$ belong to $\cP$ for every $\tau \in [0, t]$, the expression $\frac{1}{t} \int_0^t \frac{\vx{( \tau )}}{\ell} \text{d}\tau$ on the rightmost hand
  side of the above inequality belongs to $\cP$.
  Thus, by the down-closedness of $\cP$, $\vy{(t)}/t$ also belongs to $\cP$.
\end{proof}

Let $\opt(t) \triangleq \min_{i \in [\ell ]} \FF{ \vz^{(i)} \psum \vy(t)}$. Intuitively, $\opt(t)$ is the minimum value that can be obtained by adding one of the vectors $\vz^{(1)}, \vz^{(2)}, \dotsc, \vz^{(\ell)}$ ``on top'' of the solution $\vy(t)$ maintained by Algorithm~\ref{alg:mcg-disjoint} at time $t$.

\begin{observation} \label{obs:opt-t-disjoint}
  For every $0 \leq t \leq T$, $\opt(t) \ge \opt \cdot e^{-t/\ell}$.
\end{observation}
\begin{proof}
  Observe that
  $ \frac{d \vy(t) }{dt} = \frac{1}{\ell} \vx(t) \hprod ( \vone_\cN - \vy(t) ) \le \frac{1}{\ell}(\vone_\cN - \vy(t))$
  since $\vx(t) \in [0,1]^\cN$, and $\vy(t) \leq \vone_\cN$.
  Combining this differential inequality with the boundary condition $\vy(0) = \vzero$ yields
  $\| \vy(t) \|_{\infty} \le 1 - e^{-t/\ell}.$ 
  Therefore, for every $i \in [\ell]$,
	\begin{multline*}
		F( \vz^{(i)} \psum \vy(t) )
		=
		\bE[f(\RSet(\vz^{(i)} \psum \vy(t)))]
		=
		\bE_{\RSet(\vz^{(i)})}[\bE_{\RSet(\vy(t))}[f(\RSet(\vz^{(i)}) \cup \RSet(\vy(t))]]\\
		\geq
		\bE_{\RSet(\vz^{(i)})}[(1 - \|\vy(t)\|_\infty) \cdot f(\RSet(\vz^{(i)})]
		=
		(1 - \|\vy(t)\|_\infty) \cdot F(\vz^{(i)})
		\geq
		e^{-t/\ell} \cdot \opt
		\enspace,
	\end{multline*}
	where the first inequality follows from Lemma~\ref{lemma:prob-lemma} since $f(S \cup \cdot)$ is a non-negative submodular function for every set $S \subseteq \cN$.
\end{proof}

Using the notation $\opt(t)$, we can lower bound the rate in which the value of the solution $\vy(t)$ of Algorithm~\ref{alg:mcg-disjoint} increases as a function of $t$.

\begin{observation} \label{obs:disjoint-deriv}
	For every $0 \leq t \leq T$,
  \[ \frac{d \FF{ \vy{(t)} }}{dt} \ge \opt(t) - \FF{ \vy(t) } \enspace. \]
\end{observation}
\begin{proof}
By the chain rule,
      \begin{align*}
        \frac{d \FF{ \vy{(t)} }}{dt} &= \Big\langle \nabla F( \vy{(t)} ) , \frac{1}{\ell}\vx(t) \hprod (\vone - \vy{(t)}) \Big\rangle \\
                                     &= \frac{1}{\ell} \inner{ \vx(t) }{ \nabla F( \vy{(t)}) \hprod (\vone - \vy{(t)}) } \\
                                     &\ge \frac{1}{\ell} \bigg \langle \sum_{i=1}^\ell \vz^{(i)}, \nabla F( \vy{(t)}) \hprod (\vone - \vy{(t)}) \bigg\rangle \\
                                     &= \frac{1}{\ell} \sum_{i=1}^\ell \inner{\vz^{(i)} }{ \nabla F( \vy{(t)}) \hprod (\vone - \vy{(t)}) } \\
                                     &\ge \frac{1}{\ell} \bigg( \sum_{i=1}^\ell F( \vy(t) \psum \vz^{(i)} ) - \FF{ \vy(t) } \bigg) \\
                                     &\ge \opt(t) - F( \vy(t) ),
      \end{align*}
      where the first inequality follows from the way $\vx(t)$ is chosen by Algorithm~\ref{alg:mcg-disjoint}
      since $\sum_{i=1}^\ell \vz_i \in [0, 1]^\cN$ and $\frac{1}\ell \sum_{i=1}^\ell \vz_i \in \cP$, and the second inequality follows from Lemma~\ref{lem:positive_direction}.
  \end{proof}

From Observations~\ref{obs:opt-t-disjoint} and~\ref{obs:disjoint-deriv}, we obtain the following
differential inequality.
\[ \frac{dF(\vy{(t)})}{dt} \ge f(\opt) \cdot e^{-t/\ell} - \FF{ \vy{(t)} }; \FF{\vy{(0)}} \ge 0 \enspace, \]
which yields
\[ \FF{ \vy{(t)} } \ge \frac{\ell  ( e^{-t/\ell} - e^{-t}) }{ \ell - 1 } \cdot f( \opt ) \enspace. \]
Together with Observation~\ref{obs:violation_multiple}, this inequality implies that plugging $T = \beta$ into Algorithm~\ref{alg:mcg-disjoint} yields $(\frac{\ell  ( e^{-\beta/\ell} - e^{-\beta}) }{ \ell - 1 }, \beta)$-biciteria integer-approximation algorithm (under Assumption~\ref{ass:disjoint-opts}). See Figure~\ref{fig:performance_for_ell} for some intuition about this guarantee. Furthermore, one can verify that for $\ell = 2$, the approximation ratio part of the guarantee is exactly equal to $\nu(\beta)$, which completes the proof of Proposition~\ref{prop:double_measured_continuous_greedy}.

\begin{figure}[t]
  \centering
  \includegraphics[width=0.4\textwidth]{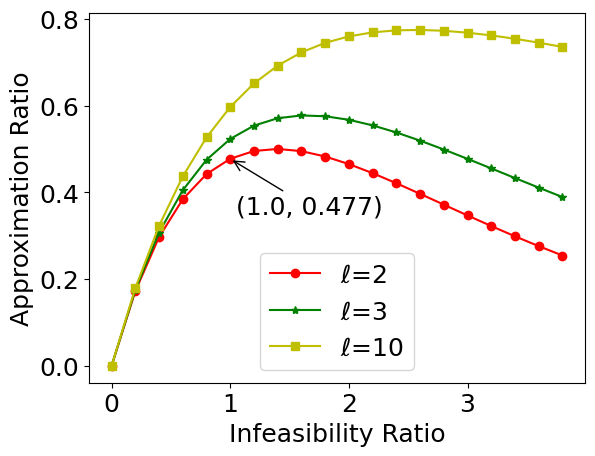}
  \caption{A plot of the approximation ratio achieved by Algorithm~\ref{alg:mcg-disjoint} vs.\ the infeasibility ratio $\beta$ for a few values of $\ell$, the number of disjoint optimal solutions. We remark that with $\ell > 2$, approximation ratios better than $1/2$ are achievable (even with an infeasibility ratio of less than $1$, i.e., feasible solutions). } \label{fig:performance_for_ell}
\end{figure}



%% file: SymmetricFunctions.tex
\section{Symmetric Submodular Functions} \label{sec:symmetric}

This section includes our results for bicriteria maximization of symmetric submodular functions. In Section~\ref{ssc:symmetric_down_closed}, we present results for general down-closed convex set constraints. For cardinality and knapsack constraints, we have improved results depending on the density of these constraints, which are presented in Section~\ref{ssc:improved_results_symmetric}. Finally, in Section~\ref{ssc:simple_greedy_symmetric}, we present a simple greedy algorithm for bicriteria maximization of symmetric submodular functions subject to cardinality and knapsack constraints.

\subsection{Algorithms for Down-Closed Convex Sets} \label{ssc:symmetric_down_closed}

In this section, we prove the following theorem.
\begin{theorem} \label{thm:symmetric_polytope}
For every constant $\eps \in (0, 1/2]$, there exists a polynomial time $(1/2 - \eps - o(1),\allowbreak \frac{1}{2}\ln(\eps^{-1}/2))$-bicriteria integer-approximation algorithm for the problem of maximizing the multilinear extension $F$ of a non-negative symmetric submodular function $f\colon 2^\cN \to \nnR$ subject to a down-closed solvable convex set $\cP \subseteq [0, 1]^\cN$.
\end{theorem}

One can observe that Theorem~\ref{thm:symmetric_polytope} follows from the following known result by choosing $T = \frac{1}{2}\ln(\eps^{-1}/2)$.

\begin{theorem}[Part of Theorem~1.1 of \cite{feldman2017maximizing}] \label{thm:measured_continuous_greedy_symmetric}
Given a non-negative symmetric submodular function $f\colon 2^\cN \to \nnR$, a down-closed solvable convex set $\cP \subseteq [0, 1]^N$ and a constant stopping time $T \geq 0$, there exists a polynomial time algorithm that finds a point $\vx \in [0, 1]^\cN$ such that $\vx / T \in \cP$ and $\bE[F(\vx)] \geq \frac{1}{2}(1 - e^{-2T} - o(1)) \cdot f(\opt)$, where $\opt$ is a set maximizing $f(\opt)$ among all the sets $\{S \subseteq \cN \mid \characteristic_S \in P\}$.\footnote{Technically, as stated in~\citep{feldman2017maximizing}, the result is restricted to down-closed solvable polytopes rather than general down-closed solvable convex sets. However, the proof of~\cite{feldman2017maximizing} does not use this restrictions in any way.}
\end{theorem}

Despite Theorem~\ref{thm:measured_continuous_greedy_symmetric} being known, in the rest of this section, we describe and analyze an algorithm (Algorithm~\ref{alg:super_measured_continuous_greedy}) that has the properties stated in Theorem~\ref{thm:measured_continuous_greedy_symmetric}. This algorithm is based on a somewhat different technique compared to the algorithm used by \cite{feldman2017maximizing} to prove Theorem~\ref{thm:measured_continuous_greedy_symmetric}, and it will play a central role in obtaining the results of Section~\ref{ssc:improved_results_symmetric}. As we often do in this paper, we present a version of Algorithm~\ref{alg:super_measured_continuous_greedy} that is not implementable because it is continuous and assumes direct access to the multilinear extension $F$ of $f$. Using the (by now standard) techniques of~\cite{calinescu2011maximizing}, it is possible to get an efficiently implementable version of \cite{calinescu2011maximizing} which with high probability obtains an approximation ratio that is worse only by $o(1)$ compared to the approximation ratio we prove for the unimplementable version of Algorithm~\ref{alg:super_measured_continuous_greedy}. One can observe that Algorithm~\ref{alg:super_measured_continuous_greedy} is very similar to the Measured Continuous Greedy algorithm of \cite{feldman2011unified}. The only difference between them is the factor of $2$ multiplying $\vy(t)$ on Lines~\ref{line:increase} and~\ref{line:direction}.

\begin{algorithm}
\DontPrintSemicolon
\caption{\textsc{More Measured Continuous Greedy}$(f \colon 2^\cN \to \nnR, \cP, T)$} \label{alg:super_measured_continuous_greedy}
Let $\vy(0) \gets \vzero$.\\
\For{each time $t \in [0, T]$}
{
	Let $\vx(t)$ be a vector in $\cP$ maximizing $\inner{\vx(t)}{(\characteristic_\cN - 2\cdot \vy(t)) \hprod \nabla F(\vy(t))}$.\label{line:direction}\\
	Increase $\vy(t)$ at a rate $\frac{d\vy(t)}{dt} = \vx(t) \hprod (\characteristic_\cN - 2 \cdot \vy(t))$.\label{line:increase}
}
\Return $y(T)$.
\end{algorithm}

We begin the analysis of Algorithm~\ref{alg:super_measured_continuous_greedy} with the following observation.
\begin{observation} \label{obs:y_properties}
For every time $t \in [0, T]$, it holds that
\begin{itemize}
	\item $\vzero \leq \vy \leq \frac{1}{2} \cdot \characteristic_\cN$, and
	\item $y(t) \in t \cdot \cP$.
\end{itemize}
\end{observation}
\begin{proof}
Consider a particular element $u \in \cN$. The rate in which $y_u(t)$ is increased is $x_u(t) \cdot (1 - 2y_u(t))$. Since $x_u(t)$ is non-negative, this rate is non-negative when $y_u(t) = 0$ and $0$ when $y_u(t) = 1/2$. Therefore, the value of $y_u(t)$ can never leave the range $[0, 1/2]$, which proves the first part of the observation. To see why the second part of the observation holds as well, note that
\[
	\vy(t)
	=
	\int_0^t \frac{d\vy(\tau)}{d\tau} \text{d}\tau
	=
	\int_0^t \vx(\tau) \hprod (\characteristic_\cN - 2 \cdot \vy(\tau)) \text{d}\tau
	\leq
	\int_0^t \vx(\tau) \text{d}\tau
	\enspace,
\]
which implies that $\vy(t) \in t \cdot \cP$ since $\cP$ is convex and down-closed.
\end{proof}

Next, we need to lower bound the rate in which $F(\vy(t))$ increases as a function of $t$.
\begin{lemma} \label{lem:increase_symmetric}
For every $t \in (0, T)$, $\frac{dF(\vy(t))}{dt} \geq f(\opt) - 2F(\vy)$.
\end{lemma}
\begin{proof}
By the chain rule,
\begin{align} \label{eq:chain_rule}
	\frac{dF(\vy(t))}{dt}
	={} &
	\Big\langle\frac{d\vy(t)}{dt}, \nabla F(\vy(t))\Big\rangle
	=
	\inner{\vx(t) \hprod (\characteristic_\cN - 2 \cdot \vy(t))}{\nabla F(\vy(t))}\\\nonumber
	={} &
	\inner{\vx(t)}{(\characteristic_\cN - 2 \cdot \vy(t)) \hprod \nabla F(\vy(t))}
	\geq
	\inner{\characteristic_{\opt}}{(\characteristic_\cN - 2 \cdot \vy(t)) \hprod \nabla F(\vy(t))}\\\nonumber
	={} &
	\inner{\characteristic_{\opt}}{(\characteristic_\cN - \vy(t)) \hprod \nabla F(\vy(t))} + \inner{\characteristic_{\opt}}{\vy(t) \hprod \nabla F(\characteristic_\cN - \vy(t))}
	\enspace,
\end{align}
where the inequality holds since $\opt$ is one candidate to be $\vx(t)$, and the last equality holds since the symmetry of $f$ implies that, for every element $u \in \cN$,
\begin{align*}
	\left.\frac{\partial F(\vz)}{z_u}\right|_{\vz = \characteristic_\cN - \vy(t)}
	={} &
	\bE[f(\RSet(\characteristic_\cN - \vy(t)) \vee \characteristic_{\{u\}}) - f(\RSet(\characteristic_\cN - \vy(t)) \wedge \characteristic_{\cN \setminus \{u\}})]\\
	={} &
	\bE[f(\RSet(\vy(t)) \wedge \characteristic_{\cN \setminus \{u\}}) - f(\RSet(\vy(t)) \vee \characteristic_{\{u\}})]
	=
	- \left.\frac{\partial F(\vz)}{z_u}\right|_{\vz = \vy(t)}
	\enspace.
\end{align*}

By Observation~\ref{obs:multilinear-partial}, Inequality~\eqref{eq:chain_rule} implies that
\begin{align*}
	\frac{dF(\vy(t))}{dt}
	={} &
	\sum_{u \in \opt} [F(\vy(t) \vee \characteristic_{\{u\}}) - F(\vy(t))] + \sum_{u \in \opt} [F((\characteristic_\cN - \vy(t)) \vee \characteristic_{\{u\}}) - F(\characteristic_\cN - \vy(t))]\\
	\geq{} &
	F(\vy(t) \vee \characteristic_{\opt}) - F(\vy(t)) + F((\characteristic_\cN - \vy(t)) \vee \characteristic_{\opt}) - F(\characteristic_\cN - \vy(t))\\
	={} &
	F(\vy(t) \vee \characteristic_{\opt}) + F((\characteristic_\cN - \vy(t)) \vee \characteristic_{\opt}) - 2F(\vy(t))
	\geq
	f(\opt) - 2F(\vy)
	\enspace,
\end{align*}
where the first inequality follows from the submodularity of $f$, the second equality follows by the symmetry of $f$, and finally, the second inequality holds since the submodularity and non-negativity of $f$ imply together that
\begin{align*}
	F(\vy(t) \vee \characteristic_{\opt}) + F((\characteristic_\cN - \vy(t)) \vee \characteristic_{\opt})
	={} &
	\bE_{S \sim \RSet(\vy(t))}[f(S \cup \opt) + f((\cN \setminus S) \cup \opt)]\\
	\geq{} &
	\bE_{S \sim \RSet(\vy(t))}[f(\opt)]
	=
	f(\opt)
	\enspace.
	\qedhere
\end{align*}
\end{proof}

\begin{corollary} \label{cor:super_measured_value}
The output vector $\vy(T)$ of Algorithm~\ref{alg:super_measured_continuous_greedy} obeys $F(\vy(T)) \geq \frac{1 - e^{-2T}}{2} \cdot f(\opt)$.
\end{corollary}
\begin{proof}
By the non-negativity of $f$, $F(\vy(0)) \geq 0$. Additionally, Lemma~\ref{lem:increase_symmetric} proves the differential inequality $\frac{dF(\vy(t))}{dt} \geq f(\opt) - 2F(\vy)$. By solving this differential inequality with the inequality $F(\vy(0)) \geq 0$ as the boundary condition, we get that for every $t \in [0, T]$ it holds that
\[
	F(\vy(t)) \geq \frac{1 - e^{-2t}}{2} \cdot f(\opt)
	\enspace.
\]
The corollary now follows by plugging $t = T$ into this inequality.
\end{proof}

Observation~\ref{obs:y_properties} and Corollary~\ref{cor:super_measured_value} prove together that Algorithm~\ref{alg:super_measured_continuous_greedy} has the properties guaranteed by Theorem~\ref{thm:measured_continuous_greedy_symmetric}.

\subsection{Improved Results for Knapsack and Cardinality Constraints} \label{ssc:improved_results_symmetric}

In this section, we prove the following theorem.

\begin{theorem} \label{thm:symmetric_knapsack_fractional}
For every constant $\eps \in (0, 1/2]$, there exists a polynomial time $(1/2 - \eps - o(1),\allowbreak \frac{1 - (2\eps)^c}{2c})$-bicriteria integer-approximation algorithm for the problem of maximizing the multilinear extension of a non-negative symmetric submodular function $f\colon 2^\cN \to \nnR$ subject to the polytope $\{\vx \in [0, 1]^\cN \mid \inner{\vp}{\vx} \leq B\}$ of a knapsack constraint with density $c$.
\end{theorem}

Before getting to the proof of Theorem~\ref{thm:symmetric_knapsack_fractional}, let us discuss its implications. By Lemma~\ref{lem:rounding_cardinality_knapsack}, Theorem~\ref{thm:symmetric_knapsack_fractional} immediately yields $(1/2 - \eps - o(1), \frac{1 - (2\eps)^c}{2c} + 1/\cardBudget)$-bicriteria approximation for maximizing a non-negative symmetric submodular function subject to a cardinality constraint.\footnote{The error term $1/\cardBudget$ should be viewed as an arbitrary small constant since the problem is polynomially solvable for a constant $\cardBudget$. Furthermore, the infeasibility ratio obtained is in fact $\cardBudget^{-1} \cdot \lceil \cardBudget \cdot \frac{1 - (2\eps)^c}{2c}\rceil$, which is sometimes better than $\frac{1 - (2\eps)^c}{2c} + 1/\cardBudget$.} It is particularly interesting to consider the single-criteria optimization that can be derived from Theorem~\ref{thm:symmetric_knapsack_fractional} for cardinality constraints.
\begin{corollary} \label{cor:symmetry_cardinality_single_criteria}
For every constant $c \in [0, 1/2)$, there exists a polynomial time $(\frac{1 - \sqrt[c]{1 - 2c}}{2} - o(1))$-approximation algorithm for the problem of maximizing the a non-negative symmetric submodular function subject to a cardinality constraint of density $c$.
\end{corollary}
\begin{proof}
The corollary follows from Theorem~\ref{thm:symmetric_knapsack_fractional} by plugging $\eps = \frac{\sqrt[c]{1 - 2c}}{2}$, and then rounding the resulting fractional solution using Lemma~\ref{lem:rounding_cardinality_knapsack}.
\end{proof}

The approximation ratio from the last corollary is better than the state-of-the-art for $c$ values that are smaller than roughly $0.282$ (the state-of-the-art is either an algorithm of~\cite{buchbinder2014submodular} or an algorithm of~\cite{feldman2017maximizing}, depending on the value of $c$). However, for larger values of $c$, the algorithm of~\cite{buchbinder2014submodular} outperforms Corollary~\ref{cor:symmetry_cardinality_single_criteria}.

For the more general problem of maximizing a non-negative symmetric submodular function subject to a knapsack constraint, combining Theorem~\ref{thm:knapsack_fractional} with Lemma~\ref{lem:rounding_cardinality_knapsack} yields the somewhat weaker guarantee of $(1/2 - \eps - o(1), \frac{1 - (2\eps)^c}{2c} + 1)$-bicriteria approximation. Notice that this bicriteria approximation guarantee is incomparable with the $(1/2 - \eps - \delta, \frac{1}{2} \ln (\eps^{-1}/2))$-bicriteria approximation obtained for this problem by Theorem~\ref{thm:symmetric_knapsack_greedy} below.

We also note that the technique of \cite{feldman2017maximizing} can be used to get the following variant of Theorem~\ref{thm:symmetric_knapsack_fractional} (see Appendix~\ref{app:equality} for more detail). Notice that the guarantee of Theorem~\ref{thm:symmetric_knapsack_equality} is not monotone in $c$, and therefore (as mentioned in Section~\ref{sec:preliminaries}), one cannot assume without loss of generality in the context of this theorem that no element has a cost exceeding the budget. This does not affect the proof of the theorem, but does affect the possibility of rounding because the proof of Lemma~\ref{lem:rounding_cardinality_knapsack} (for non-cardinality constraints) depends on this assumption.
\begin{restatable}{theorem}{thmSymmetricKnapsackEquality} \label{thm:symmetric_knapsack_equality}
For every constat $\eps \in (0, 1/2]$, there exists a polynomial time $(1/2 - \eps - o(1),\allowbreak \max\big\{1, \frac{1 - (2\eps)^{\min\{c, 1 - c\}}}{2\min\{c, 1 - c\}}\big\})$-bicriteria integer-approximation algorithm for the problem of maximizing the multilinear extension of a non-negative symmetric submodular function $f\colon 2^\cN \to \nnR$ subject to the polytope $\{\vx \in [0, 1]^\cN \mid \inner{\vp}{\vx} = B\}$ of a knapsack constraint with density $c$.
\end{restatable}

We now get to the promised proof of Theorem~\ref{thm:symmetric_knapsack_fractional}. This proof is based on Algorithm~\ref{alg:super_measured_continuous_greedy} from Section~\ref{ssc:symmetric_down_closed}. To get improved results from this algorithm for cardinality and knapsack constraints, we need a better upper bound on $\inner{\vp}{\vy(t)}$.

\begin{lemma} \label{lem:bound_with_density}
Assume Algorithm~\ref{alg:super_measured_continuous_greedy} is executed on a knapsack constraint defined by a cost vector $\vp$ and budget $B$. Then, for every time $t \in [0, T]$, it holds that $\inner{\vp}{\vy(t)} \leq B \cdot \frac{1 - e^{-2tc}}{2c}$.
\end{lemma}
\begin{proof}
Recall that $\frac{d\vy(t)}{dt} = \vx(t) \hprod (\characteristic_\cN - 2 \cdot \vy(t))$. Together with the boundary condition $\vy(0) = \vzero$, the solution for this differential equation is
\[
	\vy(t)
	=
	\frac{1}{2}(\characteristic_\cN - e^{-2\int_0^t \vx(\tau) \text{d}\tau})
	\enspace.
\]
Thus,
\begin{align*}
	2\inner{\vp}{\vy(t)}
	={} &
	|\vp|_1 - \sum_{u \in \cN} p_u \cdot e^{-2 \int_0^t x_u(\tau) \text{d}\tau}
	\leq
	|\vp|_1 - |\vp|_1 \cdot e^{-\frac{2}{|\vp|_1} \sum_{u \in \cN} p_u\int_0^t x_u(\tau) \text{d}\tau}\\
	={} &
	|\vp|_1 \big(1 - e^{-\frac{2}{|\vp|_1} \int_0^t \inner{\vp}{\vx(\tau)} \text{d}\tau}\big)
	\leq
	|\vp|_1 \big(1 - e^{-\frac{2tB}{|\vp|_1}}\big)
	=
	B \cdot \frac{1 - e^{-2tc}}{c}
	\enspace.
\end{align*}
where the first inequality holds since $e^{-x}$ is a convex function. The lemma now follows by rearranging this inequality.
\end{proof}

Using the bound from the last lemma, we can prove Theorem~\ref{thm:symmetric_knapsack_fractional} as follows.
\begin{proof}[Proof of Theorem~\ref{thm:symmetric_knapsack_fractional}]
We would like to show that the implementable version Algorithm~\ref{alg:super_measured_continuous_greedy} has all the properties stated in the theorem when its stopping time $T$ is set to $\frac{1}{2} \ln(\eps^{-1}/2)$. According to Corollary~\ref{cor:super_measured_value}, the output vector $\vy(T)$ of Algorithm~\ref{alg:super_measured_continuous_greedy} obeys
\[
	F(\vy(T))
	\geq
	\frac{1 - e^{-2T}}{2} \cdot f(\opt)
	=
	\frac{1 - e^{-\ln(\eps^{-1}/2)}}{2} \cdot f(\opt)
	=
	\frac{1 - 2\eps}{2} \cdot f(\opt)
	=
	(1/2 - \eps) \cdot f(\opt)
	\enspace.
\]
Thus, the integer-approximation ratio of Algorithm~\ref{alg:super_measured_continuous_greedy} is at least $1/2 - \eps$. Since the implementable version of Algorithm~\ref{alg:super_measured_continuous_greedy} has an integer-approximation ratio worse by $o(1)$, we get that its integer-approximation ratio is at least $1/2 - \eps - o(1)$. To complete the proof of the theorem, it remains to observe that the output vector $\vy(T)$ of Algorithm~\ref{alg:super_measured_continuous_greedy} also obeys, by Lemma~\ref{lem:bound_with_density},
\[
	\inner{\vp}{\vy(T)}
	\leq
	B \cdot \frac{1 - e^{-2Tc}}{2c}
	=
	B \cdot \frac{1 - e^{-c\ln(\eps^{-1}/2)}}{2c}
	=
	B \cdot \frac{1 - (2\eps)^{c}}{2c}
	\enspace,
\]
and therefore, the infeasibility ratio of Algorithm~\ref{alg:super_measured_continuous_greedy} is at most $\frac{1 - (2\eps)^{c}}{c}$.
\end{proof}

\subsection{Simple Greedy Algorithm} \label{ssc:simple_greedy_symmetric}

In this section, we present a simple greedy algorithm for bicriteria maximization of symmrtric submodular functions subject to cardinality and knapsack constraints. The algorithm appears as Algorithm~\ref{alg:greedy_density_symmetric}. One can observe that this algorithm is similar to the density-based greedy algorithm (Algorithm~\ref{alg:greedy_density}) from Section~\ref{ssc:simple_greedy_monotone}. However, there are some differences. The first difference between the algorithms is that Algorithm~\ref{alg:greedy_density_symmetric} removes from its solution set $S_i$ elements whose marginal contribution is negative enough (smaller than $-\frac{\delta m}{|\cN|}$). The second difference between the algorithms is that Algorithm~\ref{alg:greedy_density_symmetric} cannot always include in its solution all the elements of zero cost because they might be detrimental when $f$ is non-monotone. Thus, these elements require special treatment. The last difference between Algorithms~\ref{alg:greedy_density} and~\ref{alg:greedy_density_symmetric} is that Algorithm~\ref{alg:greedy_density_symmetric} might terminate before using all the available budget if no element has a positive marginal contribution with respect to the current solution.

\begin{algorithm}[ht] 
\caption{\textsc{Density-Based Greedy}$(f \colon 2^\cN \to \nnR, c, B, \eps, \delta)$}\label{alg:greedy_density_symmetric}
\DontPrintSemicolon
Let $m \gets \max\{f(\varnothing), \max_{u \in \cN} f(\{u\})\}$.\\
Let $S_0 \gets \varnothing$ and $i \gets 0$.\\
\While{$c(S_i) < \frac{B}{2} \cdot \ln (\eps^{-1}/2)$\label{line:greedy_density_loop_symmetric}}
{
	\While{there exists an element $u \in S_i$ for which $f(u \mid S_i - u) < -\frac{\delta m}{|\cN|}$\label{line:removal_loop}}
	{
		Remove from $S_i$ an arbitrary element $u$ obeying the condition of the loop on Line~\ref{line:removal_loop}.\label{line:element_removal}
	}
	Update $i \gets i + 1$.\\
	\If{there exists an element $u \in \cN \setminus S_{i - 1}$ such that $c(u) = 0$ and $f(u \mid S_{i - 1}) > 0$ \label{line:good_zero_value_element}}
	{
		Let $u_i$ be an element obeying the condition on Line~\ref{line:good_zero_value_element}.\\
		Let $S_i \gets S_{i - 1} + u_i$.
	}
	\ElseIf{there exists an element $u \in \cN \setminus S_{i - 1}$ such that $f(u \mid S_{i - 1}) > 0$\label{line:good_non-zero_value_element}}
	{
		Let $u_i$ be an element of $\{u \in \cN \setminus S_{i - 1} \mid c(u) > 0\}$ maximizing $\frac{f(u_i \mid S_{i - 1})}{c(u_i)}$.\label{line:select_u_by_density}\\
		Let $S_i \gets S_{i - 1} + u_i$.
	}
	\Else
	{
		\Return $S_{i - 1}$. \label{line:early_return}
	}
}
\Return $S_i$.\label{line:late_return}
\end{algorithm}

Intuitively, Algorithm~\ref{alg:greedy_density_symmetric} is based on the continuous algorithm of \cite{feldman2017maximizing}. The ideas underlying this algorithm were previously adapted to a discrete algorithm by \cite{wan2024efficient} in the context of (single-criteria) optimization subject to a cardinality constraint, but not in the context of a knapsack constraint. The properties of Algorithm~\ref{alg:greedy_density_symmetric} are summarized by the following theorem.
\begin{theorem} \label{thm:symmetric_knapsack_greedy}
For every $\eps \in (0, 1)$ and constant $\delta > 0$, Algorithm~\ref{alg:greedy_density_symmetric} is a polynomial time $(1 - \eps - \delta, 1 + \frac{1}{2}\ln(\eps^{-1}/2))$-bicriteria approximation algorithm for the problem of maximizing a non-negative symmetric submodular function subject to a knapsack constraint. In the special case of a cardinality constraint, the infeasibility ratio of Algorithm~\ref{alg:greedy_density_symmetric} improves to $\lceil \ln \frac{1}{2}\ln(\eps^{-1}/2) \rceil$.
\end{theorem}

We begin the analysis of Theorem~\ref{thm:symmetric_knapsack_greedy} with the following observation.
\begin{observation} \label{obs:m_bounds}
$|\cN|^{-1} \cdot \max_{S \subseteq \cN} f(S) \leq m \leq f(\opt)$.
\end{observation}
\begin{proof}
Since we assume that $c(u) \leq B$ for every $u \in \cN$, the value of $m$ is the maximum of the values of various feasible solutions, and thus, is upper bounded by $f(\opt)$. Additionally, by the submodularity and non-negativity of $f$, for every set $S \subseteq \cN$,
\[
	f(S)
	\leq
	f(\varnothing) + \sum_{u \in S} f(u \mid \varnothing)
	\leq
	\max\Big\{f(\varnothing), \sum_{u \in S} f(u)\Big\}
	\leq
	|\cN| \cdot m
	\enspace.
	\qedhere
\]
\end{proof}

\begin{corollary}
Algorithm~\ref{alg:greedy_density_symmetric} runs in polynomial time.
\end{corollary}
\begin{proof}
Let $\Delta$ be the total number of elements removed from the solution by Line~\ref{line:element_removal} of Algorithm~\ref{alg:greedy_density_symmetric}. Every removal of such an element increases the value of the solution set $S_i$ of Algorithm~\ref{alg:greedy_density_symmetric} by at least $\frac{m\delta}{|\cN|}$. Since the value of $f(S_i)$ never decreases during the execution of the algorithm, the final value of $f(S_i)$ must obey
\[
	f(S_i)
	\geq
	\Delta \cdot \frac{m\delta}{|\cN|}
	\geq
	\Delta \cdot \frac{(|\cN|^{-1} \cdot f(S_i))\delta}{|\cN|}
	=
	\frac{\Delta\delta \cdot f(S_i)}{|\cN|^2}
	\enspace,
\]
where the second inequality follows from Observation~\ref{obs:m_bounds}. Rearranging this inequality yields $\Delta \leq |\cN|^2/\delta$. Thus, throughout the execution of Algorithm~\ref{alg:greedy_density_symmetric}, at most $|\cN|^2/\delta$ elements are removed from its solution. In contrast, every iteration of the main loop of this algorithm (except maybe the last iteration) adds one element to this solution. Since the size of the solution of Algorithm~\ref{alg:greedy_density_symmetric} cannot exceed $|\cN|$, this implies that the main loop of Algorithm~\ref{alg:greedy_density_symmetric} makes at most $1 + |\cN| + |\cN|^2/\delta$ iterations, which is a polynomial number of iterations since $\delta$ is a constant.
\end{proof}

Next, we would like to analyze the approximation ratio of Algorithm~\ref{alg:greedy_density_symmetric}. The first step towards this goal is the next lemma. Let $\ell$ be the number of iterations performed by the main loop of Algorithm~\ref{alg:greedy_density_symmetric}.
\begin{lemma} \label{lem:opt_value_guarantee}
For every $i \in [\ell]$, $f(S_{i - 1} \cup \opt) \geq (1 - \delta) \cdot f(\opt) - f(S_{i - 1})$.
\end{lemma}
\begin{proof}
The loop on Line~\ref{line:removal_loop} of Algorithm~\ref{alg:greedy_density_symmetric} guarantees that $f(u \mid S_{i - 1} - u) \geq -\frac{\delta m}{|\cN|}$ for every $u \in S_i$. Therefore,
\begin{align*}
	f(S_i) \cup \opt)
	\geq{} &
	f(\opt) - f(\bar{S}_i \cup \opt)
	=
	f(\opt) - f(S_i \cap \overline{\opt})\\
	\geq{} &
	f(\opt) - f(S_i) + \sum_{u \in S_i \cap \opt} \mspace{-9mu} f(u \mid S_i - u)
	\geq
	f(\opt) - f(S_i) - |S_i \cap \opt| \cdot \frac{\delta m}{|\cN|}\\
	\geq{} &
	f(\opt) - f(S_i) - \delta m
	\geq
	(1 - \delta) \cdot f(\opt) - f(S_i)
	\enspace,
\end{align*}
where the first and second inequalities follow from the non-negativity and submodularity of $f$, the equality holds by the symmetry of $f$, and the last inequality follows from Observation~\ref{obs:m_bounds}.
\end{proof}

We can now prove the approximation ratio of Algorithm~\ref{alg:greedy_density_symmetric} in the case that it returns before exhausting the budget.
\begin{corollary}
If Algorithm~\ref{alg:greedy_density_symmetric} returns from Line~\ref{line:early_return}, then its approximation ratio is at least $1/2 - \eps - \delta$.
\end{corollary}
\begin{proof}
The fact that Algorithm~\ref{alg:greedy_density_symmetric} returns from Line~\ref{line:early_return} implies that $f(u \mid S_{i - 1}) \leq 0$ for every $u \in \cN \setminus S_{i - 1}$. Thus, its output set $S_{i - 1}$ obeys
\begin{align*}
	f(S_{i - 1})
	\geq{} &
	(1 - \delta) \cdot f(\opt) - f(S_{i - 1} \cup \opt)\\
	\geq{} &
	(1 - \delta) \cdot f(\opt) - f(S_{i - 1}) - \sum_{u \in \opt \setminus S_{i - 1}} \mspace{-18mu} f(u \mid S_{i - 1})
	\geq
	(1 - \delta) \cdot f(\opt) - f(S_{i - 1}) 
	\enspace,
\end{align*}
where the first inequality holds by Lemma~\ref{lem:opt_value_guarantee}, and the second inequality follows from the submodulairty of $f$. Rearranging this inequality now gives $f(S_{i - 1}) \geq \frac{1 - \delta}{2} \cdot f(\opt) \geq (1/2 - \delta) \cdot f(\opt)$, which proves the corollary.
\end{proof}

Proving the approximation ratio of Algorithm~\ref{alg:greedy_density_symmetric} when it returns from Line~\ref{line:late_return} is somewhat more involved.
\begin{lemma} \label{lem:symmetric_knapsack_approximation}
If Algorithm~\ref{alg:greedy_density_symmetric} returns from Line~\ref{line:late_return}, then its approximation ratio is at least $1/2 - \eps - \delta$.
\end{lemma}
\begin{proof}
Below, we prove by induction that for every integer $0 \leq i \leq \ell$, $f(S_i) \geq \frac{1 - e^{-2c(S_i) / B}}{2} \cdot (1 - \delta) \cdot f(\opt)$.  This will complete the proof of the lemma since the fact that Algorithm~\ref{alg:greedy_density_symmetric} has returned from Line~\ref{line:late_return} implies that $c(S_\ell) \geq \frac{B}{2} \cdot \ln (\eps^{-1}/2)$, and thus,
\begin{align*}
	f(S_\ell)
	\geq{} &
	\frac{1 - e^{-2c(S_i) / B}}{2} \cdot (1 - \delta) \cdot f(\opt)\\
	\geq{} &
	\frac{1 - e^{-\ln (\eps^{-1} / 2)}}{2} \cdot (1 - \delta) \cdot f(\opt)
	=
	\frac{1 - 2\eps}{2} \cdot (1 - \delta) \cdot f(\opt)
	\geq
	(1/2 - \eps - \delta) \cdot f(\opt)
	\enspace.
\end{align*}

For $i = 0$, the inequality $f(S_i) \geq \frac{1 - e^{-c(S_i) / (2B)}}{2} \cdot (1 - \delta) \cdot f(\opt)$ holds since the non-negativity of $f$ guarantees that
\[
	f(S_0)
	\geq
	0
	=
	\frac{1 - e^{-c(\varnothing) / (2B)}}{2} \cdot (1 - \delta) \cdot f(\opt)
	=
	\frac{1 - e^{-c(S_0) / (2B)}}{2} \cdot (1 - \delta) \cdot f(\opt)
	\enspace.
\]
Next, we need to prove the inequality $f(S_i) \geq \frac{1 - e^{-2c(S_i) / B}}{2} \cdot f(\opt)$ for $i \in [\ell]$ given the induction hypothesis that this inequality holds for $i - 1$. There are two cases to consider. If $c(u_i) = 0$, then $c(S_i) \leq c(S_{i - 1} + u_i) = c(S_{i - 1})$, and therefore,
\[
	f(S_i)
	\geq
	f(S_{i - 1})
	\geq
	\frac{1 - e^{-2c(S_{i - 1}) / B}}{2} \cdot (1 - \delta) \cdot f(\opt)
	\geq
	\frac{1 - e^{-2c(S_i) / B}}{2} \cdot (1 - \delta) \cdot f(\opt)
	\enspace.
\]
Otherwise, we know that the condition on Line~\ref{line:good_zero_value_element} evaluated to false in the iteration in which $u_i$ was selected, and thus, $f(u \mid S_{i - 1}) \leq 0$ for every $u \in \cN \setminus S_{i - 1}$ that obeys $c(u) = 0$. Since every element of $\{u \in \opt \setminus S_{i - 1} \mid c(u) > 0\}$ is a possible candidate to be chosen as $u_i$ on Line~\ref{line:select_u_by_density}, we have
\begin{align*}
	f(S_i)
	\geq{} &
	f(S_{i - 1}) + f(u_i \mid S_i)
	\geq
	f(S_{i - 1}) + c(u_i) \cdot \max_{u \in \opt \setminus S_{i - 1} \mid c(u) > 0} \frac{f(u \mid S_{i - 1})}{c(u)}\\
	\geq{} &
	f(S_{i - 1}) + c(u_i) \cdot \frac{\sum_{u \in \opt \setminus S_{i - 1} \mid c(u) > 0} f(u \mid S_{i - 1})}{\sum_{u \in \opt \setminus S_{i - 1} \mid c(u) > 0} c(u)}\\
	\geq{} &
	f(S_{i - 1}) + c(u_i) \cdot \frac{\sum_{u \in \opt \setminus S_{i - 1}} f(u \mid S_{i - 1})}{\sum_{u \in \opt \setminus S_{i - 1}} c(u)}
	\geq
	f(S_{i - 1}) + c(u_i) \cdot \frac{f(\opt \mid S_{i - 1})}{B}\\
	\geq{} &
	f(S_{i - 1}) + c(u_i) \cdot \frac{(1 - \delta) \cdot f(\opt) - 2 f(S_{i - 1})}{B}\\
	={} &
	\Big(1 - \frac{2c(u_i)}{B}\Big) \cdot f(S_{i - 1}) + \frac{c(u_i)}{B} \cdot (1 - \delta) \cdot f(\opt)
	\enspace,
\end{align*}
where the fifth inequality follows from $f$'s submodularity, and the last inequality holds because of Lemma~\ref{lem:opt_value_guarantee}. Using the induction hypothesis, the previous inequality implies
\begin{align*}
	f(S_i)
	\geq{} &
	\Big(1 - \frac{2c(u_i)}{B}\Big) \cdot \frac{1 - e^{-2c(S_{i - 1}) / B}}{2} \cdot (1 - \delta) \cdot f(\opt) + \frac{c(u_i)}{B} \cdot (1 - \delta) \cdot f(\opt)\\
	={} &
	\frac{1 - e^{-2c(S_{i - 1}) / B}\Big(1 - \frac{2c(u_i)}{B}\Big)\Big)}{2} \cdot (1 - \delta) \cdot f(\opt)\\
	\geq{} &
	(1 - e^{-c(S_{i - 1}) / B}\cdot e^{- 2c(u_i) / B}) \cdot (1 - \delta) \cdot f(\opt)\\
	={} &
	(1 - e^{-c(S_{i - 1} + u_i) / B}) \cdot (1 - \delta) \cdot f(\opt)
	\geq
	(1 - e^{-c(S_i) / B}) \cdot (1 - \delta) \cdot f(\opt)
	\enspace,
\end{align*}
where the last inequality holds since $S_i \subseteq S_{i - 1} + u_i$. This completes the proof by induction.
\end{proof}

To complete the proof of Theorem~\ref{thm:symmetric_knapsack_greedy}, it only remains to bound the infeasibility ratio of Algorithm~\ref{alg:greedy_density_symmetric}, which is done by the next lemma.

\begin{lemma}
The infeasibility ratio of Algorithm~\ref{alg:greedy_density_symmetric} is at most $1 + \frac{1}{2}\ln(\eps^{-1}/2)$. In the special case of a cardinality constraint, this infeasibility ratio improves to $\lceil \frac{1}{2}\ln(\eps^{-1}/2) \rceil$.
\end{lemma}
\begin{proof}
Consider first the case that Algorithm~\ref{alg:greedy_density_symmetric} returns from Line~\ref{line:early_return}. In this case, the fact that the algorithm has begun its final iteration implies that its output set $S_{\ell - 1}$ obeys $c(S_{\ell - 1}) < \frac{B}{2} \cdot \frac{1}{2}\ln(\eps^{-1}/2)$ (notice that this set can only lose elements after the beginning of the iteration). Thus, the infeasibility ratio of the algorithm in this case is at most $\frac{1}{2}\ln(\eps^{-1}/2)$.

Assume now that Algorithm~\ref{alg:greedy_density_symmetric} returns from Line~\ref{line:late_return}, and thus, returns $S_\ell$. Since Algorithm~\ref{alg:greedy_density} did not terminate after $\ell - 1$ iterations, we must have $c(S_{\ell - 1}) < \frac{B}{2} \cdot \ln(\eps^{-1}/2)$ (notice that $\ell \geq 1$ because $c(S_0) = c(\varnothing) = 0 < \frac{B}{2} \cdot \ln(\eps^{-1}/2)$). Thus, the output set $S_\ell$ of Algorithm~\ref{alg:greedy_density} obeys
\[
	c(S_\ell)
	\leq
	c(S_{\ell - 1}) + c(u_\ell)
	<
	\frac{B}{2} \cdot \ln(\eps^{-1}/2) + B
	=
	B \cdot \Big(1 + \frac{1}{2}\ln(\eps^{-1}/2)\Big)
	\enspace.
\]
In the special case of a cardinality constraint, $c(S_\ell)$ and $B$ are integers and $c(u_\ell) = 1$, and thus, the last inequality improves to
\[
	c(S_\ell)
	\leq
	c(S_{\ell - 1}) + c(u_\ell)
	\leq
	\bigg[\bigg\lceil \frac{B}{2} \cdot \ln(\eps^{-1}/2) \bigg\rceil - 1\bigg] + 1
	=
	\bigg\lceil \frac{B}{2} \cdot \ln(\eps^{-1}/2) \bigg\rceil
	\leq
	B \cdot \bigg\lceil \frac{1}{2} \ln(\eps^{-1}/2) \bigg\rceil
	\enspace.
	\qedhere
\]
\end{proof}

%% file: PipageRounding.tex
\section{Pipage Rounding}

Pipage Rounding is an algorithm suggested by \citet{calinescu2011maximizing} for rounding a fractional solution in a matroid polytope into an integral solution without reducing (in expectation) the value of a submodular objective function. In Section~\ref{app:pipage_rounding_matroid}, we show how Pipage Rounding can be used for matroid constraints also in the bicriteria approximation setting. In Section~\ref{app:pipage_rounding_knapsack}, we present and analyze a variant of Pipage Rounding for knapsack constraints.

\subsection{Pipage Rounding for Matroid Constraints} \label{app:pipage_rounding_matroid}

In this section, we prove Lemma~\ref{lem:rounding_matroid}, which we repeat here for convenience.

\lemRoundingMatroid*

A central component in the proof of Lemma~\ref{lem:rounding_matroid} is the matroid $\cM' = (\cN, \cI')$ defined by $\cI' = \{\cup_{i = 1}^{\lceil \beta \rceil} S_i \mid \forall_{i \in [\lceil \beta \rceil]}\; S_i \in \cI\}$. By the Matroid Union Theorem (Corollary 42.1a of~\cite{schijver2003combinatorial}), $\cM'$ is indeed a matroid.

\begin{lemma} \label{lem:matroid_union_polytope}
The vector $\vx$ belongs to the matroid polytope $\cP'$ of $\cM'$.
\end{lemma}

The proof of Lemma~\ref{lem:matroid_union_polytope} is technical, and thus, we defer it to the end of this section. Lemma~\ref{lem:matroid_union_polytope} implies that one can use the standard Pipage Rounding (of \cite{calinescu2011maximizing}) to get a random set $S$ such that $\bE[f(S)] \geq F(\vx)$ and $S \in \cI'$. The following observation shows that this set $S$ has all the properties guaranteed by Lemma~\ref{lem:rounding_matroid}.

\begin{observation}
It holds that $\characteristic_S / \lceil \beta \rceil \in \cP$.
\end{observation}
\begin{proof}
By the definition of $\cM'$, the fact that $S$ belongs to $\cI'$ implies that there are sets $S_1, S_2, \dotsc,\allowbreak S_{\lceil \beta \rceil} \in \cI$ such that $S = \cup_{i = 1}^{\lceil \beta \rceil} S_i$. Since removing elements from an independent set of a matroid always preserves its independence, we may assume without loss of generality that the sets $S_1, S_2, \dotsc, S_{\lceil \beta \rceil}$ are disjoint, and thus,
\[
	\characteristic_S = \sum_{i = 1}^{\lceil \beta \rceil} \characteristic_{S_i}
	\enspace.
\]
Note also that for every $i \in [\lceil \beta \rceil]$, the vector $\characteristic_{S_i}$ belongs to $\cP$ because $S_i \in \cI$. Combining these observations, we get
\[
	\frac{\characteristic_S}{\lceil \beta \rceil}
	=
	\frac{\sum_{i = 1}^{\lceil \beta \rceil} \characteristic_{S_i}}{\lceil \beta \rceil}
	=
	\sum_{i = 1}^{\lceil \beta \rceil} \tfrac{1}{\lceil \beta \rceil} \cdot \characteristic_{S_i}
	\in
	\cP
	\enspace,
\]
where the inclusion holds since $\sum_{i = 1}^{\lceil \beta \rceil} \tfrac{1}{\lceil \beta \rceil} \cdot \characteristic_{S_i}$ is a convex combination of vectors in $\cP$.
\end{proof}
 
To complete the proof of Lemma~\ref{lem:rounding_matroid}, it only remains to prove Lemma~\ref{lem:matroid_union_polytope}.
\begin{proof}[Proof of Lemma~\ref{lem:matroid_union_polytope}]
The rank function of a matroid $\cM = (\cN, \cI)$ is a function $\rank_\cM \colon 2^\cN \to \bZ_{\geq 0}$ defined as follows. For every $S \subseteq \cN$,
$
	\rank\nolimits_\cM(S)
	\triangleq
	\max\{|T| \mid T \subseteq S, T \in \cI\}
$.
It is well-known (see, for example, \cite{schijver2003combinatorial}) that the matroid polytope $\cP$ of $\cM$ can be stated using the rank function of $\cM$ as follows.
\[
	\cP
	=
	\{\vz \in [0, 1]^\cN \mid \forall_{S \subseteq \cN}\; \inner{\characteristic_S}{\vz} \leq \rank\nolimits_\cM(S)\}
	\enspace.
\]
Thus, the fact that $\vx/\lceil \beta \rceil \in \cP$ implies that, for every set $S \subseteq \cN$,
\begin{equation} \label{eq:matroid_inclusion}
	\frac{\inner{\characteristic_S}{\vx}}{\lceil \beta \rceil}
	=
	\bigg\langle\characteristic_S, \frac{\vx}{\lceil \beta \rceil} \bigg\rangle
	\leq
	\rank\nolimits_\cM(S)
	\enspace.
\end{equation}
Analogously, to prove that $\vx \in \cP'$, we need to prove that $\inner{\characteristic_S}{\vx} \leq \rank_{\cM'}(S)$ for every set $S \subseteq \cN$. Assume towards a contradiction that this is not the case, i.e., that there exists a set $S \subseteq \cN$ such that $\inner{\characteristic_S}{\vx} > \rank_{\cM'}(S)$.

At this point we would like to use the last inequality and Inequality~\eqref{eq:matroid_inclusion} to derive a contradiction. However, these two inequalities use rank of functions of two different matroids. Thus, to combine the two inequalities, we need a way to express the rank function of one of these matroids in terms of the other. Such a way is given by the Matroid Union Theorem (Corollary 42.1a of~\cite{schijver2003combinatorial}), which proves that
\[
	\rank_{\cM'}(S)
	=
	\min_{T \subseteq S} \{|S \setminus T| + \lceil \beta \rceil \cdot \rank_\cM(T)\}
	\enspace.
\]
Thus, the inequality that we assumed implies
\begin{align*}
	\inner{\characteristic_S}{\vx}
	>{} &
	\rank_{\cM'}(S)
	=
	\min_{T \subseteq S} \{|S \setminus T| + \lceil \beta \rceil \cdot \rank_\cM(T)\}\\
	\geq{} &
	\min_{T \subseteq S} \bigg\{|S \setminus T| + \lceil \beta \rceil \cdot \frac{\inner{\characteristic_T}{\vx}}{\lceil \beta \rceil}\bigg\}
	=
	\min_{T \subseteq S} \{|S \setminus T| + \inner{\characteristic_T}{\vx}\}
	\enspace,
\end{align*}
where the second inequality follows from Inequality~\eqref{eq:matroid_inclusion}. Consider now the subset $T$ for which the minimum is obtained on the rightmost side of the last inequality. Then, we have $\inner{\characteristic_S}{\vx} > |S \setminus T| + \inner{\characteristic_T}{\vx}$, and by rearranging this inequality, we get $\inner{\characteristic_{S \setminus T}}{\vx} > |S \setminus T|$. However, this inequality leads to the promised contradiction because it is always true that
\[
	\inner{\characteristic_{S \setminus T}}{\vx}
	\leq
	\|\characteristic_{S \setminus T}\|_1
	=
	|S \setminus T|
	\enspace.
	\qedhere
\]
\end{proof}

\subsection{Pipage Rounding for Knapsack Constraints} \label{app:pipage_rounding_knapsack}

In this section, we present and analyze a version of Pipage Rounding for knapsack constraints, and show that it obeys the properties guaranteed by Lemma~\ref{lem:rounding_cardinality_knapsack}. Our version of Pipage Rounding is formally given as Algorithm~\ref{alg:pipage_rounding_knapsack}. It gets a vector $\vx \in [0, 1]^\cN$ to round and a vector $\vp$ specifying the costs of a knapsack constraint whose budget is $B$.

\SetKwIF{With}{OtherwiseWith}{Otherwise}{with}{do}{otherwise with}{otherwise}{endwith}
\begin{algorithm}
\DontPrintSemicolon
\caption{\textsc{Pipage Round for Knapsack Constraints}$(\vx, \vp)$} \label{alg:pipage_rounding_knapsack}
\While{$\vx$ contains a fractional entry corresponding to an element $u \in \cN$ with $p_u = 0$\label{line:loop_zero_cost}}
{
	\lWith{probability $x_u$}{Update $x_u \gets 1.$}
	\lOtherwise{Update $x_u \gets 0.$}
}
\While{$\vx$ contains at least two fractional entries\label{line:loop_pipage}}
{
	Let $u$ and $v$ be two elements of $\cN$ corresponding to fractional entries of $\vx$.\\
	Let $s_u \gets \min\{1 - x_u, p_v x_v / p_u\}$ and $s_v \gets \min\{1 - x_v, p_u x_u / p_v\}$.\\
	\lWith{probability $\frac{p_vs_v}{p_vs_v + p_us_u}$}{Increase $x_u$ by $s_u$ and decrease $x_v$ by $p_u s_u / p_v$.\label{line:increase_u}}
	\lOtherwise{Increase $x_v$ by $s_v$ and decrease $x_u$ by $p_v s_v / p_u$.\label{line:increase_v}}
}
\If{$\vx$ contains a fractional entry\label{line:single_rounding}}
{
	Let $u$ be the element corresponding to the fractional entry of $\vx$.\\
	\lWith{probability $x_u$}{Update $x_u \gets 1.$\label{line:single_rounding_up}}
	\lOtherwise{Update $x_u \gets 0.$}
}
\Return the set $\{u \in \cN \mid x_u = 1\}$.
\end{algorithm}

We begin the analysis of Algorithm~\ref{alg:pipage_rounding_knapsack} with the following observation.
\begin{observation}
Each iteration of the loops starting on Lines~\ref{line:loop_zero_cost} and~\ref{line:loop_pipage} of Algorithm~\ref{alg:pipage_rounding_knapsack} keeps $\vx$ as a vector in $[0, 1]^\cN$, but makes one more entry of $\vx$ integral. Thus, Algorithm~\ref{alg:pipage_rounding_knapsack} is guaranteed to run in polynomial time, and when it terminates the vector $\vx$ is integral. 
\end{observation}
\begin{proof}
The observation is immediate for the loop starting on Line~\ref{line:loop_zero_cost}. Thus, we concentrate on an iteration of the loop starting on Line~\ref{line:loop_pipage} of Algorithm~\ref{alg:pipage_rounding_knapsack}. Assume without loss of generality that Line~\ref{line:increase_u} is executed on this iteration rather than Line~\ref{line:increase_v} (the other case is analogous).

We now need to consider two cases. The first case is that $s_u$ is set to $1 - x_u$. In this case, $x_u$ is increased to be $x_u + s_u = 1$ (i.e., integral), while $x_v$ is decreased to be $x_v - p_u s_u / p_v \geq x_v - p_u (p_v x_v / p_u) / p_v = 0$. The other case is that $x_u$ is set to $p_v x_v / p_u$. In this case, $x_u$ is increased to be $x_u + s_u \leq x_u + (1 - x_u) \leq 1$, while $x_v$ is decreased to be $x_v - p_u s_u / p_v = x_v - p_u (p_v x_v / p_u) / p_v = 0$ (i.e., integral).
\end{proof}

Next, we study the cost of the output set of Algorithm~\ref{alg:pipage_rounding_knapsack}.
\begin{lemma} \label{lem:cost_pipage}
Let $S$ be the output set of Algorithm~\ref{alg:pipage_rounding_knapsack}. Then, $\sum_{u \in S} p_u \leq \inner{\vp}{\vx} + B$. Furthermore, if the knapsack constraint is a cardinality constraint, then the last inequality improves to $|S| = \sum_{u \in S} p_u \leq \lceil \inner{\vp}{\vx}\rceil$.
\end{lemma}
\begin{proof}
The loop starting on Line~\ref{line:loop_zero_cost} of Algorithm~\ref{alg:pipage_rounding_knapsack} does not affect $\inner{\vp}{\vx}$ because it only modifies coordinates of $\vx$ corresponding to elements of zero cost. Next, we would like to show that the loop starting on Line~\ref{line:loop_pipage} of Algorithm~\ref{alg:pipage_rounding_knapsack} does not affect $\inner{\vp}{\vx}$ either. Consider a particular iteration of this loop, and let us show that it does not affect $\inner{\vp}{\vx}$. If this iteration executes Line~\ref{line:increase_u}, then the change in $\inner{\vp}{\vx}$ during the iteration is $p_u s_u - p_v (p_u s_u / p_v) = 0$. Similarly, if this iteration executes Line~\ref{line:increase_v}, then the change in $\inner{\vp}{\vx}$ during the iteration is $p_v s_v - p_u (p_v s_v / p_u) = 0$. 

Given the above, the value of $\inner{\vp}{\vx}$ when Algorithm~\ref{alg:pipage_rounding_knapsack} reaches Line~\ref{line:single_rounding} is equal to its original value. The only line that can increase $\inner{\vp}{\vx}$ after this point is Line~\ref{line:single_rounding_up}, which (if executed) increases $\inner{\vp}{\vx}$ by $p_u \leq B$ for some element $u$ (the inequality holds by our assumption that $\|\vp\|_\infty \leq B$). Thus, Algorithm~\ref{alg:pipage_rounding_knapsack} increase $\inner{\vp}{\vx}$ by at most $B$ during its excution. This implies the inequality $\sum_{u \in S} p_u \leq \inner{\vp}{\vx} + B$ since the definition of $S$ and the fact that $\vx$ is integral when Algorithm~\ref{alg:pipage_rounding_knapsack} terminates guarantee together that $\sum_{u \in S} p_u$ is equal to the final value of $\inner{\vp}{\vx}$.

In the special case of a cardinality constraint we know that the increase in Line~\ref{line:single_rounding_up} is by less than $1$, which by the above arguments implies that $|S| = \sum_{u \in S} p_u < \inner{\vp}{\vx} + 1$. Since the left hand side of this inequality is integral, it must also hold that $|S| = \sum_{u \in S} p_u \leq \lceil \inner{\vp}{\vx} \rceil$, which completes the proof of the lemma.
\end{proof}

It remains to analyze the effect of Algorithm~\ref{alg:pipage_rounding_knapsack} on the submodular objective function.
\begin{lemma} \label{lem:pipage_objective}
Let $S$ be the output set of Algorithm~\ref{alg:pipage_rounding_knapsack}. For every submodular function $f\colon 2^\cN \to \bR$ and its multilinear extension, it holds that $\bE[f(S)] \geq F(\vx)$. 
\end{lemma}
\begin{proof}
Observe that $f(S)$ is equal to $F(\vx)$ when $\vx$ is the value of this vector at the end of Algorithm~\ref{alg:pipage_rounding_knapsack}. Thus, to prove the lemma it suffices to show that $\bE[f(\vx)]$ does not decrease as the algorithm progresses. We begin by explaining why an iteration of the loop starting on Line~\ref{line:loop_zero_cost} does not change $\bE[f(\vx)]$. Let $\vx$ and $\vx'$ denote the values of $\vx$ before and after the iteration, respectively. Then, $\vx'$ is equal to $\vx \vee \characteristic_{\{u\}}$ with probability $x_u$, and is equal to $\vx \wedge \characteristic_{\cN - u}$ otherwise. By the law of total probability, this implies
\begin{align*}
	\bE[F(\vx')]
	={} &
	x_u \cdot F(\vx \vee \characteristic_{\{u\}}) + (1 - x_u) \cdot F(\vx \wedge \characteristic_{\cN - u})\\
	={} &
	x_u \cdot \bE[f(\RSet(\vx) + u)] + (1 - x_u) \cdot \bE[f(\RSet(\vx) - u)]
	=
	\bE[f(\RSet(\vx))]
	=
	F(\vx)
	\enspace.
\end{align*}
A similar argument shows that the body of the ``if'' statement on Line~\ref{line:single_rounding} of Algorithm~\ref{alg:pipage_rounding_knapsack} does not affect $\bE[f(\vx)]$ as well.

To complete the proof of the lemma, it remains to show that an iteration of the loop starting on Line~\ref{line:loop_pipage} of Algorithm~\ref{alg:pipage_rounding_knapsack} does not affect $\bE[f(\vx)]$ either. Let $\vx$ and $\vx'$ denote the values of $\vx$ before and after the iteration, respectively. Then, $\vx'$ is equal to $\vx + s_u \cdot \characteristic_{\{u\}} - \frac{p_u s_u}{p_v} \cdot \characteristic_{\{v\}}$ with probability $\frac{p_vs_v}{p_vs_v + p_us_u}$ and is equal to $\vx + s_v \cdot \characteristic_{\{v\}} - \frac{p_v s_v}{p_u} \cdot \characteristic_{\{v\}}$ with probability $1 - \frac{p_vs_v}{p_vs_v + p_us_u} = \frac{p_us_u}{p_vs_v + p_us_u}$. By the submodularity and multilinearity of $F$, we have
\begin{multline*}
	F\Big(\vx + s_u \cdot \characteristic_{\{u\}} - \frac{p_u s_u}{p_v} \cdot \characteristic_{\{v\}}\Big) - F(\vx)
	=
	s_u \cdot \frac{\partial F(\vx)}{\partial x_u} - \frac{p_u s_u}{p_v} \cdot \frac{\partial F(\vx + s_u \cdot \characteristic_{\{u\}})}{\partial x_v}\\
	\geq
	s_u \cdot \frac{\partial F(\vx)}{\partial x_u} - \frac{p_u s_u}{p_v} \cdot \frac{\partial F(\vx)}{\partial x_v}
	\enspace,
\end{multline*}
and
\begin{multline*}
	F\Big(\vx + s_v \cdot \characteristic_{\{v\}} - \frac{p_v s_v}{p_u} \cdot \characteristic_{\{u\}}\Big) - F(\vx)
	=
	s_v \cdot \frac{\partial F(\vx)}{\partial x_v} - \frac{p_v s_v}{p_u} \cdot \frac{\partial F(\vx + s_v \cdot \characteristic_{\{v\}})}{\partial v_u}\\
	\geq
	s_v \cdot \frac{\partial F(\vx)}{\partial x_v} - \frac{p_v s_v}{p_u} \cdot \frac{\partial F(\vx)}{\partial x_u}
	\enspace,
\end{multline*}
Thus, by the law of total probability,
\begin{align*}
	\bE[F(\vx')] - F(\vx)
	\geq{} &
	\frac{p_vs_v}{p_vs_v + p_us_u} \cdot \Big[s_u \cdot \frac{\partial F(\vx)}{\partial x_u} - \frac{p_u s_u}{p_v} \cdot \frac{\partial F(\vx)}{\partial x_v}\Big] \\&\mspace{100mu}+ \frac{p_us_u}{p_vs_v + p_us_u} \cdot \Big[s_v \cdot \frac{\partial F(\vx)}{\partial x_v} - \frac{p_v s_v}{p_u} \cdot \frac{\partial F(\vx)}{\partial x_u}\Big]
	=
	0
	\enspace.
	\qedhere
\end{align*}
\end{proof}

Lemma~\ref{lem:rounding_cardinality_knapsack} now follows by combining Lemmata~\ref{lem:cost_pipage} and~\ref{lem:pipage_objective}.

%% file: ImprovedAnalysisUnified.tex
\section{More Details for the Proof of Proposition~\texorpdfstring{\ref{prop:measured_continuous_greedy_density}}{\ref*{prop:measured_continuous_greedy_density}}} \label{app:improved_analysis_unified}

The proof of Proposition~\ref{prop:measured_continuous_greedy_density} refers to a slightly improved version of Lemma~III.19 of \cite{feldman2011unified}. In this section, we prove this improved version. Throughout the section, we use the notation of \cite{feldman2011unified}.

The original proof of Lemma~III.19 of \cite{feldman2011unified} was based on Lemma~III.7 of the same paper, which proved that for every time $0 \leq t \leq T$ that is an integer multiple of $\delta$ and element $e \in \cN$
\[
	y_e(t) \leq 1 - e^{-I_e^t} + O(\delta) \cdot t
	\enspace.
\]
The proof of Lemma~III.19 then argues that due to this inequality, it holds that
\begin{equation} \label{eq:original}
	\sum_{e \in \cN'} a_e \cdot y_e(T)
	\leq
	\sum_{e \in \cN'} a_e \cdot (1 - e^{-I^T_e} + O(\delta) \cdot T)
	\leq
	\sum_{e \in \cN'} a_e(1 - e^{-I_e^T}) + O(\delta)Tb/d(P)
	\enspace,
\end{equation}
where the last term $O(\delta)Tb/d(P)$ ends up appearing in the guarantee of the lemma. To improve over this, we prove below the following lemma.
\begin{lemma} \label{lem:improved_inequality}
For every time $0 \leq t \leq T$ that is an integer multiple of $\delta$ and element $e \in \cN$
\[
	y_e(t) \leq 1 - e^{-I_e^t} + O(\delta) \cdot I_e^t
	\enspace.
\]
\end{lemma}

Using this lemma instead of Lemma~III.7 of \cite{feldman2011unified}, we get
\begin{align*}
	\sum_{e \in \cN'} a_e \cdot y_e(T)
	\leq{} &
	\sum_{e \in \cN'} a_e \cdot (1 - e^{-I^T_e} + O(\delta) \cdot TI^T_e)\\
	\leq{} &
	\sum_{e \in \cN'} a_e(1 - e^{-I_e^T}) + O(\delta) \cdot \sum_{e \in \cN'} a_e I_e^T
	\leq
	\sum_{e \in \cN'} a_e(1 - e^{-I_e^T}) + O(\delta) \cdot T b
	\enspace.
\end{align*}
Thus, the last term improves by a factor of $d(P)$ compared to Inequality~\eqref{eq:original}. This factor $d(P)$ is the density of the constraint, or $c$ in our notation, and thus, this is exactly the improvement we have assumed in the proof of Proposition~\ref{prop:measured_continuous_greedy_density}.

It remains to prove Lemma~\ref{lem:improved_inequality}.
\begin{proof}[Proof of Lemma~\ref{lem:improved_inequality}]
We prove by induction on $t$ that $y_e(t) \leq 1 - e^{-I_e^t} + 0.5\delta I_e^t$. For $t = 0$, this is the case since $y_e(0) = 0 = 1 - e^{-0} + 0.5\delta \cdot 0$. Assume now that the claim holds for some $t$, and let us prove it for $t + \delta$.
\begin{align*}
	y_e(t + \delta)
	={} &
	y_e(t) + \delta I_e(t) \cdot (1 - y_e(t))
	\leq
	(1 - e^{-I^t_e} + 0.5\delta I^t_e)(1 - \delta I_e(t)) + \delta I_e(t)\\
	\leq{} &
	1 - e^{-I^t_e} \cdot (1 - \delta I_e(t)) + 0.5\delta I^t_e
	\leq
	1 - e^{-I^t_e} \cdot [e^{-\delta I_e(t)} - 0.5(\delta I_e(t))^2] + 0.5\delta I_e^t\\
	\leq{} &
	1 - e^{-I^t_e - \delta I_e(t)} + 0.5\delta^2 I_e(t) + 0.5\delta I_e^t
	=
	1 - e^{-I_e^{t+\delta}} + 0.5\delta I_e^{t + \delta}
	\enspace,
\end{align*}
where the first inequality follows from the induction hypothesis.
\end{proof}

%% file: KnapsackEquality.tex
\section{Proof of Theorem~\texorpdfstring{\ref{thm:symmetric_knapsack_equality}}{\ref*{thm:symmetric_knapsack_equality}}} \label{app:equality}

In this section, we prove Theorem~\ref{thm:symmetric_knapsack_equality}, which we repeat here for convenience.

\thmSymmetricKnapsackEquality*

By Reduction~2 of \cite{feldman2017maximizing}, we may assume in the proof of Theorem~\ref{thm:symmetric_knapsack_equality} that $c \leq 1/2$ (and thus, $\min\{c, 1 - c\} = c$). Using the algorithm of Theorem~\ref{thm:symmetric_knapsack_fractional}, we can now get a vector $\vy \in [0, 1]^\cN$ such that
\begin{itemize}
	\item $\inner{\vp}{\vy} \leq B \cdot \frac{1 - (2\eps)^c}{2c}$,
	\item $F(\vy) \geq (\frac{1}{2} - \eps - o(1)) \cdot f(\opt)$ with probability $1 - o(1)$, and
	\item $\|\vy\|_\infty \leq 1/2$.
\end{itemize}

The last bullet is not stated in Theorem~\ref{thm:symmetric_knapsack_equality} itself, but follows from Observation~\ref{obs:y_properties}. The second bullet is stronger than the inequality $\bE[F(\vy)] \geq (\frac{1}{2} - \eps - o(1)) \cdot f(\opt)$ guaranteed by Theorem~\ref{thm:symmetric_knapsack_equality}. To see why it is true nevertheless, we recall where the expectation in the last inequality comes from. Corollary~\ref{cor:super_measured_value} gives a deterministic lower bound on the value of the output set of Algorithm~\ref{alg:super_measured_continuous_greedy}. Since we need to use an efficient variant of Algorithm~\ref{alg:super_measured_continuous_greedy}, this deterministic lower bounds holds only with high probability (and up to an error of $o(1) \cdot f(\opt)$). In Theorem~\ref{thm:symmetric_knapsack_fractional}, we have avoided the need to explicitly state that the inequality holds only with high probability by adding the expectation. In contrast, the above second bullet explicitly states the guarantee obtained.

We also need the inequality $F(\frac{1}{2} \cdot \characteristic_\cN) \geq \frac{1}{2}f(OPT)$, which follows from Theorem~2.1 of \cite{feige2011maximizing}. We now consider the vector
\[
	\vz
	=
	\vy + \max\bigg\{\frac{B - \inner{\vp}{\vy}}{B / (2c) - \inner{\vp}{\vy}}, 0\bigg\} \cdot \bigg(\frac{1}{2}\characteristic_\cN - \vy\bigg)
	\enspace.
\]
In the following, we show that $\vz$ obeys all the properties required from the output vector of the algorithm whose existence is guaranteed by Theorem~\ref{thm:symmetric_knapsack_equality}, and thus, the above described method to compute $\vz$ proves this theorem. We begin with the following observation.
\begin{observation}
It holds that $\vz \in [0, 1]^\cN$ and $B \leq \inner{\vp}{\vz} \leq B \cdot \max\big\{1, \frac{1 - (2\eps)^c}{2c}\big\}$.
\end{observation}
\begin{proof}
Clearly, $\vz \geq \vy \geq \vzero$. Additionally, we notice that $\inner{\vp}{\vy} \leq B \cdot \frac{1 - (2\eps)^c}{2c} < \frac{B}{2c}$ and $\frac{1}{2}\characteristic_\cN - \vy \geq \vzero$ because $\|\vy\|_\infty \leq 1/2$. These observations together imply that
\[
	\vz
	\leq
	\vy + \bigg(\frac{1}{2} \cdot \characteristic_\cN - \vy\bigg)
	=
	\frac{1}{2} \cdot \characteristic_\cN
	\leq
	\characteristic_\cN
	\enspace,
\]
and thus, $\vz \in [0, 1]^\cN$.

Notice now that if $\inner{\vp}{\vy} \geq B$, then $\vz = \vy$, and therefore, $\inner{\vp}{\vz} = \inner{\vp}{\vy} \in [B, B \cdot \frac{1 - (2\eps)^c}{2c}]$. In contrast, if $\inner{\vp}{\vy} < B$, then we have
\begin{align*}
	\inner{\vp}{\vz}
	={} &
	\inner{\vp}{\vy} + \bigg\langle \vp, \frac{B - \inner{\vp}{\vy}}{B / (2c) - \inner{\vp}{\vy}} \cdot \bigg(\frac{1}{2}\characteristic_\cN - \vy\bigg)\bigg\rangle\\
	={} &
	\inner{\vp}{\vy} + \frac{B - \inner{\vp}{\vy}}{B / (2c) - \inner{\vp}{\vy}} \cdot \bigg(\frac{1}{2}\|\vp\|_1 - \inner{\vp}{\vy}\bigg)\\
	={} &
	\inner{\vp}{\vy} + \frac{B - \inner{\vp}{\vy}}{B / (2c) - \inner{\vp}{\vy}} \cdot \bigg(\frac{B}{2c} - \inner{\vp}{\vy}\bigg)
	=
	B
	\enspace.
	\qedhere
\end{align*}
\end{proof}

Next, we need to show that $F(\vz)$ is large with high probability.
\begin{lemma}
With high probability, $F(\vz) \geq (\frac{1}{2} - \eps - o(1)) \cdot f(\opt)$.
\end{lemma}
\begin{proof}
We prove the lemma by showing that $F(\vz) \geq (\frac{1}{2} - \eps - o(1)) \cdot f(\opt)$ whenever $F(\vy) \geq (\frac{1}{2} - \eps - o(1)) \cdot f(\opt)$. If $\vz = \vy$, then this is obvious. Otherwise,
\[
	\vz
	=
	g\bigg(\frac{B - \inner{\vp}{\vy}}{B / (2c) - \inner{\vp}{\vy}}\bigg)
	\enspace,
\]
where
\[
	g(x)
	\triangleq
	F\bigg(\vy + x \cdot \bigg(\frac{1}{2} \cdot \characteristic_{\cN} - \vy\bigg)\bigg)
	\enspace.
\]
Since $F$ is concave along positive directions (by an observation of \cite{calinescu2011maximizing}), and $\frac{1}{2} \cdot \characteristic_{\cN} - \vy \geq 0$, the function $g$ is concave. Thus,
\begin{align*}
	F(\vz)
	={} &
	g\bigg(\frac{B - \inner{\vp}{\vy}}{B / (2c) - \inner{\vp}{\vy}}\bigg)
	\geq
	\min\{g(0), g(1)\}\\
	={} &
	\min\bigg\{F(\vy), F\bigg(\frac{1}{2} \cdot \characteristic_{\cN}\bigg)\bigg\}
	\geq
	\min\bigg\{F(\vy), \frac{1}{2}f(\opt)\bigg\}
	\enspace.
\end{align*}
The last inequality implies that $F(\vz) \geq (\frac{1}{2} - \eps - o(1)) \cdot f(\opt)$ whenever $F(\vy) \geq (\frac{1}{2} - \eps - o(1)) \cdot f(\opt)$, which is what we have wanted to prove.
\end{proof}